\documentclass[11pt,fleqn]{article}

\usepackage{amsmath,amssymb,amsthm,enumerate}

\newcommand{\rsemioplus}{\mathbin{\mbox{$\lefteqn{\hspace{.67ex}\rule{.4pt}{1.2ex}}{\ni}$}}}

\newcommand{\p}{\partial}
\newcommand{\const}{\mathop{\rm const}\nolimits}
\newcommand{\sign}{\mathop{\rm sign}\nolimits}
\newcommand{\CV}{\mathop{\rm CV}\nolimits}
\newcommand{\CL}{\mathop{\rm CL}\nolimits}
\newcommand{\Ch}{\mathop{\rm Ch}\nolimits}
\newcommand{\Eop}{\mathop{\sf E}\nolimits}
\newcommand{\Fder}{\mathop{\sf D}\nolimits}
\newcommand{\diag}{\mathop{\rm diag}\nolimits}

\newtheorem{theorem}{Theorem}
\newtheorem{lemma}{Lemma}
\newtheorem{corollary}{Corollary}
\newtheorem{proposition}{Proposition}
\newtheorem*{proposition*}{Proposition}
{\theoremstyle{definition}
\newtheorem{definition}{Definition}

\newtheorem{note}{Note}
\newtheorem{problem}{Problem}
}

\begin{document}

\par\noindent {\LARGE\bf
Conservation Laws and Potential Symmetries \\
of Linear Parabolic Equations
\par}

{\vspace{4mm}\par\noindent \it 
Roman O. POPOVYCH~$^\dag$, Michael KUNZINGER~$^\ddag$ and Nataliya M. IVANOVA~$^\S$
\par\vspace{2mm}\par}

{\vspace{2mm}\par\noindent \it
$^{\dag,\S}$~Institute of Mathematics of NAS of Ukraine, 3 Tereshchenkivska Str., Kyiv-4, Ukraine
 \par}

{\vspace{2mm}\par\it 
\noindent $^{\dag,\ddag}$Fakult\"at f\"ur Mathematik, Universit\"at Wien, Nordbergstra{\ss}e 15, A-1090 Wien, Austria
\par}

{\vspace{2mm}\par\noindent $\phantom{^{\dag,\S}}$\rm E-mail: \it 
$^\dag$rop@imath.kiev.ua, $^\ddag$michael.kunzinger@univie.ac.at, $^\S$ivanova@imath.kiev.ua
 \par}

{\vspace{9mm}\par\noindent\hspace*{8mm}\parbox{140mm}{\small 
We carry out an extensive investigation of conservation laws and potential symmetries 
for the class of linear $(1+1)$-dimensional second-order parabolic equations. 
The group classification of this class is revised by employing admissible transformations, 
the notion of normalized classes of differential equations and the adjoint variational principle. 
All possible potential conservation laws are described completely. They are in fact
exhausted by local conservation laws.
For any equation from the above class the characteristic space of local conservation laws
is isomorphic to the solution set of the adjoint equation. 
Effective criteria for the existence of potential symmetries are proposed. 
Their proofs involve a rather intricate interplay between different representations of potential systems, 
the notion of a potential equation associated with a tuple of characteristics, 
prolongation of the equivalence group to the whole potential frame and 
application of multiple dual Darboux transformations.
Based on the tools developed, a preliminary analysis of generalized potential symmetries
is carried out and then applied to substantiate our construction of potential systems. 
The simplest potential symmetries of the linear heat equation, which are associated with single conservation laws, 
are classified with respect to its point symmetry group. 
Equations possessing infinite series of potential symmetry algebras are studied in detail.
}\par\vspace{7mm}}

\section{Introduction}

In the present paper we classify local and potential conservation laws and potential symmetries of
linear $(1+1)$-dimensional homogeneous second-order parabolic equations of the general form
\begin{equation}\label{EqGenLPE}
u_t=A(t,x)u_{xx}+B(t,x)u_x+C(t,x)u,
\end{equation}
where $A=A(t,x)$, $B=B(t,x)$ and $C=C(t,x)$ are arbitrary smooth functions, $A\ne0$.

This class contains a number of physically important subclasses that are widely investigated 
and may be applied in many situations. Probably the most famous examples are 
the Kolmogorov equations ($C=0$) and adjoint to them the Fokker--Planck equations ($A_{xx}-B_x+C=0$)
which are often considered as main equations of continuous Markov processes. 
Fokker--Planck equations appeared first in~\cite{Fokker1914} where the Brownian motion in 
the radiation field was studied,
and in~\cite{Planck1917} where one of the first systematic attempts of constructing
a complete theory of fluctuations was made.
They are also derived from the Boltzmann equation in the limit of large impact parameters~\cite{Liboff1969}.
The rigorous mathematical substantiation of the Fokker--Planck equation 
in the framework of probability theory was given by Kolmogorov~\cite{Kolmogorov1938}.
Now the Fokker--Planck equations form a basis for analytical methods in the 
investigation of continuous Markov processes.
Fokker--Planck equations with different coefficients 
describe the evolution of one-particle distribution functions of a dilute gas with long-range collisions,
Brownian motion without drift,
problems of diffusion in colloids, population genetics,
financial markets, quantum chaos, etc. \cite{Feller1951,Gardiner1985,Gihman&Skorohod1980,Risken1989}.

Local conservation laws of linear parabolic equations are, indeed, well understood.
More precisely, it is well known that the space of characteristics of the local conservation laws of 
a linear system of partial differential equations comprises those functions of the independent variables 
which solve the adjoint system. 
However, to the best of our knowledge, the statement that all characteristics are equivalent to such functions 
was proved, in an explicit way, only for the linear heat equation~\cite{Dorodnitsyn&Svirshchevskii1983} so far.
Below this statement is extended to the entire class~\eqref{EqGenLPE}.
Moreover, it is proved that \emph{linear second-order parabolic equations 
have no purely potential conservation laws}. 
In other words, for these equations potential conservation laws of any level are exhausted by local conservation laws. 
This generalizes the analogous statement from~\cite{Popovych&Ivanova2004ConsLawsLanl} on the 
linear heat equation and finally solves the problem on potential conservation laws in class~\eqref{EqGenLPE}. 

Let us emphasize that in fact it is not a common situation when all the characteristics of the local and, especially, 
potential conservation laws of a linear system can be taken as functions of the independent variables alone, i.e., 
the associated conserved vectors are linear in the unknown functions and their derivatives. 
Even the usual wave equation and third-order parabolic equations have characteristics essentially depending on derivatives. 
An example of a third-order parabolic equation with such characteristics is given in Note~\ref{NoteOnThirdOrderLPEs} 
of the present paper.
Moreover, a self-adjoint linear system possessing a quadratic conservation law has an infinite series of 
such conservation laws.
See, e.g., the section on symmetric linear systems in Chapter~5 and the concluding remarks of the corresponding 
chapter in~\cite{Olver1986}.

In contrast to conservation laws, potential symmetries have not been sufficiently investigated 
even for simplest equations from class~\eqref{EqGenLPE}.
Thus, e.g., for the linear heat equation potential symmetries were studied 
only in the case of the single characteristics~$1$~\cite{Bluman&Kumei1989,Sophocleous1996,Popovych&Ivanova2003PETs} 
and~$x$~\cite{Ivanova&Popovych2007CommentOnMei}. 
In~\cite{Ivanova&Popovych2007CommentOnMei,Pucci&Saccomandi1993,Saccomandi1997} potential symmetries 
of the Fokker--Planck equation $u_t=u_{xx}+(xu)_x$, associated with the characteristic~$1$, were found.
First-order conservation laws of the Fokker--Planck equations of the form $u_t=u_{xx}+(B(x)u)_x$ 
and potential symmetries of such equations, associated with the characteristic~$1$, were investigated in~\cite{Pucci&Saccomandi1993b}.
Note that the idea of symmetry extension via involving potentials and pseudopotentials in the transformations 
as new (nonlocal) dependent variables was already presented in the monograph by Edelen~\cite{Edelen1980}.
The concept of potential symmetry was explicitly formulated first by Bluman {\it et al} 
and was subsequently applied in investigations of important classes of partial differential equations~\cite{Bluman&Kumei1989,Bluman&Reid&Kumei1988}.
The related notion of quasilocal symmetry was proposed in~\cite{Akhatov&Gazizov&Ibragimov1987,Akhatov&Gazizov&Ibragimov1989}.
A systematic procedure of constructing quasilocal symmetries of (1+1)-dimensional evolution equations 
was described in \cite{Basarab-Horwath&Lahno&Zhdanov2001,Zhdanov&Lahno2005} and its application was illustrated by nontrivial examples.
This procedure is based on the exhaustive group classification of such equations. 
To the best of our knowledge, the problem of finding criteria for the existence of potential symmetries 
for classes of differential equations was first posed by Pucci and Saccomandi~\cite{Pucci&Saccomandi1993}.

The problem of a complete description of potential symmetries is very difficult to solve not only for classes
of equations but even for single equations.
In particular, it includes studying symmetry properties of infinite series of potential systems associated with 
tuples of an arbitrary number of linearly independent characteristics.
We propose effective criteria for the existence of potential symmetries of equations from class~\eqref{EqGenLPE}. 
They are subsequently applied to the classification of simplest potential symmetries of the linear heat equation and 
the separation of subclasses possessing infinite series of purely potential symmetry algebras. 
The framework of potential symmetries in class~\eqref{EqGenLPE} appears to be closely related to the theory 
Darboux transformations~\cite{Matveev&Salle1991} in the same class. 

Our paper is organized as follows: 
Using as a guideline the notion of normalization of classes of differential equations, 
in Section~\ref{SectionOnLieSymLPEs} we review and extend the classical results
on Lie symmetries and equivalence transformations of class~\eqref{EqGenLPE}, 
introduced by Lie \cite{Lie1881} and Ovsiannikov \cite{Ovsiannikov1982}.
In particular, difficulties arising under group classification of the Kolmogorov and Fokker--Planck equations 
are satisfactorily explained in terms of normalized classes.
The necessary definitions and statements on conservation laws and their characteristics, 
the equivalence of conservation laws with respect to transformations groups and sets of admissible transformations 
and potential systems are collected in Section~\ref{SectionOnTheorBackgroundOnCLs} for convenient reference.
Special attention is paid to the equivalence of conservation laws with respect to transformation groups 
(Subsection~\ref{SectionOnEquivOfConsLaws}) as this notion is essential in studying the potential conservation 
laws of equations from class~\eqref{EqGenLPE}.

The local conservation laws of these equations are exhaustively described in 
Section~\ref{SectionOnLocalCLsOfLPEs} based on the direct method. 
Since for any equation under consideration the characteristic space of local conservation laws
is isomorphic to the solution set of the adjoint equation, 
in Section~\ref{SectionOnAdjointVariationalPrincipleForLPEs} the adjoint variational principle is 
specified for class~\eqref{EqGenLPE} and extended to the corresponding set of admissible transformations. 
A number of auxiliary statements on admissible transformations of second-order evolution systems is proved. 
The adjoint variational principle is then applied to the group classification of the Fokker--Planck equations. 

The main result on potential conservation laws in class~\eqref{EqGenLPE} is presented 
in Section~\ref{SectionOnPotCLsLPEs}. 
Namely, it is proved that the local conserved vectors of potential systems are equivalent to 
local conserved vectors of the corresponding equations. 
This also provides a complete description of potential systems, allowing us to initiate the investigation of 
potential symmetries. 

The simplest potential symmetries considered in Section~\ref{SectionOnSimplestPotSymsOfLPEs} form 
a subject whose investigation generates a number of ideas on a special technique of working with class~\eqref{EqGenLPE}.
The attribute `simplest' refers to the fact that these symmetries are associated with single characteristics
(as opposed to simplicity of calculation). 
Thus, in Section~\ref{SectionOnSimpestPotSymsOfLHE} the simplest potential symmetries of the linear heat equation
are classified with respect to its point symmetry group. There are only two inequivalent characteristics 
$\alpha=1$ and $\alpha=x$ giving simplest purely potential symmetries of the linear heat equation. 
This is the only example in the literature with an exhaustive investigation of at least 
simplest potential symmetries of equations from class~\eqref{EqGenLPE}. 
The obtained classification directly leads to the complete description, e.g., of the simplest second-level 
potential symmetries of the Burgers equation and the simplest potential symmetries of the equations  
which are equivalent to the linear heat equation with respect to point transformations. 
The classification of simplest potential symmetries of the Fokker--Planck equation $u_t=u_{xx}+(xu)_x$ 
is presented for illustration. 

Section~\ref{SectionOnGenPotSysForLPEs} is devoted to the construction of the potential frame 
over class~\eqref{EqGenLPE} and studying its (non-symmetry) properties. 
Different kinds of potentials, potential systems and potential equations 
($p$-order and $p$-level, usual and modified ones) associated with characteristic tuples are defined. 
Explicit expressions for all these object are found. 
The multiple dual Darboux transformation provides a clear connection between components of the potential frame. 
This is why enhanced statements on Darboux transformations in class~\eqref{EqGenLPE} are also presented. 
Probably the most important components of the potential frame are the so-called modified potentials and 
modified potential equations. 
In contrast to other components, they are invariant under nonsingular linear combining of characteristics in 
the associated tuples. Moreover, as proved in Section~\ref{SectionOnGeneralPotSymsOfLPEs}, Lie symmetry analysis of 
potential systems is reduced to group classification of modified potential equations 
with respect to the equivalence group of class~\eqref{EqGenLPE} prolonged to the whole potential frame. 
Another result of this section is the possibility and precise realization of the prolongation. 
A statement on generalized potential symmetries of equations from class~\eqref{EqGenLPE} 
is also proved. This serves to substantiate the use of the canonical form of the conserved vectors 
in the construction of the potential frame.

After analyzing Lie invariance of potential systems, we formulate, in different terms, 
criteria on the existence of general potential symmetries. 
Their effectiveness is demonstrated in Section~\ref{SectionOnNumberAndOrderOfPotSyms} 
via the construction of wide subclasses of class~\eqref{EqGenLPE} whose equations admit 
infinite series of potential symmetry algebras of arbitrarily large order. 
The multiple auto-Darboux transformation is used in this construction as a powerful auxiliary tool. 

In the final section the results of the paper are summarized and some open problems on potential symmetries of 
equations from class~\eqref{EqGenLPE} are formulated and discussed. 

We will refer to formula~\eqref{EqGenLPE} and similar ones describing classes of systems of differential equations in a 
twofold manner, 
namely either as to a whole class (then the arbitrary elements are assumed to run through all possible values) 
or as a single equation from this class (then the arbitrary elements are assumed to take fixed values). 

By default, the indices~$i$, $j$ and~$k$ run from~1 to~$n$, 
the indices~$a$ and~$b$ run from~1 to~$m$ and 
the indices~$s$, $\sigma$ and~$\varsigma$ run at most from~1 to~$p$. 
Additional or other constraints on indices are indicated explicitly. 
The summation convention over repeated indices is used unless otherwise stated 
or it is obvious from the context that indices are fixed.

\section{Group classification}\label{SectionOnLieSymLPEs}

The complete group classification of equations~\eqref{EqGenLPE} was performed by Sophus Lie~\cite{Lie1881} 
as a part of the more general group classification of linear second-order partial differential equations in two independent variables.
A modern treatment of the subject is given in~\cite{Ovsiannikov1982}.
There exist also a number of papers rediscovering results of Lie and Ovsiannikov \cite{Lie1881,Ovsiannikov1982} 
partially (see, e.g., 
\cite{Bluman1990,Cicogna&Vitali1990,Lahno&Spichak&Stognii2004,Sastri&Dunn1985,Shtelen&Stogny1989,Spichak&Stognii1999a,
Spichak&Stognii1999b,Spichak&Stognii1999c,Stohny1997}). 

Since our investigation on potential symmetries of equations from class~\eqref{EqGenLPE} 
is essentially based on the above results, we review them for the reader's convenience. 
Another reason of their consideration is to justify the choice of a suitable subclass 
whose investigation allows the description of potential symmetries and potential conservation laws 
in the whole class~\eqref{EqGenLPE}.
Moreover, we extend these results to the framework of admissible transformations in classes of differential equations.
The normalization properties of the class of linear parabolic equations and its subclasses 
with respect to point transformations are also studied. 

Roughly speaking, an admissible transformation in a class of systems of differential equations is 
a point transformation connecting at least two systems from this class (in the
sense that one system is transformed into the other by the transformation).
The class is called \emph{normalized} if any
admissible transformation in this class belongs to its equivalence group
and is called \emph{strongly normalized} if additionally the equivalence group
is generated by transformations from the point symmetry groups of systems from the class.
The set of admissible transformations of a \emph{semi-normalized class} is generated by
the transformations from the equivalence group of the whole class and the transformations
from the point symmetry groups of initial or transformed systems. 
Strong semi-normalization is defined in the same way as strong normalization.
Any normalized class is semi-normalized. 
Two systems from a semi-normalized class are transformed into one another by a point transformation
iff they are equivalent with respect to~the equivalence group of this class.
See \cite{Popovych2006a,Popovych2006c,Popovych&Kunzinger&Eshraghi2006,Popovych&Eshraghi2005} 
for precise definitions and statements.

To begin with, consider the class of inhomogeneous equations corresponding to~\eqref{EqGenLPE} 
of the general form 
\begin{equation}\label{EqGenInhomLPE}
u_t=A(t,x)u_{xx}+B(t,x)u_x+C(t,x)u+D(t,x),
\end{equation}
where $A(t,x)$, $B(t,x)$, $C(t,x)$ and $D(t,x)$ are arbitrary smooth functions, $A(t,x)\ne0$.
A~motivation for the provisional extension of the class under consideration is that 
class~\eqref{EqGenInhomLPE} has a nicer normalization property than the initial class. 

Any point transformation~$\mathcal T$ in the space of variables $(t,x,u)$ has the form 
\[
\tilde t=\mathcal T^t(t,x,u),\quad
\tilde x=\mathcal T^x(t,x,u),\quad 
\tilde u=\mathcal T^u(t,x,u),\quad 
\]
where the Jacobian $|\p(\mathcal T^t,\mathcal T^x,\mathcal T^u)/\p(t,x,u)|$ does not vanish.

\begin{lemma}\label{LemmaOnAdmTransOfInhomLPEs}
A point transformation~$\mathcal T$ connects two equations from class~\eqref{EqGenInhomLPE} iff 
$\,\mathcal T^t_x=\mathcal T^t_u=0$, $\mathcal T^x_u=0$, $\mathcal T^u_{uu}=0$, i.e.,
\begin{equation}\label{EqGenFormOfTransOfInhomLPEs}
\mathcal T^t=T(t), \quad \mathcal T^x=X(t,x), \quad \mathcal T^u=U^1(t,x)u+U^0(t,x),
\end{equation}
where $T$, $X$, $U^1$ and $U^0$ are arbitrary smooth functions of their arguments such that $T_tX_xU^1\ne0$. 
The arbitrary elements are transformed by the formulas
\begin{gather}\label{EqTransOfCoeffsOfLPE}
\tilde A=\frac{X_x^2}{T_t}A,\quad 
\tilde B=\frac{X_x}{T_t}\left(B-2\frac{U^1_x}{U^1}A\right)-\frac{X_t-AX_{xx}}{T_t},\quad 
\tilde C=-\frac{U^1}{T_t}L\frac1{U^1},
\\ \label{EqTransOfInhomCoeffOfInhomLPE}
\tilde D=\frac{U^1}{T_t}\left(D+L\frac{U^0}{U^1}\right).
\end{gather}
Here $L=\p_t-A\p_{xx}-B\p_x-C$ is the second-order linear differential operator associated with 
the initial (non-tilde) equation.
\end{lemma}
\begin{proof}
The proof is based on the direct method. We recalculate the old derivatives in the new variables, 
substitute the obtained expressions in the initial equation and then split it on the manifold of 
the transformed equation. 
As a result, we derive the determining equations for the components of~$\mathcal T$. 
The calculations can be simplified if we take into account known restrictions for transformations between 
equations from wider classes. Thus, equations of the form~\eqref{EqGenInhomLPE} are evolutionary. 
Any transformation between two evolution equations satisfies the conditions 
$\mathcal T^t_x=\mathcal T^t_u=0$ \cite{Kingston&Sophocleous1998}. 
Then, a transformation between quasi-linear equations satisfies the conditions $\mathcal T^x_u=0$ and $\mathcal T^u_{uu}=0$
(see, e.g., \cite{Popovych&Ivanova2004NVCDCEs,Prokhorova2005}). 
The remaining determining equations implies the formulas for transforming the arbitrary elements.
\end{proof}

\begin{corollary}\label{CorollaryOnNormalizationOfInhomLPEs}
Class~\eqref{EqGenInhomLPE} is strongly normalized. 
The equivalence group~$G^\sim_{\rm inh}$ of class~\eqref{EqGenInhomLPE} is formed by the transformations 
determined in the space of variables and arbitrary elements by formulas~\eqref{EqGenFormOfTransOfInhomLPEs}, 
\eqref{EqTransOfCoeffsOfLPE} and~\eqref{EqTransOfInhomCoeffOfInhomLPE}, 
where $T$, $X$, $U^1$ and~$U^0$ are arbitrary smooth functions of their arguments such that $T_tX_xU^1\ne0$.
\end{corollary}

Using transformations from~$G^\sim_{\rm inh}$, we can gauge arbitrary elements of class~\eqref{EqGenInhomLPE}. 
Thus, applying the equivalence transformation with $T=t$, $X=x$, $U^1=1$ and $U^0$ a solution of the 
equation $LU^0=-D$, we obtain the standard conversion of the inhomogeneous equation $Lu=D$ to 
the homogeneous one $Lu=0$. As a result, class~\eqref{EqGenInhomLPE} is mapped to class~\eqref{EqGenLPE}. 
Unfortunately, the normalization property is broken under this mapping.

\begin{corollary}\label{CorollaryOnAdmTransOfLPEs}
A point transformation~$\mathcal T$ connects two equations from class~\eqref{EqGenLPE} iff 
its components are of the form~\eqref{EqGenFormOfTransOfInhomLPEs}, 
where $T$, $X$ and $U^1$ are arbitrary smooth functions of their arguments such that $T_tX_xU^1\ne0$ 
and additionally $U^0/U^1$ is a solution of the initial equation. 
The arbitrary elements are transformed by formulas~\eqref{EqTransOfCoeffsOfLPE}.
\end{corollary}

\begin{corollary}\label{CorollaryOnSemiNormOfLPEs}
Class~\eqref{EqGenLPE} is strongly semi-normalized. 
The equivalence group~$G^\sim$ of class~\eqref{EqGenLPE} is formed by the transformations 
determined in the space of variables and arbitrary elements by formulas~\eqref{EqGenFormOfTransOfInhomLPEs}, 
\eqref{EqTransOfCoeffsOfLPE}, where $T$, $X$ and $U^1$ are arbitrary smooth functions of their arguments 
such that $\,T_tX_xU^1\ne0$ and $U^0=0$ additionally. 
\end{corollary}

\begin{note}\label{NoteOnNormalizationOfInhomAndHomLinSys}
A similar relation holds between normalization properties of general classes of 
inhomogeneous and the corresponding homogeneous linear systems of differential equations. 
Namely, consider a class $\mathcal L_{\rm hom}$ of homogeneous linear systems 
of $l$ differential equations of the form $L_\theta u=0$, 
for $m$ unknown functions $u=(u^1,\dots,u^m)$ of $n$ independent variables $x=(x_1,\dots,x_n)$. 
Here $L_\theta=(L^{\mu a}_\theta)$, $\mu=1,\dots,l$, is a matrix differential operator parameterized 
with $\theta$ running through a parameter set. 
Let the corresponding class $\mathcal L_{\rm inh}$ of inhomogeneous systems be normalized and 
its equivalence group~$G^\sim_{\rm inh}$ consist of transformations projectable in~$x$ 
and affine in~$u$, i.e. the transformations of $(x,u)$ have the form 
$\tilde x=X(x)$ and $\tilde u=U^{ab}(x)u^b+U^{a0}(x)$ for any $\mathcal T\in G^\sim_{\rm inh}$. 
Then the class $\mathcal L_{\rm hom}$ is semi-normalized. 
The equivalence group~$G^\sim_{\rm hom}$ of $\mathcal L_{\rm hom}$ is isomorphic 
to the subgroup of~$G^\sim_{\rm inh}$ formed by the transformations with 
$(\hat U^{ab}U^{a0})$ running through the intersection of the solution sets of the systems from~$\mathcal L_{\rm hom}$, 
where $(\hat U^{ab})$ is the inverse matrix of~$(U^{ab})$. 
Often this implies that $U^{a0}=0$.
The additional condition for the admissible transformations in $\mathcal L_{\rm hom}$ is that 
$(\hat U^{ab}U^{b0})$ is a solution of the initial system (with fixed values of the arbitrary elements).
Therefore, normalization is broken under restricting to $\mathcal L_{\rm hom}$ 
due to the presence of the linear superposition principle. 
This justifies the consideration of inhomogeneous linear systems in the framework of admissible transformations. 

\end{note}

Another possibility is to gauge the arbitrary element~$A$ in class~\eqref{EqGenInhomLPE} to~1 with a transformation 
of form~\eqref{EqGenFormOfTransOfInhomLPEs}, where $T_t=\sign A$, $X_x=|A|^{-1/2}$, $U^1=1$ and $U^0=0$. 
The admissible transformations in the subclass of~\eqref{EqGenInhomLPE} with $A=1$ are those
transformations~\eqref{EqGenFormOfTransOfInhomLPEs} which preserve the condition $A=1$, 
i.e., which additionally satisfy the condition $\mathcal T^t_t=(\mathcal T^x_x)^2$.

\begin{corollary}\label{CorollaryOnAdmTransOfLPEsA1}
A point transformation~$\mathcal T$ connects two equations from class~\eqref{EqGenInhomLPE} 
with $A=\tilde A=1$~iff 
\begin{equation}\label{EqGenFormOfTransOfInhomA1LPEs}
\mathcal T^t=T(t), \quad \mathcal T^x=X=\pm \sqrt{T_t(t)}\,x+\zeta(t), \quad \mathcal T^u=U^1(t,x)u+U^0(t,x),
\end{equation}
where $T$, $\zeta$, $U^1$ and $U^0$ are arbitrary smooth functions of their arguments such that $T_t>0$ and $U^1\ne0$. 
The transformations of this form, prolonged to the arbitrary elements $B$, $C$ and $D$ 
by formulas~\eqref{EqTransOfCoeffsOfLPE} and~\eqref{EqTransOfInhomCoeffOfInhomLPE} 
constitute the equivalence group of the subclass of~\eqref{EqGenInhomLPE} with $A=1$.
This subclass is strongly normalized. 
\end{corollary}

Analogously to the case of the entire class~\eqref{EqGenInhomLPE}, we can convert the inhomogeneous equations with $A=1$ to 
the homogeneous ones. 
As a result, the subclass of~\eqref{EqGenInhomLPE} with $A=1$ is mapped to 
the subclass of~\eqref{EqGenLPE} satisfying the same condition. 
The normalization property is again broken under this mapping.
A point transformation~$\mathcal T$ connects two equations from class~\eqref{EqGenLPE} with $A=\tilde A=1$~iff 
it has the form adduced in Corollary~\ref{CorollaryOnAdmTransOfLPEsA1} 
and additionally $U^0/U^1$ is a solution of the initial equation. 
The subclass of~\eqref{EqGenLPE} with $A=1$ is strongly semi-normalized. 
Its equivalence group is formed by the transformations~\eqref{EqGenFormOfTransOfInhomA1LPEs} with $U^0=0$, 
prolonged to the arbitrary elements $B$ and $C$ by formulas~\eqref{EqTransOfCoeffsOfLPE}. 

The arbitrary elements~$A$ and~$B$ can be simultaneously gauged to~1 and~0. 
The subclass of~\eqref{EqGenInhomLPE} with $(A,B)=(1,0)$ is a restriction of the one with $A=1$ 
and is investigated in a similar way. 

\begin{corollary}\label{CorollaryOnAdmTransOfLPEsA1B0}
A point transformation~$\mathcal T$ connects two equations from class~\eqref{EqGenInhomLPE} 
with $A=\tilde A=1$ and $B=\tilde B=0$ iff it has the form~\eqref{EqGenFormOfTransOfInhomA1LPEs}, 
where additionally
\[
U^1=\theta(t)\exp\left(-\frac{T_{tt}}{8T_t}x^2\mp\frac{\zeta_t}{2T_t{}^{1/2}}x\right)
\]
and $T$, $\zeta$, $\theta$ and $U^0$ are arbitrary smooth functions of their arguments such that $T_t>0$ and $\theta\ne0$. 
The transformations of this form, prolonged to arbitrary elements $C$ and $D$ 
by formulas~\eqref{EqTransOfCoeffsOfLPE} and~\eqref{EqTransOfInhomCoeffOfInhomLPE} 
constitute the equivalence group of the subclass of~\eqref{EqGenInhomLPE} with $A=1$ and $B=0$.
This subclass is strongly normalized. 
\end{corollary}

The inhomogeneous equations with $A=1$ and $B=0$ are mapped to the homogeneous ones in the standard way. 
Hence, any equation from class~\eqref{EqGenInhomLPE} or class~\eqref{EqGenLPE} can be 
reduced by a transformation from the corresponding equivalence group to an equation of the general form
\begin{equation}\label{EqReducedLPE}
u_t-u_{xx}+V(t,x)u=0.
\end{equation}
The normalization property is broken for the class~\eqref{EqReducedLPE}.
A point transformation~$\mathcal T$ connects two equations from class~\eqref{EqReducedLPE} iff 
it has the form adduced in Corollary~\ref{CorollaryOnAdmTransOfLPEsA1B0} 
and additionally $U^0/U^1$ is a solution of the initial equation. 
At the same time, class~\eqref{EqReducedLPE} is strongly semi-normalized. 
Its equivalence group~$G^\sim_1$ is formed by the transformations with $U^0=0$, 
prolonged to the arbitrary elements $V=-C$ by formulas~\eqref{EqTransOfCoeffsOfLPE}. 
Therefore, the functions parameterizing~$G^\sim_1$ depend only on~$t$. 
The narrower equivalence group under preserving certain normalization properties suggests
class~\eqref{EqReducedLPE} as the most convenient one for group classification. 
Moreover, solving the group classification problem for any of the above classes is reduced 
to solving the group classification problem for class~\eqref{EqReducedLPE}.
This is why we formulate the main result on admissible transformations in class~\eqref{EqReducedLPE} as a theorem. 

\begin{theorem}
Class~\eqref{EqReducedLPE} is strongly semi-normalized. 
Any transformation from the equivalence group~$G^\sim_1$ of class~\eqref{EqReducedLPE} has the form
\begin{gather}\textstyle\nonumber
\tilde t=\int\! \sigma^2dt,\quad \tilde x=\sigma x+\zeta,\quad 
\tilde u=u\theta\exp\left(-\dfrac{\sigma_t}{4\sigma}x^2-\dfrac{\zeta_t}{2\sigma}x\right),
\\ \label{EqTransFromEquivGroupOfReducedLPEs}
\tilde V=\frac1{\sigma^2}\left(V+\frac{\sigma\sigma_{tt}-2\sigma_t{}^2}{4\sigma^2}x^2
+\frac{\sigma\zeta_{tt}-2\sigma_t\zeta_t}{2\sigma^2}x-\frac{\theta_t}\theta
-\frac{\sigma_t}{2\sigma}-\dfrac{\zeta_t{}^2}{4\sigma^2}\right),
\end{gather}
where $\sigma=\sigma(t)$, $\zeta=\zeta(t)$ and $\theta=\theta(t)$ are arbitrary smooth functions, 
$\sigma\theta\ne0$.
\end{theorem}

For our further considerations we need to introduce terminology connected with the symmetry structure of linear equations. 
In view of the linear superposition principle, the point symmetry group~$G(\mathcal L)$ and 
the maximal Lie invariance algebra~$\mathfrak g(\mathcal L)$ 
of any (homogeneous) linear differential equation
(or any system of such equations)~$\mathcal L$ have certain properties~\cite{Ovsiannikov1982}. 
Namely, $G(\mathcal L)$ contains the translations $u\to u+\varepsilon f$ of the unknown function~$u$ 
by an arbitrary solution~$f$ of~$\mathcal L$ 
and the scale transformations $u\to\pm e^{\varepsilon}u$ superimposed with reflection of~$u$. 
These sets of transformations form subgroups of~$G(\mathcal L)$ both separately and simultaneously. 
The whole subgroup of symmetry transformations associated with the linear superposition principle of~$\mathcal L$ 
is called the \emph{trivial symmetry group} of the linear differential equation~$\mathcal L$ and 
will be denoted by $G^{\rm triv}(\mathcal L)$. 
The algebra~$\mathfrak g(\mathcal L)$ contains the corresponding operators $f\p_u$ and $u\p_u$ 
which form the Lie algebra $\mathfrak g^{\rm triv}(\mathcal L)$ called the \emph{trivial invariance algebra} 
of the linear differential equation~$\mathcal L$. 
Hereafter the function~$f$ runs through the solution set of~$\mathcal L$.

The ideal $\mathfrak g^\infty(\mathcal L)=\langle f\p_u\rangle$ of $\mathfrak g^{\rm triv}(\mathcal L)$ is 
called the (trivial) infinite-dimensional part of~$\mathfrak g(\mathcal L)$. 
The corresponding normal subgroup of~$G(\mathcal L)$ is similarly denoted by~$G^\infty(\mathcal L)$.
This notation is justified in the following way: 
Usually~\cite{Ovsiannikov1982} the maximal Lie invariance algebra~$\mathfrak g(\mathcal L)$ 
of a linear differential equation~$\mathcal L$ can be represented in the form
\[
\mathfrak g(\mathcal L)=\mathfrak g^{\rm ess}(\mathcal L)\rsemioplus\mathfrak g^\infty(\mathcal L). 
\]
Here $\mathfrak g^\infty(\mathcal L)$ is an (infinite-dimensional) Abelian ideal of
the algebra~$\mathfrak g(\mathcal L)$ 
and $\mathfrak g^{\rm ess}(\mathcal L)$ is its finite-dimensional subalgebra spanned by 
the Lie invariance operators of~$\mathcal L$, which are projectable to the independent variables 
and whose coefficients of~$\p_u$ depend linearly on~$u$. 
In particular, the above representation is true for the equations from class~\eqref{EqGenLPE}. 
Then a similar representation also holds for the group~$G(\mathcal L)$. 
Namely, $G(\mathcal L)=G^{\rm ess}(\mathcal L)\times G^\infty(\mathcal L)$, 
where $G^\infty(\mathcal L)$ is a normal subgroup of~$G(\mathcal L)$ and $G^{\rm ess}(\mathcal L)$ its subgroup.

The operators from~$\mathfrak g^{\rm ess}(\mathcal L)$ will be called \emph{essential symmetry operators} 
of the linear differential equation~$\mathcal L$ 
since mainly they and the corresponding finite transformations are useful for group analysis of~$\mathcal L$. 
Such operators can be found in a particular way via the commutation relation with the differential operator associated 
with the equation~$\mathcal L$~\cite{FushchichNikitin1994} (see also Section~\ref{SectionOnAdjointVariationalPrincipleForLPEs})
that provides possibilities for various generalizations of the notion of symmetry operators. 
They are employed for finding finite symmetry transformations of~$\mathcal L$, 
for constructing exact solutions via the Lie reduction procedure 
to differential equations with fewer independent variables 
and for the direct generation of new exact solutions by acting on known ones~\cite{Ovsiannikov1982}. 
The algebra~$\mathfrak g^{\rm ess}(\mathcal L)$ will be called the \emph{essential Lie invariance algebra} 
of the linear differential equation~$\mathcal L$. 

The intersection $\mathfrak g^{\rm ess}(\mathcal L)\cap\mathfrak g^{\rm triv}(\mathcal L)=\langle u\p_u\rangle$ 
is contained in the center of~$\mathfrak g^{\rm ess}(\mathcal L)$. 
The dimension of~$\mathfrak g^{\rm ess}(\mathcal L)/\langle u\p_u\rangle$ is called the \emph{dimension of extension} of 
the maximal Lie invariance algebra or the number of independent nontrivial Lie symmetry operators.  

In terms of these notations we can formulate Corollary~\ref{CorollaryOnSemiNormOfLPEs} and 
similar statements more precisely.
Any point transformation between two equations from class~\eqref{EqGenLPE} is the composition 
of a trivial symmetry transformation from~$G^\infty$ of the initial equation and a transformation from~$G^\sim$.

The results on the group classification of class~\eqref{EqReducedLPE} can be formulated in the form of 
the following theorem~\cite{Lie1881,Ovsiannikov1982}. 

\begin{theorem}\label{TheoremOnGroupClassificationOfLPEs}
The kernel Lie algebra of class~\eqref{EqReducedLPE} is $\langle u\p_u\rangle$.
Any equation from class~\eqref{EqReducedLPE} is invariant with respect 
to the operators $f\p_u$,
where the parameter-function $f=f(t,x)$ runs through the solution set of this equation.
All possible~$G^\sim_1$-inequivalent cases of extension of the maximal 
Lie invariance algebra are exhausted by the following ones 
(the values of~$V$ are given together with the corresponding maximal Lie invariance algebras):

\vspace{1.5ex}

$\makebox[6mm][l]{\rm 1.}
V=V(x) \colon\quad \langle\p_t,\,u\p_u,\,f\p_u\rangle;$

\vspace{1.5ex}

$\makebox[6mm][l]{\rm 2.}
V=\mu x^{-2},\ \mu\ne0 \colon\quad \langle\p_t,\, D,\, \Pi,\, u\p_u,\, f\p_u\rangle;$

\vspace{1.5ex}

$\makebox[6mm][l]{\rm 3.}
V=0 \colon\quad \langle\p_t,\, \p_x,\, G,\, D,\, \Pi,\, u\p_u,\, f\p_u\rangle$.

\vspace{1.5ex}

\noindent
Here $D=2t\p_t+x\p_x,\ \Pi=4t^2\p_t+4tx\p_x-(x^2+2t)u\p_u,\ G=2t\p_x-xu\p_u.$
\end{theorem}

\begin{note}
It is assumed in case~1 of Theorem~\ref{TheoremOnGroupClassificationOfLPEs} that 
the value~$V=V(x)$ is $G^\sim_1$-inequivalent to the value $\mu x^{-2}$, where $\mu\in\mathbb R$.
\end{note}

\begin{note}
Theorem~\ref{TheoremOnGroupClassificationOfLPEs} can be reformulated 
for the entire classes~\eqref{EqGenLPE} and~\eqref{EqGenInhomLPE} 
if $G^\sim_1$-equivalence is replaced by $G^\sim$- and $G^\sim_{\rm inh}$-equivalences correspondingly. 
A similar reformulation is possible also for subclasses with $A=1$. 
Let us emphasize that the group classification in a semi-normalized class with respect to its equivalence group 
is identical to the classification up to all admissible point transformations.
\end{note}

\begin{corollary}\label{CorollaryOnNumberOfNontrivSymsOfLPEs}
For any equation~$\mathcal L$ from class~\eqref{EqGenLPE} $\dim\mathfrak g^{\rm ess}(\mathcal L)\in\{1,2,4,6\}$, i.e., 
the number of independent nontrivial symmetries belongs to $\{0,1,3,5\}$.
If $\dim\mathfrak g^{\rm ess}(\mathcal L)=6$ then the equation~$\mathcal L$ is $G^\sim$-equivalent to 
the linear heat equation $u_t=u_{xx}$.
\end{corollary}

\begin{note}
The presented way of gauging the arbitrary elements is optimal for group classification. 
The hierarchy of normalized classes of inhomogeneous equations 
and the corresponding semi-normalized classes of homogeneous equations are constructed. 
Due to its properties, the subclass~\eqref{EqReducedLPE} is convenient for solving the group classification problem. 
The obtained results can be extended in an obvious way to all classes from the hierarchy. 
Different choices of gauges for the arbitrary elements (e.g., reduction to the `Kolmogorov' or `Fokker--Planck' form)
may lead to a considerable complication of the problem. 
\end{note}

Consider the group classification problem for the `Kolmogorov' form ($C=0$) of equations from class~\eqref{EqGenLPE}.
(It follows from results of Section~\ref{SectionOnLocalCLsOfLPEs} 
that the symmetry analysis of the `Fokker--Planck' form, being adjoint to the `Kolmogorov' form,
is reduced to an investigation of the `Kolmogorov' form.) 
Note that the signs of~$A$ and~$B$ are inessential under symmetry investigation 
due to the presence of equivalence transformations alternating the signs.

The gauge $C=0$ totally breaks the normalization properties.
Indeed, a point transformation~$\mathcal T$ connects two equations from the class~\eqref{EqGenLPE} with $C=\tilde C=0$
iff its components are of the form~\eqref{EqGenFormOfTransOfInhomLPEs}, 
where $T$ and $X$ are arbitrary smooth functions of their arguments such that $T_tX_x\ne0$, $U^1\ne0$ 
and additionally $1/U^1$ and $U^0/U^1$ are solutions of the initial equation. 
The arbitrary elements $A$ and $B$ are transformed by formulas~\eqref{EqTransOfCoeffsOfLPE}.
The equivalence group of the subclass of~\eqref{EqGenLPE} with $C=0$
consists of only those transformations of the form~\eqref{EqGenFormOfTransOfInhomLPEs} with $U^1,U^0=\const$.
Therefore, this subclass is not semi-normalized.
There exist equations in it, transformed into one another by a point transformation,
which are inequivalent with respect to the equivalence group.
The structure of admissible transformations which are not generated by transformations from the equivalence group 
is quite complicated. 
That is why it seems too difficult to present a classification for the subclass with respect to its equivalence group. 
The additional gauge $A=1$ does not improve the situation. 
A classification up to its set of admissible transformations is derived 
from Theorem~\ref{TheoremOnGroupClassificationOfLPEs} by mapping the listed equations to the `Kolmogorov' form.
 
\begin{corollary}\label{CorollaryOnClassificationOfLPEsC0}
The kernel Lie algebra of the subclass of~\eqref{EqGenLPE} with $C=0$ is $\langle u\p_u,\p_u\rangle$.
Any equation from this subclass is invariant with respect to the operators $f\p_u$,
where the parameter-function $f=f(t,x)$ runs through the solution set of this equation. 
By a point transformation it is reduced to an equation with $A=1$ from the same subclass.
All possible cases of extension of the maximal Lie invariance algebras in this subclass are exhausted, 
up to point transformations, by the following ones 
(in all the cases $A=1$; the values of~$B$ are given together with the corresponding maximal Lie invariance algebras):

\vspace{1.5ex}

$\makebox[6mm][l]{\rm 1.}
B=B(x) \colon\quad \langle\p_t,\,u\p_u,\,f\p_u\rangle;$

\vspace{1.5ex}

$\makebox[6mm][l]{\rm 2.}
B=\nu x^{-1},\ \nu\geqslant1,\ \nu\ne2 \colon\quad \langle\p_t,\, D,\, \Pi-2\nu tu\p_u,\, u\p_u,\, f\p_u\rangle;$

\vspace{1.5ex}

$\makebox[6mm][l]{\rm 3.}
B=x^{-1}\bigl(1-2\varkappa\tan(\varkappa\ln|x|)\bigr),\ \varkappa\ne0 \colon\quad 
\langle\p_t,\, 2D-xBu\p_u,\, \Pi-2txBu\p_u,\, u\p_u,\, f\p_u\rangle;$

\vspace{1.5ex}

$\makebox[6mm][l]{\rm 4.}
B=0 \colon\quad \langle\p_t,\, \p_x,\, G,\, D,\, \Pi,\, u\p_u,\, f\p_u\rangle$.

\vspace{1.5ex}

\noindent
Here $D=2t\p_t+x\p_x,\ \Pi=4t^2\p_t+4tx\p_x-(x^2+2t)u\p_u,\ G=2t\p_x-xu\p_u.$
\end{corollary}

Corollary~\ref{CorollaryOnClassificationOfLPEsC0} shows that 
even up to all admissible transformations, the group classification of equations in the `Kolmogorov' form 
is more complicated than the group classification in class~\eqref{EqReducedLPE}. 
Parameters of equations are explicitly included in expressions of symmetry operators. 
Case~2 of Theorem~\ref{TheoremOnGroupClassificationOfLPEs} is split into two cases of 
Corollary~\ref{CorollaryOnClassificationOfLPEsC0} 
(case 2 with $\nu=1+\sqrt{1+4\mu}$ if $4\mu\geqslant-1$ and case~3 with $\varkappa=\sqrt{-1/4-\mu}$ if $4\mu<-1$). 
These cases can be united only over the complex number.
The equations with $B=\nu x^{-2}$ and $B'=\nu' x^{-2}$ ($A=A'=1$)
are equivalent with respect to point 
transformations iff $\nu+\nu'=2$. The corresponding transformation is $u'=x^{\nu-1}u$, $t$ and $x$ remain unchanged. 
In particular, the equation with $B=2x^{-2}$ is reduced by the transformation $u'=xu$ to 
the linear heat equation ($B=0$). That is why the parameter~$\nu$ should be constrained in case~2.
The form of the arbitrary element~$B$ is not simple in case 3 and cannot be simplified in the real case. 
Therefore, it is preferable to carry out the symmetry analysis of equations in the form~\eqref{EqReducedLPE} and then 
to derive results for the `Kolmogorov' form. 

The group classification of the subclass with $C=0$ with respect to its equivalence group 
can be obtained from the classification presented in Corollary~\ref{CorollaryOnClassificationOfLPEsC0} 
by extending the classification cases by essential admissible transformations which are not 
generated by the equivalence group. Such transformations have the form 
$\tilde t=t$, $\tilde x=x$, $\tilde u=U^1(t,x)u$, where 
$U^1\ne0$ and $1/U^1$ is an arbitrary solution of the initial equations.

\section{Theoretical background on conservation laws}\label{SectionOnTheorBackgroundOnCLs}

To begin with, we present the necessary theoretical background on conservation laws and potential systems,
basically following~\cite{Bocharov&Co1997,Olver1986,Popovych&Ivanova2004ConsLawsLanl,Zharinov1986}. 

\subsection{Definition of local conservation laws}\label{SectionOnCLsDef}

Let~$\mathcal L$ be a system~$L(x,u_{(\rho)})=0$ of $l$ differential equations $L^1=0$, \ldots, $L^l=0$
for $m$ unknown functions $u=(u^1,\ldots,u^m)$
of $n$ independent variables $x=(x_1,\ldots,x_n).$
Here $u_{(\rho)}$ denotes the set of all the derivatives of the functions $u$ with respect to $x$
of order not greater than~$\rho$, including $u$ as the derivative of order zero.
Let $\mathcal L_{(k)}$ denote the set of all algebraically independent differential consequences
that have, as differential equations, orders not greater than $k$. We identify~$\mathcal L_{(k)}$ with
the manifold determined by~$\mathcal L_{(k)}$ in the jet space~$J^{(k)}$.

\begin{definition}\label{def.conservation.law}
A {\em conserved vector} of the system~$\mathcal L$ is
an $n$-tuple $F=(F^1(x,u_{(r)}),\ldots,F^n(x,u_{(r)}))$ for which the divergence $\mathop{\rm Div}\nolimits F:=D_iF^i$
vanishes for all solutions of~$\mathcal L$, i.e., 
\begin{equation}\label{EqDefOfConservedVector}
\mathop{\rm Div}\nolimits F\,\bigl|_{\mathcal L}=0.
\end{equation}
\end{definition}

In Definition~\ref{def.conservation.law} and below
$D_i=D_{x_i}$ denotes the operator of total differentiation with respect to the variable~$x_i$, i.e.,
$D_i=\p_{x_i}+u^a_{\alpha,i}\p_{u^a_\alpha}$, where
$u^a_\alpha$ and $u^a_{\alpha,i}$ stand for the variables in jet space
which correspond to the derivatives
$\p^{|\alpha|}u^a/\p x_1^{\alpha_1}\ldots\p x_n^{\alpha_n}$ and $\p u^a_\alpha/\p x_i$,
$\alpha=(\alpha_1,\ldots,\alpha_n)$,
$\alpha_i\in\mathbb{N}\cup\{0\}$, $|\alpha|{:}=\alpha_1+\cdots+\alpha_n$.
We use the summation convention for repeated indices and assume any function as its zero-order derivative. 
The indices $i$, $j$ and $k$ run from~1 to~$n$, 
the index~$a$ runs from 1 to~$m$.
The notation~$V\bigl|_{\mathcal L}$ means that the values of $V$ are considered
only on solutions of the system~$\mathcal L$.
 
Heuristically, a conservation law of the system~$\mathcal L$ 
is an expression $\mathop{\rm Div}\nolimits F$ vanishing on the solutions of~$\mathcal L$. 
The more rigorous definition of conservation laws given below is based on the factorization of the space 
of conserved vectors with respect to the subspace of trivial conserved vectors. 
Note that there is also a formalized definition of conservation laws of~$\mathcal L$ 
as $(n-1)$-dimensional cohomology classes 
in the so-called horizontal de Rham complex on the infinite prolongation of the system~$\mathcal L$ 
\cite{Bocharov&Co1997,Tsujishita1982,Vinogradov1984}. 
The formalized definition is appropriate for certain theoretical considerations and 
reduces to the usual one after local coordinates are fixed. 

\begin{definition}
A conserved vector $F$ is called {\em trivial} if $F^i=\hat F^i+\check F^i$, 
where $\hat F^i$ and $\check F^i$ are, like $F^i$, smooth functions of $x$ and derivatives of $u$
(i.e. differential functions),
$\hat F^i$ vanishes on the solutions of~$\mathcal L$ and the $n$-tuple $\check F=(\check F^1,\ldots,\check F^n)$
is a null divergence (i.e., its divergence vanishes identically).
\end{definition}

The triviality concerning conserved vectors vanishing on solutions of the system can easily be
eliminated by restricting to the manifold of the system, taking into account all its differential consequences.
A (local) characterization of all null divergences is given by the following lemma (see e.g.~\cite{Olver1986}).

\begin{lemma}\label{lemma.null.divergence}
The $n$-tuple $F=(F^1,\ldots,F^n)$, $n\geqslant2$, is a null divergence ($\mathop{\rm Div}\nolimits F\equiv0$)
iff there exist smooth functions $v^{ij}$ of $x$ and derivatives of $u$,
such that $v^{ij}=-v^{ji}$ and $F^i=D_jv^{ij}$.
\end{lemma}

The functions $v^{ij}$ are called {\em potentials} corresponding to the null divergence~$F$.
If $n=1$ any null divergence is constant.

\begin{definition}\label{DefinitionOfConsVectorEquivalence}
Two conserved vectors $F$ and $F'$ are called {\em equivalent} if
the vector-function $F'-F$ is a trivial conserved vector.
\end{definition}

If $\mathcal L$ is a system of ordinary differential equations ($n=1$) 
then the conserved quantities (first integrals) $F$ and $F'$ are equivalent 
by Definition~\ref{DefinitionOfConsVectorEquivalence} if their difference is constant on the solutions of~$\mathcal L$.

In the case of two independent variables we re-denote $(F^1,F^2)\to (F,G)$ and $(x_1,x_2)\to(t,x)$.
Two conserved vectors $(F,G)$ and $(F',G')$ are equivalent if
there exist functions~$\hat F$, $\hat G$ and~$H$ of~$t$, $x$ and derivatives of~$u$ such that
$\hat F$ and $\hat G$ vanish on~$\mathcal L_{(k)}$ for some~$k$~and
\[
F'=F+\hat F+D_xH ,\qquad G'=G+\hat G-D_tH.
\]

The above definitions of triviality and equivalence of conserved vectors are natural
in view of the usual ``empiric'' definition of conservation laws of a system of differential equations
as divergences of its conserved vectors, i.e., divergence expressions which vanish for all solutions of this system.
For example, equivalent conserved vectors correspond to the same conservation law.
It allows us to formulate the definition of conservation law in a rigorous style (see, e.g.,~\cite{Bocharov&Co1997,Zharinov1986}).
Namely, for any system~$\mathcal L$ of differential equations the set~$\CV(\mathcal L)$ of conserved vectors of
its conservation laws is a linear space,
and the subset~$\CV_0(\mathcal L)$ of trivial conserved vectors is a linear subspace in~$\CV(\mathcal L)$.
The factor space~$\CL(\mathcal L)=\CV(\mathcal L)/\CV_0(\mathcal L)$
coincides with the set of equivalence classes of~$\CV(\mathcal L)$ with respect to the equivalence relation introduced in
Definition~\ref{DefinitionOfConsVectorEquivalence}.

\begin{definition}\label{DefinitionOfConsLaws}
The elements of~$\CL(\mathcal L)$ are called {\em conservation laws} of the system~$\mathcal L$,
and the factor space~$\CL(\mathcal L)$ is called {\em the space of conservation laws} of~$\mathcal L$.
\end{definition}

This is why we understand the description of the set of conservation laws
as finding~$\CL(\mathcal L)$, which in turn is equivalent to constructing either a basis if
$\dim \CL(\mathcal L)<\infty$ or a system of generators in the infinite dimensional case.
The elements of~$\CV(\mathcal L)$ which belong to the same equivalence class giving a conservation law~${\cal F}$
are all considered as conserved vectors of this conservation law,
and we will additionally identify elements from~$\CL(\mathcal L)$ with their representatives
in~$\CV(\mathcal L)$.
For $F\in\CV(\mathcal L)$ and ${\cal F}\in\CL(\mathcal L)$
the notation~$F\in {\cal F}$ will mean that $F$ is a conserved vector corresponding
to the conservation law~${\cal F}$.
In contrast to the order $r_F$ of a conserved vector~$F$ as the maximal order of derivatives explicitly appearing in~$F$,
the {\em order of the conservation law}~$\cal F$
is defined as $\min\{r_F\,|\,F\in{\cal F}\}$.
By linear dependence of conservation laws we mean linear dependence of them as elements of~$\CL(\mathcal L)$.
Therefore, in the framework of the ``representative'' approach
conservation laws of a system~$\mathcal L$ are considered {\em linearly dependent} if
there exists linear combination of their representatives which is a trivial conserved vector.

\subsection{Characteristics of conservation laws}\label{SectionOnCharacteristicsOfConsLaws}

Let the system~$\cal L$ be totally nondegenerate~\cite{Olver1986}.
Then an application of the Hadamard lemma to the definition of conserved vector and integration by parts imply that
the divergence of any conserved vector of~$\mathcal L$ can always be represented,
up to the equivalence relation of conserved vectors,
as a linear combination of the left hand sides of the independent equations from $\mathcal L$
with coefficient functions $\lambda^\mu$ on a suitable jet space~$J^{(k)}$:
\begin{equation}\label{CharFormOfConsLaw}
\mathop{\rm Div}\nolimits F=\lambda^\mu L^\mu.
\end{equation}
Here the order~$k$ is determined by~$\mathcal L$ and the allowable order of conservation laws,
$\mu=\overline{1,l}$.

\begin{definition}\label{DefCharForm}
Formula~\eqref{CharFormOfConsLaw} and the $l$-tuple $\lambda=(\lambda^1,\ldots,\lambda^l)$
are respectively called the {\it characteristic form} and the {\it characteristic}
of the conservation law associated with the conserved vector~$F$.
\end{definition}

The characteristic~$\lambda$ is {\em trivial} if it vanishes for all solutions of $\cal L$.
Since $\cal L$ is nondegenerate, the characteristics~$\lambda$ and~$\tilde\lambda$ satisfy~\eqref{CharFormOfConsLaw}
for the same~$F$ and, therefore, are called {\em equivalent}
iff $\lambda-\tilde\lambda$ is a trivial characteristic.
Similarly to conserved vectors, the set~$\Ch(\mathcal L)$ of characteristics
corresponding to conservation laws of the system~$\cal L$ is a linear space,
and the subset~$\Ch_0(\mathcal L)$ of trivial characteristics is a linear subspace in~$\Ch(\mathcal L)$.
The factor space~$\Ch_{\rm f}(\mathcal L)=\Ch(\mathcal L)/\Ch_0(\mathcal L)$
coincides with the set of equivalence classes of~$\Ch(\mathcal L)$
with respect to the above characteristic equivalence relation.

The following result~\cite{Olver1986} forms the cornerstone for the methods of studying conservation laws,
which are based on formula~\eqref{CharFormOfConsLaw}, including the Noether theorem and
the direct method in the version by
Anco and Bluman~\cite{Anco&Bluman2002a,Anco&Bluman2002b}.

\begin{theorem}[\cite{Olver1986}]\label{TheoremIsomorphismChCV}
Let~$\mathcal L$ be a normal, totally nondegenerate system of differential equations.
Then the representation of conservation laws of~$\mathcal L$ in the characteristic form~\eqref{CharFormOfConsLaw}
generates a one-to-one linear mapping between~$\CL(\mathcal L)$ and~$\Ch_{\rm f}(\mathcal L)$.
\end{theorem}

Using properties of total divergences, we can exclude the conserved vector~$F$ from~\eqref{CharFormOfConsLaw}
and obtain a condition for the characteristic~$\lambda$ only.
Namely, a differential function~$f$ is a total divergence, i.e. $f=\mathop{\rm Div} F$
for some $n$-tuple~$F$ of differential functions iff $\Eop(f)=0$.
Here, the Euler operator~$\Eop=(\Eop^1,\ldots, \Eop^m)$ is the $m$-tuple of differential operators
\[
{\Eop}^a=(-D)^\alpha\p_{u^a_\alpha}, \quad a=\overline{1,m},
\]
where $(-D)^\alpha=(-D_1)^{\alpha_1}\ldots(-D_m)^{\alpha_m}$, 
$\alpha=(\alpha_1,\ldots,\alpha_n)$ runs through the multi-index set (\mbox{$\alpha_i\!\in\!\mathbb{N}\cup\{0\}$}). 
Therefore, the action of the Euler operator on~\eqref{CharFormOfConsLaw}
results in the equation
\begin{equation}\label{NSCondOnChar}
\Eop(\lambda^\mu L^\mu)={\Fder}_\lambda^*(L)+{\Fder}_L^*(\lambda)=0,
\end{equation}
which is a necessary and sufficient condition on characteristics of conservation laws for the system~$\mathcal L$.
The matrix differential operators~${\Fder}_\lambda^*$ and~${\Fder}_L^*$ are the adjoints of
the Fr\'echet derivatives~${\Fder}_\lambda^{\phantom{*}}$ and~${\Fder}_L^{\phantom{*}}$, i.e.,
\[
{\Fder}_\lambda^*(L)=\left((-D)^\alpha\left( \dfrac{\p\lambda^\mu}{\p u^a_\alpha}L^\mu\right)\right), \qquad
{\Fder}_L^*(\lambda)=\left((-D)^\alpha\left( \dfrac{\p L^\mu}{\p u^a_\alpha}\lambda^\mu\right)\right).
\]
Since ${\Fder}_\lambda^*(L)=0$ automatically holds on solutions of~$\mathcal L$,
equation~\eqref{NSCondOnChar} implies a necessary condition for $\lambda$ to belong to~$\Ch(\mathcal L)$:
\begin{equation}\label{NCondOnChar}
{\Fder}_L^*(\lambda)\bigl|_{\mathcal L}=0.
\end{equation}
Condition~\eqref{NCondOnChar} can be considered as adjoint to the criterion
${\Fder}_L^{\phantom{*}}(\eta)\bigl|_{\mathcal L}=0$ for infinitesimal invariance of $\mathcal L$
with respect to an evolutionary vector field having the characteristic~$\eta=(\eta^1,\ldots,\eta^m)$.
This is why solutions of~\eqref{NCondOnChar} are sometimes called 
{\em cosymmetries} or {\em adjoint symmetries}.

\subsection{Equivalence of conservation laws with respect to transformation groups}\label{SectionOnEquivOfConsLaws}

We can substantially simplify and systematize the classification of conservation laws by additionally taking into account 
symmetry transformations of a system or equivalence transformations of a whole class of systems.
This problem is similar to that of group classification of differential equations. 
The following statement on transformations of equations in the conserved form is true
(see, e.g.,~\cite{Popovych&Ivanova2004ConsLawsLanl}).

\begin{proposition}\label{PropositionOnTransOfCLs}
Any point transformation~$\mathcal T$ maps a class of equations in the conserved form into itself.
More exactly, the transformation~$\mathcal T$: $\tilde x={\mathcal T}^x(x,u)$, $\tilde u={\mathcal T}^u(x,u)$ 
prolonged to the jet space~$J^{(r+1)}$ transforms the equation $D_iF^i=0$ to the equation $\tilde D_i\tilde F^i=0$. 
The transformed conserved vector~$\tilde F={\mathcal T}^F(x,u_{(r)},F)$ is determined
by the formula
\begin{equation}\label{eq.tr.var.cons.law}
\tilde F^i(\tilde x,\tilde u_{(r)})=\frac{D_{x_j}\tilde x_i}{|D_x\tilde x|}\,F^j(x,u_{(r)}),
\quad\mbox{i.e.}\quad
\tilde F(\tilde x,\tilde u_{(r)})=\frac{1}{|D_x\tilde x|}(D_x\tilde x)F(x,u_{(r)})
\end{equation}
in matrix notation. Here $|D_x\tilde x|$ is the determinant of the matrix $D_x\tilde x=(D_{x_j}\tilde x_i)$.
\end{proposition}

\begin{proof}
We give two equivalent versions of the proof. The first is direct and based on the usual definition of conservation laws. 
The second is closer to the formal definition and naturally involves the technique of differential forms.

We prolong the transformation~$\mathcal T$ to the jet space~$J^{(r+1)}$ in the standard way~\cite{Olver1986}, i.e., 
we recalculate all derivatives up to order $r+1$ in the new (`tilde') variables:
$\tilde u_{(r+1)}=\mathrm{pr}_{(r+1)}{\mathcal T}^u(x,u_{(r+1)})$. 
Since the equation $D_iF^i=0$ is linear in~$F$, 
the transformation for the tuple~$F$ is found in the form $\tilde F^i=G^{ij}F^j$. 
The smooth functions~$G^{jk}$ of $x$ and the derivatives of $u$ should be selected by 
the condition $\tilde D_i\tilde F^i=\Lambda D_iF^i$, 
where $\Lambda$ also is a smooth function of $x$ and the derivatives of $u$. 
Let $\hat G=(\hat G^{ij})$ be the inverse matrix to the matrix $G=(G^{ij})$, $\hat\Lambda=1/\Lambda$.
Then 
\[
D_jF^j=D_j(\hat G^{jk}\tilde F^k)
=(D_j\hat G^{jk})\tilde F^k+(D_j\tilde x_i)\hat G^{jk}\tilde D_i\tilde F^k=\hat\Lambda\tilde D_i\tilde F^i
\]
for any~$\tilde F$ iff $(D_j\tilde x_i)\hat G^{jk}=\hat\Lambda\delta_{ik}$, $D_j\hat G^{jk}=0$. 
Here $\delta_{ik}$ is the Kronecker delta. 
The first set of equations on $\hat G^{jk}$ implies that $\hat G=\hat\Lambda(D_x\tilde x)^{-1}=\hat\Lambda(D_{\tilde x}x)$.
Substituting these expressions for $\hat G^{jk}$ into the second set of equations we get:
\[
D_j\hat G^{jk}=D_j(\hat\Lambda\tilde D_kx_j)=D_j(\hat\Lambda)\tilde D_kx_j
+\hat\Lambda (D_j\tilde x_{k'})\tilde D_{k'\!}\tilde D_kx_j=0.
\] 
(Note that $D_j=(D_j\tilde x_{k'})\tilde D_{k'\!}$.) 
After dividing the result by~$\hat\Lambda$ and convolving it with $D_i\tilde x_k$, we obtain
\begin{gather*}
\frac{D_i\hat\Lambda}{\hat\Lambda}
=-(D_i\tilde x_k)(D_j\tilde x_{k'})\tilde D_{k'\!}\tilde D_kx_j
=-(D_j\tilde x_{k'})(D_i\tilde x_k)\tilde D_k\tilde D_{k'\!}x_j
=-(D_j\tilde x_{k'})D_i\tilde D_{k'\!}x_j\\
\phantom{\frac{D_i\hat\Lambda}{\hat\Lambda}}
=-\mathop{\rm tr}\nolimits\bigl((D_x\tilde x)D_i(D_x\tilde x)^{-1}\bigr)
=\mathop{\rm tr}\nolimits\bigl((D_iD_x\tilde x)(D_x\tilde x)^{-1}\bigr)
=\frac{D_i|D_x\tilde x|}{|D_x\tilde x|},
\end{gather*}
where we have used the commutation property of the total derivative operators and 
the well-known equalities for matrix derivatives $(A^{-1})'=-A^{-1}A'A^{-1}$ and 
$\mathop{\rm tr}\nolimits(A'A^{-1})=|A|'/|A|$. 
Here the prime denotes the derivative with respect to a parameter. 
$\mathop{\rm tr}\nolimits A$ and $|A|$ are the trace and the determinant of a square matrix~$A$, respectively. 
The above equations imply that $\hat\Lambda=|D_x\tilde x|$ up to an arbitrary nonzero constant multiplier 
(which is inessential) and, therefore, $\hat G=|D_x\tilde x|(D_x\tilde x)^{-1}$, i.e., 
$G=|D_x\tilde x|^{-1}D_x\tilde x$. 

The second version of the proof is much simpler and indeed justifies the first one. 
We associate any tuple~$F$ with the differential form 
$\omega_F=(-1)^{i-1}F^i\,dx_1\wedge\dots\wedge \lefteqn{\smash{\,\diagdown}}dx_i\wedge\dots\wedge dx_n$ 
in the `horizontal' de Rahm complex~\cite{Bocharov&Co1997} called also $\sf D$-complex~\cite{Olver1986}
over the space of the independent variable~$x$ and the dependent variable~$u$.
Hereafter the~notation $\lefteqn{\smash{\,\diagdown}}dx_i$ means that the term~$dx_i$ is absent in the 
corresponding external product. 
In the $\sf D$-complex the differential of the usual de Rahm complex is replaced 
by the total differential $\sf D$. 
Due to the invariance of differential forms under transformations of variables,
\begin{gather*}
\omega_{\tilde F}=(-1)^{i-1}\tilde F^i\,
d\tilde x_1\wedge\dots\wedge \lefteqn{\smash{\,\diagdown}}d\tilde x_i\wedge\dots\wedge d\tilde x_n\\
\phantom{\omega_{\tilde F}}=
(-1)^{i-1}\tilde F^iM^{ij}\,dx_1\wedge\dots\wedge \lefteqn{\smash{\,\diagdown}}dx_j\wedge\dots\wedge dx_n
=\omega_F,
\end{gather*}
where $M^{ij}$ is the minor of the element $D_j\tilde x_i$ in the matrix~$D_x\tilde x$.
Therefore, $F^j=(-1)^{i+j}M^{ij}\tilde F^i$, i.e., $\tilde F=|D_x\tilde x|^{-1}(D_x\tilde x)F$.
Applying the total differential, we also have 
$\mathop{\tilde{\sf D}}\nolimits\omega_{\tilde F}=\mathop{{}\sf D}\nolimits\omega_F$, where 
\begin{gather*}
\mathop{{}\sf D}\nolimits{}\omega_F=(D_jF^j)x_1\wedge\dots\wedge dx_n,
\\ 
\mathop{\tilde{\sf D}}\nolimits\omega_{\tilde F}=(\tilde D_i\tilde F^i)\,d\tilde x_1\wedge\dots\wedge d\tilde x_n
=(\tilde D_i\tilde F^i)|D_x\tilde x|\,dx_1\wedge\dots\wedge dx_n,
\end{gather*}
i.e., $\tilde D_i\tilde F^i=|D_x\tilde x|^{-1}D_jF^j$.
\end{proof}

\begin{note}
In the case of one dependent variable ($m=1$) $\mathcal T$ can be a contact transformation:
$\tilde x={\mathcal T}^x(x,u_{(1)})$, $\tilde u_{(1)}=\mathrm{pr}_{(1)}{\mathcal T}^u(x,u_{(1)})$. 
The proof is entirely analogous.
Similar remarks apply to the statements below.
\end{note}

\begin{note}
After~\cite{Popovych&Ivanova2004ConsLawsLanl}, formula~\eqref{eq.tr.var.cons.law} and related results 
were also presented in \cite{Bluman2005,Bluman&Temuerchaolu&Anco2006}. 
The above proofs of Proposition~\ref{PropositionOnTransOfCLs} essentially differ from the analogous proof 
in \cite{Bluman&Temuerchaolu&Anco2006} in that the formula for transformations of conserved vectors 
is derived during the proof while in \cite{Bluman&Temuerchaolu&Anco2006} it was proved a posteriori. 
\end{note}

\begin{definition}
Let $G$ be a symmetry group of the system~$\mathcal L$.
Two conservation laws with the conserved vectors $F$ and $F'$ are called {\em $G$-equivalent} if
there exists a transformation $\mathcal T\in G$ such that the conserved vectors $\tilde F={\mathcal T}^F$ and $F'$
are equivalent in the sense of Definition~\ref{DefinitionOfConsVectorEquivalence}.
\end{definition}

Any transformation $\mathcal T\in G$ induces a linear one-to-one mapping $\mathcal T_*$ in~$\CV(\mathcal L)$,
transforms trivial conserved vectors only to trivial ones
(i.e. $\CV_0(\mathcal L)$ is invariant with respect to~$\mathcal T_*$)
and therefore induces a linear one-to-one mapping $\mathcal T_{\rm f}$ in~$\CL(\mathcal L)$.
It is obvious that $\mathcal T_{\rm f}$ preserves linear (in)dependence of elements
in~$\CL(\mathcal L)$ and maps a basis (a set of generators) of~$\CL(\mathcal L)$
in a basis (a set of generators) of the same space.
In this way we can consider the \emph{$G$-equivalence relation of conservation laws}
as well-defined on~$\CL(\mathcal L)$ and use it to classify conservation laws.

\begin{proposition}\label{PropositionOnInducedMappingOfCL}
Any point transformation $\mathcal T$ between systems~$\mathcal L$ and~$\tilde{\mathcal L}$
induces a linear one-to-one mapping $\mathcal T_*$ from~$\CV(\mathcal L)$ onto~$\CV(\tilde{\mathcal L})$,
which maps $\CV_0(\mathcal L)$ into~$\CV_0(\tilde{\mathcal L})$
and generates a linear one-to-one mapping $\mathcal T_{\rm f}$ from~$\CL(\mathcal L)$ onto~$\CL(\tilde{\mathcal L})$.
\end{proposition}

\begin{corollary}\label{CorollaryOnInducedMappingOfChar}
Any point transformation $\mathcal T$ between systems~$\mathcal L$ and~$\tilde{\mathcal L}$
induces a linear one-to-one mapping from~$\Ch_{\rm f}(\mathcal L)$ onto~$\Ch_{\rm f}(\tilde{\mathcal L})$.
\end{corollary}

It is possible to obtain an explicit formula for the correspondence between characteristics of~$\mathcal L$ and~$\tilde{\mathcal L}$. 
This formula obviously depends on representation of~$\mathcal L$ and~$\tilde{\mathcal L}$ as systems of differential equations. 

\begin{proposition}\label{PropositionOnInducedMappingOfChars}
Let $\mathcal T$ be a point transformation of a system~$\mathcal L$ to a system~$\tilde{\mathcal L}$ and 
$\tilde L^\mu=\Lambda^{\mu\nu}L^\nu$,
where $\Lambda^{\mu\nu}=\Lambda^{\mu\nu\alpha}D^\alpha$, $\Lambda^{\mu\nu\alpha}$ are differential functions, and
$\alpha=(\alpha_1,\ldots,\alpha_n)$ runs through the multi-indices ($\alpha_i\!\in\!\mathbb{N}\cup\{0\}$),
$\mu,\nu=\overline{1,l}$. (The number of $\alpha's$ for which $\Lambda^{\mu\nu\alpha}\ne0$ is finite.)
Then the transformation~$\mathcal T$ induces the linear one-to-one mapping from~$\Ch(\mathcal L)$
onto~$\Ch(\tilde{\mathcal L})$, the inverse of which is defined by the formula 
\[
\lambda^\mu={\Lambda^{\mu\nu}}^*(|D_x\tilde x|\tilde\lambda^\nu).
\]
Here ${\Lambda^{\mu\nu}}^*=(-D)^\alpha\cdot\Lambda^{\nu\mu\alpha}$ is the adjoint to the operator~$\Lambda^{\nu\mu}$.
\end{proposition}

\begin{proof}
By the definition of characteristics, for any $\tilde\lambda\in\Ch(\tilde{\mathcal L})$ there exists $\tilde F\in\CV(\tilde{\mathcal L})$ 
such that $\tilde\lambda^\mu\tilde L^\mu=\tilde D^i\tilde F^i$. 
We take the preimage $F\in\CV(\mathcal L)$ of~$\tilde F$ with respect to the mapping induced by~$\mathcal T$. 
Then 
\[
D_jF^j=|D_x\tilde x|\,\tilde D_i\tilde F^i=|D_x\tilde x|\,\tilde\lambda^\mu\Lambda^{\mu\nu}L^\nu=
\Lambda^{\nu\mu*}(|D_x\tilde x|\,\tilde\lambda^\mu)L^\nu+D_i\hat F^i=\lambda^\mu L^\mu+D_i\hat F^i, 
\]
where $\lambda^\mu={\Lambda^{\mu\nu}}^*(|D_x\tilde x|\tilde\lambda^\nu)$ and 
each $\hat F^i$ vanishes on the solutions of~$\mathcal L$, i.e., 
the tuple $\hat F=(\hat F^1,\dots,\hat F^n)$ is a trivial conserved vector of~$\mathcal L$. 
It means that the tuple $\lambda=(\lambda^1,\dots,\lambda^l)$ is a characteristic of the system~$\mathcal L$, 
associated with the conservation law containing the preimage of the conserved vector $\tilde F$.
Since the matrix-operator $\Lambda=(\Lambda^{\mu\nu})$ has an inverse which is of a similar form,  
the induced mapping of characteristics is one-to-one. The linearity of this mapping is obvious.
\end{proof}

\begin{note}
$\Lambda^{\mu\nu\alpha}=0$ for $|\alpha|>0$ in a number of cases, 
e.g., if~$\mathcal L$ and~$\tilde{\mathcal L}$ are single partial differential equations ($l=1$).
Then the operators~$\Lambda^{\mu\nu}$ are simply differential functions
(more precisely, the operators of multiplication by differential functions) and, therefore, ${\Lambda^{\mu\nu}}^*=\Lambda^{\nu\mu}$.
\end{note}

Consider the class~$\mathcal L|_{\cal S}$ of systems~$\mathcal L_\theta$: $L(x,u_{(\rho)},\theta(x,u_{(\rho)}))=0$
parameterized with the parameter-functions~$\theta=\theta(x,u_{(\rho)}).$
Here $L$ is a tuple of fixed functions of $x,$ $u_{(\rho)}$ and $\theta.$
$\theta$~denotes a tuple of arbitrary (parametric) functions
$\theta(x,u_{(\rho)})=(\theta^1(x,u_{(\rho)}),\ldots,\theta^k(x,u_{(\rho)}))$
running through the solution set~${\cal S}$ of the system~$S(x,u_{(\rho)},\theta_{(q)}(x,u_{(\rho)}))=0$.
This system consists of differential equations for $\theta$,
where $x$ and $u_{(\rho)}$ play the role of independent variables
and $\theta_{(q)}$ stands for the set of all the partial derivatives of $\theta$ of order not greater than $q$.
Sometimes the set $\mathcal{S}$ is additionally constrained by the non-vanishing condition 
$S'(x,u_{( p)},\theta_{(q)}(x,u_{( p)}))\ne0$ with another tuple $S'$ of differential functions.
(See also~\cite{Popovych2006c} for other nuances in the rigorous definition of classes of differential equations.)
In what follows we call the functions $\theta$ arbitrary elements.
Denote the point transformation group preserving the
form of the systems from~$\mathcal L|_{\cal S}$ as $G^\sim=G^\sim(L,S).$

Let $P=P(L,S)$ denote the set of all pairs each of which consists of
a system $\mathcal L_\theta$ from~$\mathcal L|_{\cal S}$ and a conservation law~${\cal F}$ of this system.
In view of Proposition~\ref{PropositionOnInducedMappingOfCL},
the action of transformations from~$G^\sim$ on $\mathcal L|_{\cal S}$ and
$\{\CV(\mathcal L_{\theta})\,|\,\theta\in{\cal S}\}$
together with the pure equivalence relation of conserved vectors
naturally generates an equivalence relation on~$P$.

\begin{definition}
Let $\theta,\theta'\in{\cal S}$,
${\cal F}\in\CL(\mathcal L_\theta)$, ${\cal F}'\in\CL(\mathcal L_{\theta'})$,
$F\in{\cal F}$, $F'\in{\cal F'}$.
The pairs~$(\mathcal L_\theta,{\cal F})$ and~$(\mathcal L_{\theta'},{\cal F'})$
are called {\em $G^\sim$-equivalent} if there exists a transformation $\mathcal T\in G^\sim$
which transforms the system~$\mathcal L_\theta$ to the system~$\mathcal L_{\theta'}$ and
such that the conserved vectors $\tilde F=\mathcal T^F$ and $F'$
are equivalent in the sense of Definition~\ref{DefinitionOfConsVectorEquivalence}.
\end{definition}

The classification of conservation laws with respect to~$G^\sim$ will be understood as
classification in~$P$ with respect to the above equivalence relation.
This problem can be investigated in a way similar to group classification in classes
of systems of differential equations, especially if it is formulated in terms of characteristics.
Namely, we construct firstly the conservation laws that are defined for all values of the arbitrary elements.
(The corresponding conserved vectors may depend on the arbitrary elements.)
Then we classify, with respect to the equivalence group, the arbitrary elements for each of which the system
admits additional conservation laws. 

In an analogous way we also can introduce equivalence relations on~$P$ which are
generated either by generalizations of usual equivalence groups or
by all admissible point or contact transformations
(also called form-preserving in~\cite{Kingston&Sophocleous1998})
in pairs of equations from~$\mathcal L|_{\cal S}$.

\begin{note}
It can easily be shown that all the above equivalences are indeed equivalence relations
(i.e., are reflexive, symmetric and transitive).
\end{note}

\subsection{Action of symmetry operators on conservation laws}\label{SectionOnSymOpsAndConsLaws}

If the system~$\mathcal L$ admits a one-parameter group of transformations 
then the infinitesimal generator $Q=\xi^i\p_i+\eta^a\p_{u^a}$
of this group can be used for the construction of new conservation laws from known ones.
Namely, differentiating equation~(\ref{eq.tr.var.cons.law})
with respect to the parameter $\varepsilon$ and inserting the value $\varepsilon=0$,
we obtain a new conserved vector expressed via the coefficients of the operator~$Q$ and the known conserved vector.

\begin{proposition}
If $Q=\xi^i\p_i+\eta^a\p_{u^a}$ is a Lie symmetry operator of the system~$\mathcal L$ and 
$F\in\CV(\mathcal L)$ then the differential functions
\begin{equation}\label{eq.inf.tr.var.cons.law}
\widetilde F^i=-Q_{(r)}F^i+(D_j\xi^i)F^j-(D_j\xi^j)F^i
\end{equation}
also are components of a conserved vector of~$\mathcal L$.
Here $Q_{(r)}$ denotes the $r$-th prolongation~\cite{Olver1986,Ovsiannikov1982} of the operator $Q$.
\end{proposition}

\begin{note}
In contrast to formula~(\ref{eq.tr.var.cons.law}), 
formula~\eqref{eq.inf.tr.var.cons.law} is well-known and extends directly to generalized symmetries.
See, for example,~\cite{Ibragimov1985,Khamitova1982,Olver1986} for generalized symmetries 
in the evolutionary form ($\xi^i=0$) 
and~\cite{Caviglia1986} for the general case. 
Below we show that in fact it is enough to restrict oneself to the version of the formula for the evolutionary form of symmetries.
It was used in~\cite{Khamitova1982} to introduce a notion of basis of conservation laws as a set
which generates a whole space of conservation laws by the action of generalized symmetry operators 
and the operation of linear combination. 
Here we give formula~\eqref{eq.inf.tr.var.cons.law} only 
through its connection with formula~(\ref{eq.tr.var.cons.law}).
\end{note}

There is a well-defined equivalence relation on the space $\mathrm{GS}(\mathcal L)$ of generalized symmetries of 
a system~$\mathcal L$ of differential equations~\cite{Olver1995}. 
Namely, generalized symmetry operators $Q$ and $Q'$ of the system~$\mathcal L$ are 
called equivalent if the difference of their evolutionary forms vanishes on the solutions of~$\mathcal L$. 
The corresponding factor-space will be denoted by $\mathrm{GS}_{\mathrm f}(\mathcal L)$.
The equivalence relation on $\mathrm{GS}(\mathcal L)$ agrees with 
the equivalence relation on $\CV(\mathcal L)$ in view of the following statement. 
 
\begin{proposition}
The action of equivalent generalized symmetry operators on equivalent conserved vector 
generates equivalent conserved vectors.
\end{proposition}

\begin{proof}
We show at first that a generalized symmetry operator~$Q$ and its evolutionary form~$\hat Q$ 
generate equivalent conserved vectors, acting on the same conserved vector. 
Indeed, since $Q_{(\infty)}=\hat Q_{(\infty)}+\xi^i D_i$ then 
\begin{gather*}
\widetilde F^i=-\hat Q_{(\infty)}F^i+(D_j\xi^i)F^j-(D_j\xi^j)F^i-\xi^j D_jF^i
\\\phantom{\widetilde F^i}
=-\hat Q_{(r)}F^i+D_j(\xi^iF^j-\xi^jF^i)-\xi^i D_jF^j
=-\hat Q_{(r)}F^i+\check F^i+\hat F^i, 
\end{gather*}
where $\widetilde F^i$ are components of the conserved vector $\widetilde F=Q[F]$ defined by~\eqref{eq.inf.tr.var.cons.law}, and
$Q_{(\infty)}$ and $\hat Q_{(\infty)}$ denote the formal infinite prolongations of the operators~$Q$ and~$\smash{\hat Q}$, respectively. 
The differential functions $\check F^i=D_j(\xi^iF^j-\xi^jF^i)$ are the components of a null divergence and 
$\hat F^i=-\xi^i D_jF^j$ vanish on the solutions of~$\mathcal L$. 
Therefore, the conserved vectors $\smash{\widetilde F}$ and $\hat Q[F]=-\hat Q_{(r)}F$ are equivalent.

If the difference of generalized symmetry operators in the evolutionary form vanishes on the solutions of~$\mathcal L$, 
the difference of their actions on a conserved vector obviously have the same property. 

The action of any generalized symmetry operator in evolutionary form on a trivial conserved vector 
results in a trivial conserved vector. 
Indeed, consider an operator $\hat Q\in\mathrm{GS}(\mathcal L)$ in evolutionary form of order~$\hat r$
and a trivial conserved vector~$F$ of the system~$\mathcal L$, i.e., 
$F=\hat F+\check F$, where $\hat F\bigl|_{\mathcal L}=0$ and $\mathop{\rm Div}\nolimits\check F=0$.
In view of the Hadamard lemma and the condition $\hat F\bigl|_{\mathcal L}=0$, each component of~$\hat F$ 
is presented in the form $\hat F^i=\lambda^{i\mu} L^{r\mu}$. 
Here $r$ is the order of $\hat F$,
$\lambda^{i\mu}$ and $L^{r\mu}$, $\mu=1,\dots,l_r$, are smooth functions in the jet space~$J^{(r)}$,
the tuple $(L^{r1},\dots,L^{rl_r})$ determines the manifold $\mathcal L_{(r)}$ in~$J^{(r)}$. 
Then $\hat Q_{(r)}\hat F^i=L^{r\mu}\hat Q_{(r)}\lambda^{i\mu}+\lambda^{i\mu}\hat Q_{(r)}L^{r\mu}=0$ 
on $\mathcal L_{(r')}$, where $r'=r+\hat r$, since $\hat Q\in\mathrm{GS}(\mathcal L)$. 
This means that $\hat Q[\hat F]\bigl|_{\mathcal L}=0$.
The components of the null divergence~$\check F$ are represented in the form $\check F^i=D_jv^{ij}$, 
where $v^{ij}$ are smooth functions of $x$ and derivatives of $u$ such that $v^{ij}=-v^{ji}$. 
The equality 
\[
\hat Q[\check F]=-\hat Q_{(\infty)}D_jv^{ij}=-D_j\hat Q_{(\infty)}v^{ij}=-D_j\hat v^{ij}
\] 
implies that $\hat Q[\check F]$ also is a null divergence since 
$\hat v^{ij}=\hat Q_{(\infty)}v^{ij}$ are differential functions and $\hat v^{ij}=-\hat v^{ji}$.
Therefore, $\hat Q[F]=\hat Q[\hat F]+\hat Q[\check F]$ is a trivial conserved vector as the sum of 
the tuple $\hat Q[\hat F]$ vanishing on solutions of~$\mathcal L$ and the null divergence~$\hat Q[\check F]$.

Let $Q,Q'\in\mathrm{GS}(\mathcal L)$, $F,F'\in\CV(\mathcal L)$, $Q\sim Q'$ and $F\sim F'$. 
Then $Q[F]\sim Q'[F']$ since 
\[
Q'[F']-Q[F]=(Q'-\hat Q')[F']-(Q-\hat Q)[F]+(\hat Q'-\hat Q)[F']+\hat Q[F'-F]
\]
is a trivial conserved vector of the system~$\mathcal L$ in view of the above considerations.
\end{proof}

\begin{corollary}
For any system~$\mathcal L$
formula~\eqref{eq.inf.tr.var.cons.law} gives a well-defined action of elements from $\mathrm{GS}_{\mathrm f}(\mathcal L)$ 
on conservation laws of~$\mathcal L$.
\end{corollary}

That is why formula~\eqref{eq.inf.tr.var.cons.law} is usually presented only for generalized symmetry operators 
in the evolutionary form ($\xi^i=0$) 
and only inequivalent generalized symmetry operators should be used 
to generate new conservation laws from known ones. 
Note additionally that the application of formula~\eqref{eq.inf.tr.var.cons.law} does not guarantee 
the construction of nontrivial conserved vectors from nontrivial ones~\cite{Ibragimov1985,Khamitova1982,Olver1986}.

\subsection{Potential systems}\label{SectionOnPotSystems}

If the local conservation laws of a system~$\mathcal L$ of differential equations are known,
we can apply Lemma~\ref{lemma.null.divergence} to conservation laws constructed 
on the set of solutions of~\mbox{$\mathcal L={\mathcal L}^0$}. 
In this way we introduce potentials as additional dependent variables. Then we attach the equations 
connecting the potentials with the components of the corresponding conserved vectors to~${\mathcal L}^0$.
If~\mbox{$n>2$} the attached equations of this kind form an underdetermined system with respect to the potentials.
Therefore, we can also add gauge conditions on the potentials to~${\mathcal L}^0$. 
In fact, such additional conditions are absolutely necessary in the case \mbox{$n>2$}. 
It was proved in Theorem 2.7 of~\cite{Anco&Bluman1997} for a quite general situation that 
every local symmetry of a potential system with potentials which are not additionally constrained 
is projectable to a local symmetry of the initial system, i.e., such a potential system gives no nontrivial potential symmetries. 
Moreover, each conservation law of such a system is invariant with respect to gauge transformations of the potentials \cite{Anco&The2005}.

We have to use linearly independent conservation laws since otherwise the introduced potentials will be
{\em dependent} in the following sense: there exists a linear combination of the potential tuples,
which is, for some $r'\in{\mathbb N}$, a tuple of functions of $x$ and $u_{(r')}$ only.

Then we exclude the superfluous equations (i.e., the equations that are dependent on
equations from~${\mathcal L}^0$ and the attached equations simultaneously)
from the extended (potential) system~${\mathcal L}^1$,
which will be called a {\em potential system of the first level}.
Any conservation law of~${\mathcal L}^0$ is one of~${\mathcal L}^1$.
We iterate the above procedure for~${\mathcal L}^1$ to find its conservation laws
which are linearly independent with those from the previous iteration
and will be called {\em potential conservation laws of the first level}.

We continue this process as long as possible
(i.e., the iteration procedure has to be stopped if all the conservation laws of
a {\em potential system~${\mathcal L}^{k+1}$ of the $(k+1)$-st level} are linearly dependent
with the ones of~${\mathcal L}^k$) or inductively construct infinite chains of conservation laws.
This procedure may yield {\em purely potential} conservation laws of the initial system~$\mathcal L$,
which are linearly independent with local conservation laws and depend explicitly on potential variables. 
The idea of this iteration procedure can be traced back to the well-known paper 
by Wahlquist and Estabrook~\cite{Wahlquist&Estabrook1975}. 

Any conservation law from the previous step of the iteration procedure will be a conservation law for the next step. 
Conservation laws which are obtained on the next step
and depend only on variables of the previous step are linearly dependent with
conservation laws from the previous step.
It is also obvious that the conservation laws used for the construction of a potential system of the next level are
trivial on the manifold of this system. 

Since gauge conditions on potentials can be chosen in many different ways, an
exhaustive realization of the above iteration procedure is improbable if $n>2$.

The best way to calculate conservation laws on each level is to apply the direct method of finding conservation laws. 
One can distinguish four versions of this method depending on what condition on conservation laws 
(\eqref{EqDefOfConservedVector}, \eqref{CharFormOfConsLaw}, \eqref{NSCondOnChar} or \eqref{NCondOnChar})
is taken as a basis for performing calculations 
\cite{Anco&Bluman2002a,Anco&Bluman2002b,Bocharov&Co1997,Popovych&Ivanova2004NVCDCEs,Wolf2002}. 
Each of these four versions of the direct method has its advantages and disadvantages
in applications and concerning implementation in computer algebra programs~\cite{Wolf2002}. 
The version involving characteristics is close to the the symmetry group method by Noether, 
which is applicable only in the case of Euler--Lagrange equations and is effective if 
the generalized symmetry algebra of the system under consideration is already known.

The case of two independent variables is distinguished by the possible (constant) indeterminacy
after the introduction of potentials and also by the high effectiveness of the application of potential symmetries.
This is why we consider some notions connected with conservation laws in this case separately.
We denote the independent variables as $t$ (the time variable) and $x$ (the space variable).
Any local conservation law has the form
\begin{equation}\label{EqGen2DimConsLaw}
D_tF(t,x,u_{(r)})+D_xG(t,x,u_{(r)})=0,
\end{equation}
where $D_t$ and $D_x$ are the operators of total differentiation with respect to $t$ and $x$.
The components $F$ and $G$ of the conserved vector~$(F,G)$ are called
the {\em density} and the {\em flux} of the conservation law, respectively.
The conservation law allows us to introduce the new dependent (potential) variable~$v$
by means of the equations
\begin{equation}\label{potsys1}
v_x=F,\quad v_t=-G
\end{equation}
determining~$v$ up to a constant summand. 

If $\mathcal L$ is a single equation, it is a differential consequence of~\eqref{potsys1} in the case of a nonsingular characteristic 
and equations of form~\eqref{potsys1} combine into the complete potential system.
As a rule, systems of this kind admit a number of nontrivial symmetries and so they are of great interest.

Introducing a number of potentials for an iteration step in the case of two independent variables,
we can use the notion of potential dependence, which is more general than the one based on linear dependence 
of conservation laws.

\begin{definition}\label{DefinitionOfPotentialDependence}
The potentials $v^1$, \ldots, $v^p$ are called
{\em locally dependent on the set of solutions of the system~${\mathcal L}$} (or, briefly speaking, {\em dependent})
if there exist $r'\in{\mathbb N}$ and a function~$H$ of the variables $t$, $x$, $u_{(r')}$, $v^1$, \ldots, $v^p$
such that $H_{v^s}\ne0$ for some~$s$ and 
$H(t,x,u_{(r')},v^1,\ldots,v^p)=0$ for any solution $(u,v^1,\ldots,v^p)$ of the total system determining
the set of potentials~$v^1$, \ldots, $v^p$ (up to gauge transformations, i.e., up to adding negligible constants to the potentials).
\end{definition}

The proof of local dependence or independence of potentials for general classes of differential equations is difficult since
it is closely connected with a precise description of the possible structure of conservation laws.

Proposition~\ref{PropositionOnInducedMappingOfCL} and equation~\eqref{potsys1} imply 
the following statement~\cite{Popovych&Ivanova2004ConsLawsLanl}.

\begin{proposition}\label{PropositionOn2DConsLawEquivRelation}
Any point transformation connecting two systems~$\mathcal L$ and~$\tilde{\mathcal L}$
of PDEs with two independent variables generates a one-to-one mapping between the sets of potential systems
corresponding to~$\mathcal L$ and~$\tilde{\mathcal L}$. This mapping is induced by trivial prolongation
on the space of the introduced potential variables, i.e. one can assume that the potentials are not transformed.
\end{proposition}

\begin{corollary}\label{CorollaryOn2DPotSystemsEquivRelGeneratedBySymGroup}
The Lie symmetry group of a system~$\mathcal L$ of differential equations generates an equivalence group
on the set of potential systems corresponding to~$\mathcal L$.
\end{corollary}

\begin{corollary}
Let $\widehat{\mathcal L}|_S$ be the set of all potential systems constructed
for systems from the class~$\mathcal L|_S$ with their conservation laws.
The action of transformations from~$G^\sim(L,S)$ together with the equivalence relation of potentials
naturally generates an equivalence relation on~$\widehat{\mathcal L}|_S$.
\end{corollary}

\begin{note}
Proposition~\ref{PropositionOn2DConsLawEquivRelation} and its corollaries imply that the equivalence group for a class of
systems or the symmetry group for a single system can be prolonged to the potential variables for any step of
the direct iteration procedure. It is natural that the prolonged equivalence groups and symmetry groups are used to classify
possible conservation laws, potential systems and potential symmetries in each iteration.
\end{note}

\begin{definition}\label{DefinitionOfPotSyms}
Every Lie symmetry of a potential system is called a \emph{potential symmetry} of the initial system. 
A potential symmetry operator is called nontrivial if it is not projectable to the space of independent 
and (original) dependent variables, 
i.e., if some of the coefficients corresponding to these variables explicitly depend on potentials. 
\end{definition}

The notion of generalized (resp.\ nonclassical, resp.\ conditional, resp.\ approximate, etc.) potential symmetry is defined in an analogous way. 

Each tuple of independent conservation laws generates, in fact, an infinite series of potential systems 
associated with equivalent tuples of conserved vectors. 
These systems are connected via transformations of potentials 
of the form $\tilde v^s=v^s+H^s$, $s=1,\dots,p$, where $p$ is the number of conservation laws in the tuple 
and the $H^s$ are functions of $x$ and derivatives of~$u$. 
Therefore, they are equivalent in the investigation of generalized symmetries of arbitrary orders.
At the same time, the choice of representatives in sets of equivalent conserved vectors becomes significant if 
one considers generalized symmetries of a fixed order, e.g., in the case of Lie symmetries having order 0. 
For some classes of equations, a favorable choice can explicitly be given.

\section{Local conservation laws}\label{SectionOnLocalCLsOfLPEs}

We look for (local) conservation laws of equations from class~\eqref{EqGenLPE},
applying the modification of the direct method which was proposed in~\cite{Popovych&Ivanova2004ConsLawsLanl}.
Since equation~\eqref{EqGenLPE} is two-dimensional, the constructed conservation laws will have the general form~\eqref{EqGen2DimConsLaw}.

At first we prove a lemma on the order of local conservation laws for equations from class~\eqref{EqGenLPE}.

\begin{lemma}\label{LemmaOnOrderOfConsLawsOfLPEs}
Any local conservation law of any equation from class~\eqref{EqGenLPE} is of first order
and, moreover, it possesses a conserved vector with density depending at most on $t$, $x$, and $u$
and flux depending at most on $t$, $x$, $u$ and~$u_x$. 
\end{lemma}

\begin{proof}
Consider a conservation law~\eqref{EqGen2DimConsLaw} of an equation of form~\eqref{EqGenLPE}.
In view of equation~\eqref{EqGenLPE} and its differential consequences,
we can assume that $F$ and~$G$ depend only on $t$, $x$ and $u_k=\partial^k u/\partial x^k$, $k=\overline{0,r'},$
where $r'\leqslant 2r$. Suppose that $r'>1$.
We~expand the total derivatives in (\ref{EqGen2DimConsLaw}) and take into account differential consequences of the
form $u_{tk}=D_x^k(Au_{xx}+Bu_x+Cu)$, where $u_{tk}=\partial^{k+1} u/\partial t\partial x^k$, $k=\overline{0,r'}$.
As a result, we obtain the following condition
\begin{equation}\label{clcdeom}\textstyle
F_t+F_{u_k}D_x^k(Au_{xx}+Bu_x+Cu)+G_x+G_{u_k}u_{k+1}=0.
\end{equation}
Let us decompose~\eqref{clcdeom} with respect to the highest derivatives $u_k$.
Thus, the coefficients of $u_{r'+2}$ and $u_{r'+1}$ give the equations
$F_{u_{r'}}=0$, $G_{u_{r'}}+AF_{u_{r'-1}}=0$, which implies
\[
F=\hat F, \quad G=-A\hat F_{u_{r'-1}}u_{r'}+\hat G,
\]
where $\hat F$ and $\hat G$ are functions of $t$, $x$, $u$, $u_1$, \ldots, $u_{r'-1}$.
Then, after selecting the terms containing $u_{r'}^2$, we obtain that $-A\hat F_{u_{r'-1}u_{r'-1}}=0$.
It follows that $\hat F =\check F^1u_{r'-1}+\check F^0,$
where $\check F^1$ and $\check F^0$ depend only on $t$, $x$, $u$, $u_1$,~\ldots, $u_{r'-2}$.

Consider the conserved vector with density~$\tilde F=F-D_xH$ and flux~$\tilde G=G+D_tH$,
where $H=\int \check F^1du_{r'-2}$. This conserved vector is equivalent to the initial one, and
\[
\tilde F=\tilde F(t,x,u,u_1,\ldots,u_{r'-2}), \quad
\tilde G=\tilde G(t,x,u,u_1,\ldots,u_{r'-1}).
\]

Iterating the above procedure a suitable number of times, we obtain an equivalent conserved vector
depending only on $t$, $x$, $u$ and $u_x$, i.e., we can assume at once that $r'\leqslant 1$.
Then the coefficients of $u_{xxx}$ and $u_{xx}$ in~\eqref{clcdeom} lead to the equations
$F_{u_x}=0$, $G_{u_x}+AF_u=0$, implying $F=F(t,x,u)$ and, moreover, $G=-AF_uu_x+G^1$,
where $G^1=G^1(t,x,u)$.
\end{proof}

\begin{note}
A similar statement is true for an arbitrary (1+1)-dimensional evolution equation~$\cal L$ of even
order~$r=2\bar r$, $\bar r\in\mathbb{N}$. For example~\cite{Ibragimov1985}, for any conservation law of~$\cal L$
we can assume up to equivalence of conserved vectors
that $F$ and $G$ depend only on~$t$, $x$ and derivatives of~$u$ with respect to~$x$, and that
the maximal order of derivatives in~$F$ is less than~$\bar r$.

The proof of Lemma~\ref{LemmaOnOrderOfConsLawsOfLPEs} can easily be extended to other classes of 
(1+1)-dimensional evolution equations
of even order and some systems connected with evolution equations~\cite{Popovych&Ivanova2004ConsLawsLanl}.
\end{note}

\begin{theorem}\label{TheoremLocalCLsLPEs}
For an arbitrary equation of the form~\eqref{EqGenLPE} the space of all local conservation laws is generated 
by the conserved vectors
\begin{equation}\label{eqCVofLPEs}
\bigl(\alpha u,\, -\alpha Au_x+((\alpha A)_x-\alpha B)u\bigr),
\end{equation}
where the characteristic $\alpha=\alpha(t,x)$ runs through the solution set of the adjoint equation
\begin{equation}\label{EqAdjLPE}
\alpha_t+(A\alpha)_{xx}-(B\alpha)_x+C\alpha=0
\end{equation}
\end{theorem}

\begin{proof}
In view of Lemma~\ref{LemmaOnOrderOfConsLawsOfLPEs} we look for (local) conservation laws 
of equations from class~\eqref{EqGenLPE} in the form $D_tF(t,x,u)+D_xG(t,x,u,u_x)=0$.\label{CLsEvolEq}
First, we expand the total differentiation operators in the latter equality 
on the solution manifold of~\eqref{EqGenLPE}:
\[
F_t+F_u(Au_{xx}+Bu_x+Cu)+G_x+G_uu_x+G_{u_x}u_{xx}=0,
\]
and split the obtained expression with respect to the unconstrained variable $u_{xx}$.
Coefficients of the first power of $u_{xx}$ give
$
G=-AF_uu_x+G^1(t,x,u).
$
Splitting the rest of the conservation law with respect to different powers of $u_x$ yields
\[
F_{uu}=0,\quad
F_uB-A_xF_u-AF_{xu}+G^1_u=0,\quad
F_t+CuF_u+G^1_x=0.
\]
Solving the obtained system up to the usual equivalence relation of conserved vectors,
we obtain the conserved densities and fluxes of the local conservation laws of~\eqref{EqGenLPE}:
\[
F=\alpha u,\qquad G=-\alpha Au_x+((\alpha A)_x-\alpha B)u,
\]
where $\alpha=\alpha(t,x)$ is an arbitrary solution of equation~\eqref{EqAdjLPE}.
\end{proof}

\begin{note}
It is well-known that the space of characteristics of the local conservation laws of 
a linear differential equation~$Lu=0$ contains 
all functions of the independent variables which are solutions of the adjoint equation~$L^*\alpha=0$.
In view of Theorem~\ref{TheoremLocalCLsLPEs} any local conservation law of 
a $(1+1)$-dimensional linear second-order parabolic equation has such characteristics. 
Moreover, two different solutions of the adjoint equation give inequivalent characteristics 
corresponding to inequivalent conservation laws. 
Therefore, for any equation~$\mathcal L$ from class~\eqref{EqGenLPE} 
we can identify the solution space of the adjoint equation~\eqref{EqAdjLPE} with the space $\Ch_{\rm f}(\mathcal L)$ 
of `nontrivial' characteristics.
\end{note}

\begin{note}\label{NoteOnThirdOrderLPEs}
The above structure of the space $\Ch_{\rm f}(\mathcal L)$ is directly related to the restriction on 
the order of conservation laws for evolution equations of even orders. 
Third-order $(1+1)$-dimensional linear parabolic equations admit conservation laws with different characteristics. 
For example, for the equation $u_t=u_{xxx}$ the space of conservation laws of orders not greater than 3 
is generated by the conservation laws with the conserved vectors
$(\alpha u, \alpha_xu_x-\alpha_{xx}u)$, $(u^2,u_x{}^2-2uu_{xx})$ and $(u_x{}^2,u_{xx}{\!}^2-2u_xu_{xxx})$
and the characteristics $\alpha$, $2u$ and $2u_x$, respectively. 
Here $\alpha=\alpha(t,x)$ runs through the solution set of the same equation $\alpha_t=\alpha_{xxx}$.
The conservation laws parameterized by $\alpha$ have the structure usual for linear equations 
but the two other conservation laws do not. 
Moreover, acting by the symmetry operators $\p_x+Cu\p_u$ and $3t\p_t+x\p_x+Cu\p_u$ on solutions and 
conservation laws of the equation $u_t=u_{xxx}$, we construct infinite series of conservation laws
for it which are of arbitrarily large orders with conserved vectors quadratic in the derivatives of~$u$.
\end{note}

Let us emphasize once more that 
the function $\alpha=\alpha(t,x)$ is a characteristic of the conservation law of~\eqref{EqGenLPE} 
with the conserved vector~\eqref{eqCVofLPEs} and, therefore, an adjoint symmetry of equation~\eqref{EqGenLPE}.
Equation~\eqref{EqAdjLPE} is the formally adjoint equation to~\eqref{EqGenLPE}. 
Any solution of~\eqref{EqAdjLPE} is an adjoint symmetry of~\eqref{EqGenLPE}, 
and any adjoint symmetry of~\eqref{EqGenLPE} depends only on~$t$ and~$x$ and is a solution of~\eqref{EqAdjLPE}.
Conservation laws of~\eqref{EqGenLPE} are linearly independent iff 
the corresponding solutions of~\eqref{EqAdjLPE} are linearly independent.

\begin{corollary}
There is a one-to-one correspondence between local conservation laws and zero-order adjoint symmetries of equation~\eqref{EqGenLPE}.
\end{corollary}

Proposition~\ref{PropositionOnTransOfCLs}, Corollary~\ref{CorollaryOnSemiNormOfLPEs} and Theorem~\ref{TheoremLocalCLsLPEs} 
imply the following statement.

\begin{proposition}\label{PropositionOnTransOfCharsOfCLsForLPEs}
Any point transformation between equations from class~\eqref{EqGenLPE} 
is canonically prolonged to the characteristics of conservation laws of equations from this class by 
the formula
\[
\tilde\alpha=\frac{\kappa}{X_xU^1}\alpha.
\]
(The constant~$\kappa$ arises due to the linearity of the characteristic space. 
It is inessential and can be set equal to~1.)
\end{proposition}

\begin{note}\label{NoteOnSufficiencyOfRedFormForInvestigationOfCL}
Any equation~\eqref{EqGenLPE} is reduced by an equivalence transformation to an 
equation of the form~\eqref{EqReducedLPE}. 
In view of Proposition~\ref{PropositionOnInducedMappingOfCL} 
this allows us to restrict ourselves to an investigation of conservation laws for the 
simpler reduced form~\eqref{EqReducedLPE} of parabolic equations. 
The space of local conservation laws of an equation~\eqref{EqReducedLPE} 
is generated by the conserved vectors $(\alpha u,\, -\alpha u_x+\alpha_xu)$,
where the characteristic $\alpha=\alpha(t,x)$ runs through the solution set of the associated adjoint equation
$\alpha_t+\alpha_{xx}-V\alpha=0$.
At the same time, it is by no means evident that for a further study of potential conservation laws and
potential symmetries it suffices to consider only the reduced form. This point deserves additional investigation. 
\end{note}

The equivalence relation generated by the equivalence group on the set of pairs 
`(equation from the class, its conservation law)' can be used for the normalization of parameters in 
different ways. 
We can maximally simplify the form of the equations under consideration as 
in Note~\ref{NoteOnSufficiencyOfRedFormForInvestigationOfCL}. 
Another way is to simplify equations and their conserved vectors or characteristics simultaneously. 
For example, let us act on equation~\eqref{EqGenLPE} and its characteristic by 
the equivalence transformation from~$G^\sim$ with 
$T_t=\sign A$, $X_x=|A|^{-1/2}$, $U^1=\alpha|A|^{1/2}$. 
Then $\tilde A=1$, $\tilde\alpha=1$ and, moreover, $\tilde C=\tilde B_{\tilde x}$ 
since $\tilde\alpha$ should be a solution of the transformed adjoint equation 
$\tilde\alpha_{\tilde t}+(\tilde A\tilde\alpha)_{\tilde x\tilde x}-(\tilde B\tilde\alpha)_{\tilde x}
+\tilde C\tilde\alpha=0$. As a result, the following proposition holds:

\begin{proposition}
Any pair $(\mathcal L,\mathcal F)$, where $\mathcal L$ is an equation from class~\eqref{EqGenLPE} and 
$\mathcal F\in\CL(\mathcal L)$, is reduced by a point transformation from the equivalence group of 
class~\eqref{EqGenLPE} to a pair $(\tilde{\mathcal L},\tilde{\mathcal F})$, 
where $\tilde{\mathcal L}$ is a Fokker--Planck equation $u_t=u_{xx}+(\tilde Bu)_x$ and 
$\tilde{\mathcal F}$ is its conservation law with the characteristic $\tilde\alpha=1$, 
i.e., $(u,-u_x-\tilde Bu)\in\tilde{\mathcal F}$. 
\end{proposition}

\section{The adjoint variational principle}\label{SectionOnAdjointVariationalPrincipleForLPEs}

The adjoint variational principle is well-known especially in physics. 
In the nonlinear case it is also called the \emph{composite variational principle}~\cite{Atherton&Homsy1975}.
It provides a way of constructing generalized Lagrangians for systems of 
differential or other equations which have no usual Lagrangians, e.g., for evolution equations. 
Thus, the adjoint variational principle for the linear heat equation was considered in the 
classical textbook by Morse and Feshbach \cite[p.~313]{Morse&Feshbach1953}. 
It was discussed in~\cite{Roussopoulos1953} for general linear operator equations. 
A comparison of the adjoint variational principle with other approaches to deriving generalized Lagrangians 
was presented in \cite[pp.~341--342]{Tonti1973} along with a review of the literature on its application.  
A nonlinear version was proposed in~\cite{Finlayson1972} and used ibid.\ for the Navier--Stokes equations. 
The composite variational principle in its general formulation was given in~\cite{Atherton&Homsy1975}. 
See also~\cite[Exercise~5.27]{Olver1986} and the discussion in section~2 of~\cite{Caviglia1986}.

Briefly the composite variational principle can be described in the following way. 
Let~$\mathcal L$ be a system~$L(x,u_{(\rho)})=0$ of $l$ differential equations $L^1=0$, \ldots, $L^l=0$
for $m$ unknown functions $u=(u^1,\ldots,u^m)$
of $n$ independent variables $x=(x_1,\ldots,x_n).$ (Here we employ the notations of Section~\ref{SectionOnTheorBackgroundOnCLs}. 
Thus, the index~$\mu$ runs from~1 to~$l$.)
We introduce $l$ auxiliary dependent variables $v=(v^1,\dots,v^l)$ and construct the Lagrangian 
\[\mathfrak L=v\cdot L=v^\mu L^\mu(x,u_{(\rho)}).\] 
The Euler--Lagrange equations for the corresponding functional are
\[
L=0,\quad \Fder_L^*(v)=0, 
\]
where the matrix differential operator ${\Fder}_L^*$ is the adjoint of
the Fr\'echet derivative~${\Fder}_L^{\phantom{*}}$ (see subsection~\ref{SectionOnCharacteristicsOfConsLaws}).
The system $\Fder_L^*(v)=0$ will be called variationally adjoint to the system~$\mathcal L$. 
Only in the linear case the variationally adjoint equations coincide with the usual adjoint equations~$L^*v=0$~\cite{Atherton&Homsy1975}.
Therefore, simultaneously extending the tuple of dependent variables with $l$ auxiliary variables and 
the system with the variationally adjoint equations always results in the system formed by the Euler--Lagrange equations 
of a special functional. 
In physical terms this means that a dissipative system possessing a usual friction is considered simultaneously with 
its `mirror reflection' possessing a negative friction and absorbing the energy lost by the initial system 
\cite[Section~3.2]{Morse&Feshbach1953}.
As a result, the total energy of the extended systems is conserved. 
The trick is quite artificial from the physical point of view but nevertheless it has a number of applications. 
It allows one to operate with dissipative systems as if they were conservative. 

Interest in applying the adjoint variational principle in the framework of symmetry analysis has arisen 
quite recently~\cite{Ibragimov2006,Ibragimov&Kolsrud2004,Ivanova2006,Ivanova200x}.
Thus, a nice result on prolongation of symmetries of initial systems to the auxiliary variables of the adjoint variational principle 
was proved in~\cite{Ibragimov2006}. The prolongation is carried out in such a way that 
the prolonged operators are variational symmetries of extended Lagrangians and, therefore, Lie invariance operators 
of extended systems. At the same time, this result can be obtained for linear equations without the usage of the Lagrangian technique. 
It can be strengthened for special classes of differential equations including class~\eqref{EqGenLPE}.

Let $\mathcal L$ be a homogeneous linear differential equation for one unknown function~$u$, i.e., 
$m=1$ and let $\mathcal L$ be presented in the form $Lu=0$. Here $L$ is the associated linear differential operator:
$L=A^\nu(x)\p^\nu$, where 
$\nu=(\nu_1,\dots,\nu_n)$ is a multiindex, $\nu_i\in\mathbb{N}\cup\{0\}$, \mbox{$|\nu|:=\nu_1+\cdots+\nu_n\leqslant\rho$},
$\p:=(\p_{x_1},\dots,\p_{x_n})$, $\p^\nu:=\p_{x_1}^{\nu_1}\dots\p_{x_n}^{\nu_n}$, 
$A^\nu=A^\nu(x)$ are smooth functions of~$x$,
and summation over the multiindex $\nu$ is understood.
The adjoint $K^*$ to a linear differential operator~$K$ (over the real numbers) can be defined by means of the following formal rules. 
The adjoint to a linear combination of operators is the linear combination of the adjoint operators with the same coefficients. 
$(K_1K_2)^*=K_2^*K_1^*$, $(\p_{x_i})^*=-\p_{x_i}$ and 
the adjoint operator to the multiplication operator by a fixed function is the multiplication operator by the same function.
Thus, \mbox{$L^*=(-\p)^\nu A^\nu$}. 
The equation~$\mathcal L^*$: $L^*v=0$ is adjoint to $\mathcal L$ in both the usual linear and the variational sense. 
(According to the adjoint variational principle, only one auxiliary variable~$v$ should be introduced for the equation~$\mathcal L$.)
Since $(L^*)^*=L$,
the adjoint equation to~$\mathcal L^*$ also can be assumed to coincide with the initial equation, i.e., 
$(\mathcal L^*)^*=\mathcal L$.
The united system $\mathcal L\cap\mathcal L^*$: $Lu=0$, $L^*v=0$ is the system of Euler--Lagrange equations 
of the Lagrangian $vLu$ as well as of the equivalent Lagrangian $uL^*v$ and any linear combination of these Lagrangians. 
The Lagrangians $vLu$ and $uL^*v$ are equivalent since $vLu-uL^*v$ is a total divergence, 
i.e., $vLu-uL^*v=\mathop{\rm Div} F$ for some $n$-tuple~$F$ of differential functions bilinear in $u$, $v$ 
and their derivatives. 
This formula also implies that $v\in\Ch(\mathcal L)$ if $L^*v=0$. 
The tuple $F$ for a fixed solution~$v$ of the equation $L^*v=0$ gives a conserved vector 
of the conservation law corresponding to this characteristic. 

Suppose that $Q=\xi^i(x)\p_{x_i}+\eta^1(x)u\p_u$ is an essential Lie symmetry operator of~$\mathcal L$ and 
$\widehat Q=-\xi^i(x)\p_{x_i}+\eta^1(x)$ is the associated first-order differential operator acting on functions of~$x$. 
Employing the Hadamard lemma, we write the infinitesimal criterion $Q_{(\rho)}Lu|_{Lu=0}=0$ in the form 
$Q_{(\rho)}Lu=\lambda Lu$~\cite{Olver1986}. Here $Q_{(\rho)}$ is the standard $\rho$-th order prolongation of~$Q$ 
and in the case under consideration $\lambda$ is a smooth function only of~$x$.
The infinitesimal criterion is equivalent to the operator equality $[\widehat Q,L]=\widehat\lambda L$, 
where $\widehat\lambda=\eta^1-\lambda$ and $[\widehat Q,L]=\widehat QL-L\widehat Q$ is the commutator of 
the operators~$\widehat Q$ and~$L$. 

\begin{proposition}\label{PropositionOnAdjointVariationalPrincipleForEssentialSymsOfLinEqs}
$Q=\xi^i\p_{x_i}+\eta^1u\p_u\in\mathfrak g^{\rm ess}(\mathcal L)$ iff 
$\,Q^\dag=\xi^i\p_{x_i}+\theta^1v\p_v\in\mathfrak g^{\rm ess}(\mathcal L^*)$, 
where \[\eta^1+\theta^1=\widehat\lambda-\xi^i_{x_i}.\] 
Moreover, $\bar Q:=\xi^i\p_{x_i}+\eta^1u\p_u+\theta^1v\p_v\in\mathfrak g^{\rm ess}(\mathcal L\cap\mathcal L^*)$ 
and $\bar Q$ is a variational symmetry operator of the Lagrangian~$vLu$.
\end{proposition}

\begin{proof}
Conjugating the operator equality $[\widehat Q,L]=\widehat\lambda L$, we obtain 
$[-\widehat Q^*+\widehat\lambda,L^*]=\widehat\lambda L^*$. 
At~the same time, $-\widehat Q^*+\widehat\lambda=-\xi^i\p_{x_i}+\theta^1$ is the first-order differential operator 
acting on functions of~$x$ which is associated with $Q^\dag$. Therefore, $Q^\dag\in\mathfrak g^{\rm ess}(\mathcal L^*)$.
Since $(L^*)^*=L$, the converse statement is true as well. 
The operators $Q$ and $Q^\dag$ have the same $x$-part and the system~$\mathcal L\cap\mathcal L^*$ is uncoupled.
Hence $\bar Q\in\mathfrak g^{\rm ess}(\mathcal L\cap\mathcal L^*)$.
This also follows from the fact that $\mathcal L$ and $\mathcal L^*$ are the Euler--Lagrange equations of the Lagrangian~$vLu$ and 
$\bar Q$ is a variational symmetry operator thereof in view of the infinitesimal criterion 
of variational symmetry~\cite{Olver1986}. 
Indeed, \mbox{$Q_{(\rho)}(vLu)+\xi^i_{x_i}(vLu)=(vLu)(\theta^1+\lambda+\xi^i_{x_i})=0$}. 
\end{proof}

The algebra $\mathfrak g^\infty(\mathcal L\cap\mathcal L^*)$ is formed by the operators 
$f\p_u+g\p_v$, where the parameter-functions $f=f(t,x)$ and $g=g(t,x)$ run through 
the solution sets of the equations~$\mathcal L$ and~$\mathcal L^*$, respectively.
Any such operator is a variational symmetry operator of the Lagrangian~$vLu$. 
It is projectable to $(t,x,u)$ and $(t,x,v)$ and its projections belong to $\mathfrak g^\infty(\mathcal L)$ 
and $\mathfrak g^\infty(\mathcal L^*)$. It can be assumed that 
$\mathfrak g^\infty(\mathcal L\cap\mathcal L^*)=\mathfrak g^\infty(\mathcal L)\oplus\mathfrak g^\infty(\mathcal L^*)$.
The operators from $\mathfrak g^\infty(\mathcal L)$ and $\mathfrak g^\infty(\mathcal L^*)$ 
have a trivial zero complement to operators from $\mathfrak g^\infty(\mathcal L\cap\mathcal L^*)$ 
according to~\cite{Ibragimov2006}, i.e. 
the algebras $\mathfrak g^\infty(\mathcal L)$ and $\mathfrak g^\infty(\mathcal L^*)$ cannot be obtained from each other 
via the adjoint variational principle. 
A slightly different situation obtains for the operators $u\p_u$ and $v\p_v$. 
An arbitrary linear combination of them belongs to $\mathfrak g(\mathcal L\cap\mathcal L^*)$ 
but only the operators proportional to $u\p_u-v\p_v$ are variational symmetries. 
At the same time, among such linear combinations only the operators from $\langle u\p_u+v\p_v\rangle$ are trivial 
Lie invariance operators of the system $\mathcal L\cap\mathcal L^*$. 
So $\mathfrak g^{\rm triv}(\mathcal L\cap\mathcal L^*)\subsetneq
\mathfrak g^{\rm triv}(\mathcal L)\oplus\mathfrak g^{\rm triv}(\mathcal L)$.

Let us return to linear $(1+1)$-dimensional second-order parabolic equations. 
Let $\mathcal L$ be an equation from class~\eqref{EqGenLPE}. 
Then the adjoint equation~$\mathcal L^*$ is of the form~\eqref{EqAdjLPE}, i.e., the system~$\mathcal L\cap\mathcal L^*$~is 
\begin{equation}\label{EqAdjointVariationalPrincipleSystemOfLPE}
u_t=Au_{xx}+Bu_x+Cu, \quad v_t+(Av)_{xx}-(Bv)_x+Cv=0.
\end{equation}
where $A=A(t,x)$, $B=B(t,x)$ and $C=C(t,x)$ are arbitrary smooth functions, $A\ne0$.
For class~\eqref{EqGenLPE} we extend the adjoint variational principle to admissible transformations. 
At first we present auxiliary statements on admissible transformations for wider classes of 
systems of evolutionary equations. Their proofs are based on the direct method. 

Consider the class of systems of second-order evolutionary equations of the general form 
\begin{equation}\label{EqGen2ndOrderEvolSystem}
\bar u_t=\bar F(t,x,\bar u_x,\bar u_{xx}),
\end{equation}
where $\bar u=(u^1,\dots,u^m)$, $\bar F=(F^1,\dots,F^m)$ and $|\p\bar F/\p\bar u_{xx}|\ne0$. 
Any point transformation~$\mathcal T$ in the space of variables $(t,x,\bar u)$ has the form 
\[
\tilde t=\mathcal T^t(t,x,\bar u),\quad
\tilde x=\mathcal T^x(t,x,\bar u),\quad 
\tilde u^a=\mathcal U^a(t,x,\bar u),
\]
where the Jacobian $|\p(\mathcal T^t,\mathcal T^x,\bar{\mathcal U})/\p(t,x,\bar u)|$ does not vanish, 
$\bar{\mathcal U}=(\mathcal U^1,\dots,\mathcal U^m)$. 
In what follows the indices $a$, $b$ and $c$ run from 1 to~$m$.

\begin{lemma}\label{LemmaOnAdmTransOfGen2ndOrderEvolSystems}
A point transformation~$\mathcal T$ connects two systems from class~\eqref{EqGen2ndOrderEvolSystem} iff 
$\,\mathcal T^t_x=\mathcal T^t_{u^a}=0$, i.e., $\mathcal T^t=T(t)$,
where $T$ is an arbitrary smooth function of~$t$ such that $T_t\ne0$. 
The arbitrary elements are transformed by the formula
\[
\tilde F^a=\frac{\mathcal U^a_{u^b}D_x\mathcal T^x-\mathcal T^x_{u^b}D_x\mathcal U^a}{T_tD_x\mathcal T^x}F^b+
\frac{\mathcal U^a_tD_x\mathcal T^x-\mathcal T^x_tD_x\mathcal U^a}{T_tD_x\mathcal T^x},
\]
where $D_x=\p_x+u^b_x\p_{u^b}+\cdots$ is the operator of total differentiation with respect to~$x$.
Therefore, class~\eqref{EqGen2ndOrderEvolSystem} is normalized. 
The equivalence group of class~\eqref{EqGen2ndOrderEvolSystem} is formed by the transformations 
determined in the space of variables and arbitrary elements by the above formulas.
\end{lemma}

Analogously, consider the subclass of class~\eqref{EqGen2ndOrderEvolSystem} formed by systems whose arbitrary elements 
are linear in~$\bar u_{xx}$, i.e., 
\begin{equation}\label{EqQuasiLin2ndOrderEvolSystem}
u^a_t=S^{ab}(t,x,\bar u,\bar u_x)u^b_{xx}+H^a(t,x,\bar u,\bar u_x),
\end{equation}
where $|S|\ne0$, $S=(S^{ab})$, $\bar H=(H^1,\dots,H^m)$.

\begin{lemma}\label{LemmaOnAdmTransOfQuasiLin2ndOrderEvolSystems}
A point transformation~$\mathcal T$ connects two systems from class~\eqref{EqQuasiLin2ndOrderEvolSystem} iff 
$\,\mathcal T^t_x=\mathcal T^t_{u^a}=0$ and $\,\mathcal T^x_{u^a}=0$, i.e., $\mathcal T^t=T(t)$ and $\mathcal T^x=X(t,x)$,
where $T$ and~$X$ are arbitrary smooth functions of their arguments such that $T_tX_x\ne0$. 
Moreover, the Jacobian $|\p\bar{\mathcal U}/\p\bar u|\ne0$.
The arbitrary elements are transformed by the formulas
\begin{gather*}
\tilde S=\frac{X_x^{\,2}}{T_t}\frac{\p\bar{\mathcal U}}{\p\bar u}S\left(\frac{\p\bar{\mathcal U}}{\p\bar u}\right)^{\!-1},
\\
\tilde{\bar H}=\frac1{T_t}\Biggl(\frac{\p\bar{\mathcal U}}{\p\bar u}\bar H+\bar{\mathcal U}_t
-\frac{X_t}{X_x}D_x\bar{\mathcal U}
-\frac{\p\bar{\mathcal U}}{\p\bar u}S\left(\frac{\p\bar{\mathcal U}}{\p\bar u}\right)^{\!-1}
\left(D_x\bar{\mathcal U}_x+u^a_xD_x\bar{\mathcal U}_{u^a}-\frac{X_{xx}}{X_x}D_x\bar{\mathcal U}\right)
\Biggr).
\end{gather*}
Therefore, class~\eqref{EqQuasiLin2ndOrderEvolSystem} is normalized. 
The equivalence group of class~\eqref{EqQuasiLin2ndOrderEvolSystem} is formed by the transformations 
determined in the space of variables and arbitrary elements by the above formulas.
\end{lemma}

\begin{note}\label{NoteOnAdmTransOfQuasiLin2ndOrderEvolSystemsWithDiagMatrices}
If $S=\diag(\Sigma^1,\dots,\Sigma^m)$ and~$\tilde S=\diag(\tilde\Sigma^1,\dots,\tilde\Sigma^m)$ are diagonal matrix-functions then 
in view of Lemma~\ref{LemmaOnAdmTransOfQuasiLin2ndOrderEvolSystems}
there exists a permutation $\sigma\in{\rm S}_m$ such that $\tilde\Sigma^a=X_x^{\,2}T_t^{-1}\Sigma^{\sigma(a)}$ and 
$\mathcal U^a_{u^b}(\Sigma^{\sigma(a)}-\Sigma^b)=0$. (There is no summation in the last formula.) 
This implies in the case $\Sigma^a\ne\Sigma^b$ for any $a\ne b$ that $\mathcal U^a_{u^b}=0$ if $\sigma(a)\ne b$.
\end{note}

\begin{corollary}\label{CorollaryOnAdmTransOfStronglyQuasiLin2ndOrderEvolSystems}
The subclass of class~\eqref{EqQuasiLin2ndOrderEvolSystem}, formed by the systems linear in the derivatives,
i.e., defined by the constraints $S^{ab}_{u^c_x}=0$, $H^a_{u^c_x\smash{u^b_x}}=0$ on the arbitrary elements, 
is normalized.
The equivalence group of this subclass is a subgroup of the equivalence group of 
class~\eqref{EqQuasiLin2ndOrderEvolSystem} and is formed by the transformations 
in which additionally the parameter-functions~$\mathcal U^a$ satisfy the conditions $\mathcal U^a_{u^bu^c}=0$, 
i.e., \mbox{$\mathcal U^a=U^{ab}(t,x)u^b+U^{a0}(t,x)$}.
\end{corollary}

\begin{corollary}\label{CorollaryOnAdmTransOfLin2ndOrderEvolSystems}
The subclass of class~\eqref{EqQuasiLin2ndOrderEvolSystem}, formed by the linear systems 
i.e., defined by the constraints $S^{ab}_{u^c_x}=S^{ab}_{u^c}=0$, 
$H^a_{u^c_x\smash{u^b_x}}=H^a_{u^c_x\smash{u^b}}=H^a_{u^c\smash{u^b}}=0$ on arbitrary elements, 
is normalized. 
The equivalence group of this subclass coincides with the group described 
in Corollary~\ref{CorollaryOnAdmTransOfStronglyQuasiLin2ndOrderEvolSystems}.
\end{corollary}

Combining Note~\ref{NoteOnNormalizationOfInhomAndHomLinSys} and 
Corollary~\ref{CorollaryOnAdmTransOfLin2ndOrderEvolSystems} implies the following statement on 
properties of the corresponding class of homogeneous linear systems. 

\begin{corollary}\label{CorollaryOnAdmTransOfHomogenLin2ndOrderEvolSystems}
The subclass of class~\eqref{EqQuasiLin2ndOrderEvolSystem} formed by the homogeneous linear systems 
(i.e., the constraints on arbitrary elements are $S^{ab}_{u^c_x}=S^{ab}_{u^c}=0$, 
$H^a_{u^c_x\smash{u^b_x}}=H^a_{u^c_x\smash{u^b}}=H^a_{u^c\smash{u^b}}=0$ and $u^c_xH^a_{u^c_x}+u^cH^a_{u^c}=H^a$), is semi-normalized. 
The equivalence group of this subclass is a subgroup of the group from Corollary~\ref{CorollaryOnAdmTransOfStronglyQuasiLin2ndOrderEvolSystems},
which is formed by the transformations with $U^{a0}=0$.
A~point transformation~$\mathcal T$ connects two systems from this class iff 
$\,\mathcal T^t_x=\mathcal T^t_{u^a}=0$, $\,\mathcal T^x_{u^a}=0$ and $\mathcal U^a_{u^bu^c}=0$, 
i.e., $\mathcal T^t=T(t)$, $\mathcal T^x=X(t,x)$ and $\mathcal U^a=U^{ab}(t,x)u^b+U^{a0}(t,x)$,
where $T$,$X$, $U^{ab}$ and $U^{a0}$ are arbitrary smooth functions of their arguments such that $T_tX_x|U^{ab}|\ne0$, 
and additionally $(\hat U^{ab}U^{b0})$ is a solution on the initial system. 
Here $(\hat U^{ab})$ is the inverse matrix of~$(U^{ab})$.
\end{corollary}

\begin{note}\label{NoteOnChangingRepresentationOf2ndOrderEvolSystems}
Representing systems from class~\eqref{EqGen2ndOrderEvolSystem} whose arbitrary elements are linear in~$\bar u_{xx}$ 
in the form~\eqref{EqQuasiLin2ndOrderEvolSystem}, we replace each $F^a$ by the tuple of arbitrary elements $(S^{ab},H^a)$.
In the situation under consideration the different representations give equivalent results. 
Thus, the formulas for the transformation of~$S^{ab}$ and~$H^a$ are obtained from the analogous formulas for~$F^a$ by splitting 
with respect to $\bar u_{xx}$. The corresponding sets of admissible transformations as well as equivalence groups are isomorphic. 
We can work with the above subclasses of class~\eqref{EqQuasiLin2ndOrderEvolSystem} in the same way. 
Namely, we can constrain the functions~$S^{ab}$ and successively replace $H^a$ by the expressions 
$H^{ab}(t,x,\bar u)u^b_x+G^a(t,x,\bar u)$, $H^{ab}(t,x)u^b_x+G^{ab}(t,x)u^b+G^{a0}(t,x)$ and $H^{ab}(t,x)u^b_x+G^{ab}(t,x)u^b$. 
The transformation formulas for the new arbitrary elements are constructed by splitting the transformation formula for~$\bar H$ 
with respect to $\bar u_x$ or $(\bar u_x,u)$. The equivalence groups of the different representations of these classes are isomorphic.
\end{note}

Let us continue with class~\eqref{EqAdjointVariationalPrincipleSystemOfLPE}. 
At first we study the corresponding class of inhomogeneous systems, writing them in the form
\begin{equation}\label{EqInhomAdjointVariationalPrincipleSystemOfLPE}
u_t=\varepsilon_1Au_{xx}+B^1u_x+C^1u+D^1, \quad v_t=\varepsilon_2Av_{xx}+B^2v_x+C^2v+D^2.
\end{equation}
Hereafter $m=2$, $\varepsilon_1=1$, $\varepsilon_2=-1$, $B^1=B$, $C^1=C$, $B^2=B-2A_x$, $C^2=-C+B_x-A_{xx}$, 
$L^a=\p_t-\varepsilon_aA_{xx}-B^a\p_x-C^a$. All the arbitrary elements are smooth functions of~$t$ and~$x$. 

Corollary~\ref{CorollaryOnAdmTransOfLin2ndOrderEvolSystems} implies that 
any point transformation between two arbitrary systems from class~\eqref{EqAdjointVariationalPrincipleSystemOfLPE} 
has the form 
\begin{gather}\nonumber
\tilde t=T(t), \quad \tilde x=X(t,x), \\ \label{Eq1OnAdmTransOfInhomAdjointVariationalPrincipleSystemOfOfLPEs}
\tilde u=U^{11}(t,x)u+U^{12}(t,x)v+U^{10}(t,x), \\ \nonumber 
\tilde v=U^{21}(t,x)u+U^{22}(t,x)v+U^{20}(t,x), 
\end{gather}
where $T_tX_x|U^{ab}|\ne0$. 
In view of Note~\ref{NoteOnAdmTransOfQuasiLin2ndOrderEvolSystemsWithDiagMatrices} we additionally have 
that $\tilde A=\delta X_x^{\,2}T_t^{-1}A$, $\delta=\pm1$. 
Moreover, $U^{12}=U^{21}=0$ if $\delta=1$ and $U^{11}=U^{22}=0$ if $\delta=-1$. 
We take $(b_1,b_2)=(1,2)$ if $\delta=1$ and $(b_1,b_2)=(2,1)$ otherwise. 
Using these notations the transformation of the arbitrary elements can be written as 
\begin{equation}\label{Eq2OnAdmTransOfInhomAdjointVariationalPrincipleSystemOfOfLPEs}
\arraycolsep=0ex
\begin{array}{l}\displaystyle
\tilde A=\delta\frac{X_x^2}{T_t}A,\quad 
\tilde B^a=\varepsilon_a\frac{\delta A}{T_t}\left(X_{xx}-2\frac{U^{ab_a}_x}{U^{ab_a}}X_x\right)
+\frac{B^{b_a}X_x-X_t}{T_t}, 
\\[2.5ex] \displaystyle
\tilde C^a=-\frac{U^{ab_a}}{T_t}L^{b_a}\frac1{U^{ab_a}},\quad
\tilde D^a=\frac{U^{ab_a}}{T_t}\left(D^{b_a}+L^{b_a}\frac{U^{a0}}{U^{ab_a}}\right).
\end{array}
\end{equation}
The specific connections between the coefficients $(B^1,C^1)$ and $(B^2,C^2)$ imply one more equation 
\begin{equation}\label{EqProlongationOfAdmTransOfAdjointVariationalPrincipleSystemOfLPE}
X_xU^{1b_1}U^{2b_2}=\kappa,
\end{equation}
where $\kappa$ is an arbitrary nonzero constant. 
(Compare with the formula from Proposition~\ref{PropositionOnTransOfCharsOfCLsForLPEs}.)

Analyzing the obtained result, we deduce the following: 
Any admissible transformation in class~\eqref{EqInhomAdjointVariationalPrincipleSystemOfLPE} is 
either the prolongation of a transformation between the first (or second) equations of the related systems 
according to condition~\eqref{EqProlongationOfAdmTransOfAdjointVariationalPrincipleSystemOfLPE} or 
the composition of such a prolonged transformation with the transformation given by the simultaneous transposition 
of the dependent variables and equations in the resulting system 
($\tilde u=v$, $\tilde v=u$, $\tilde A=-A$, $\tilde B=-A$, $\tilde C=-C+B_x-A_{xx}$). 
In the second case connections between arbitrary elements involve their derivatives. 
That is why we should consider the extended equivalence group of 
class~\eqref{EqInhomAdjointVariationalPrincipleSystemOfLPE}, 
admitting the dependence of transformations of arbitrary elements on their (and only their) derivatives.

\begin{proposition}\label{PropositionOnAdmTransOfInhomAdjointVariationalPrincipleSystemOfOfLPEs}
Class~\eqref{EqInhomAdjointVariationalPrincipleSystemOfLPE} is normalized with respect to 
the extended equivalence group~$G^\sim_{\rm ext}$ which is formed by the transformations described by 
the formulas~\eqref{Eq1OnAdmTransOfInhomAdjointVariationalPrincipleSystemOfOfLPEs}, 
\eqref{Eq2OnAdmTransOfInhomAdjointVariationalPrincipleSystemOfOfLPEs} 
and~\eqref{EqProlongationOfAdmTransOfAdjointVariationalPrincipleSystemOfLPE}. 
The usual equivalence group~$G^\sim$ of class~\eqref{EqInhomAdjointVariationalPrincipleSystemOfLPE}
is the subgroup of~$G^\sim_{\rm ext}$, defined by the condition $\delta=1$. 
Its `essential' part formed by the transformation with $U^{a0}=0$ 
can be obtained as the prolongation of the `essential' part of the equivalence group of the class~\eqref{EqGenInhomLPE} 
(or the class of the same equations written in the adjoint form)
according to condition~\eqref{EqProlongationOfAdmTransOfAdjointVariationalPrincipleSystemOfLPE}. 
Conversely, the equivalence groups of the class~\eqref{EqGenInhomLPE} 
and the class of the same equations written in the adjoint form are the projection of~$G^\sim$ to 
the corresponding sets of variables and arbitrary elements. 
The complement of~$G^\sim$ in~$G^\sim_{\rm ext}$ is formed by the compositions of elements from~$G^\sim$ and 
the simultaneous transposition of the dependent variables and equations. 
\end{proposition}

Similarly to the above classes of linear evolution systems, 
in view of Note~\ref{NoteOnNormalizationOfInhomAndHomLinSys} and 
Proposition~\ref{PropositionOnAdmTransOfInhomAdjointVariationalPrincipleSystemOfOfLPEs} 
we obtain that the corresponding class~\eqref{EqAdjointVariationalPrincipleSystemOfLPE} of homogeneous systems
is semi-normalized in the extended sense. 
The extended/usual equivalence group of~\eqref{EqAdjointVariationalPrincipleSystemOfLPE} is 
the projection (neglecting transformations for $D^1$ and $D^2$) of the subgroup of 
the extended/usual equivalence group of~\eqref{EqInhomAdjointVariationalPrincipleSystemOfLPE}, 
defined by the constraints $U^{a0}=0$.
The usual equivalence group of~\eqref{EqAdjointVariationalPrincipleSystemOfLPE}
can be constructed via the prolongation of the equivalence group of class~\eqref{EqGenLPE} 
(or the class of the same equations written in the adjoint form)
according to condition~\eqref{EqProlongationOfAdmTransOfAdjointVariationalPrincipleSystemOfLPE}. 
Conversely, the equivalence groups of~\eqref{EqGenLPE} and the class of the same equations written in the adjoint form 
are the projections of the usual equivalence group of~\eqref{EqAdjointVariationalPrincipleSystemOfLPE} to 
the corresponding sets of variables and arbitrary elements. 
The additional condition for the admissible transformations in class~\eqref{EqAdjointVariationalPrincipleSystemOfLPE} is 
$L^{b_a}(U^{a0}/U^{ab_a})=0$, where there is no summation with respect to~$a$.

In the case of class~\eqref{EqGenLPE}
Proposition~\ref{PropositionOnAdjointVariationalPrincipleForEssentialSymsOfLinEqs} can be reformulated 
in a more precise way since any essential symmetry operator 
$Q=\tau\p_t+\xi\p_x+\eta^1u\p_u$ of an equation from class~\eqref{EqGenLPE} has 
$\widehat\lambda=\tau_t$. The same statement is derived as an infinitesimal consequence of 
formula~\eqref{EqProlongationOfAdmTransOfAdjointVariationalPrincipleSystemOfLPE} 
after taking into account that any symmetry transformation of any equation from class~\eqref{EqGenLPE} is 
an admissible transformation in this class. 

\begin{corollary}\label{CorollaryOnAdjointVariationalPrincipleForEssentialSymsOfLPEs}
Let $\mathcal L$ be an equation from class~\eqref{EqGenLPE}.
Then $Q=\tau\p_t+\xi\p_x+\eta^1u\p_u\in\mathfrak g^{\rm ess}(\mathcal L)$ iff 
$\,Q^\dag=\tau\p_t+\xi\p_x+\theta^1v\p_v\in\mathfrak g^{\rm ess}(\mathcal L^*)$, 
where $\eta^1+\theta^1=-\xi_x$. 
Moreover, \[\bar Q=\tau\p_t+\xi\p_x+\eta^1u\p_u+\theta^1v\p_v\in\mathfrak g^{\rm ess}(\mathcal L\cap\mathcal L^*)\] 
and $\bar Q$ is a variational symmetry operator of the Lagrangian~$vL^1u$ (or~$uL^2v$).
\end{corollary}

Using equivalence transformations, we can gauge the arbitrary elements~$A$, $B$ and~$C$ in different ways 
analogous to gauging them in class~\eqref{EqGenLPE}. 

The gauge~$A=1$ preserves all the normalization properties, relations between 
equivalence groups etc.\ in both the inhomogeneous and homogeneous cases. 
Under the gauge~$A=1$ we obtain the condition $\delta X^2=T_t$. 
Therefore, the constant~$\delta$ becomes coupled with the sign of~$T_t$ (i.e., $\delta=\sign T_t$) 
and $X=\delta'|T_t|^{1/2}x+\zeta(t)$, where $\delta'=\pm1$. Here $\zeta$ is an arbitrary smooth function of~$t$. 
The discrete extended equivalence transformation transposing the dependent variables and equations 
is replaced by its composition with the discrete transformation of alternating the sign of~$t$, 
which also is an extended equivalence transformation. 

The further gauging of $B$ to $0$ gives the conditions 
\[
\frac{U^{ab_a}_x}{U^{ab_a}}=-\frac{\delta\varepsilon_a}2\frac{X_t}{X_x}, 
\quad\mbox{i.e.,}\quad 
U^{ab_a}=\theta^a(t)\exp\left(-\varepsilon_a\frac{T_{tt}}{8T_t}x^2
-\frac{\varepsilon_a\delta}{2\delta'}\frac{\zeta_t}{|T_t|^{1/2}}x\right).
\]
Here $\theta^1$ and $\theta^2$ are arbitrary smooth nonvanishing functions of~$t$. 
A new nuance is that under the gauge $(A,B)=(1,0)$ the equivalence transformations involving the transposition 
of the dependent variables and equations reduce to usual ones since they contain no derivatives of arbitrary elements. 
In view of this, the subclass of class~\eqref{EqInhomAdjointVariationalPrincipleSystemOfLPE} 
with the gauge $(A,B)=(1,0)$ is normalized in the usual sense. 
The corresponding subclass of homogeneous systems~\eqref{EqAdjointVariationalPrincipleSystemOfLPE} 
is semi-normalized in the usual sense.
To avoid the transposition transformations in this case, 
statements on relations between equivalence groups should be formulated in other terms, 
e.g., for continuous equivalence groups. 

Note that point symmetry groups of some systems from 
class~\eqref{EqAdjointVariationalPrincipleSystemOfLPE}/\eqref{EqInhomAdjointVariationalPrincipleSystemOfLPE} 
contain transformations involving the transposition of~$u$ and~$v$. 
Such groups cannot be projected to the symmetry groups of single equations in contrast to the Lie symmetry groups. 
For example, any system \eqref{EqAdjointVariationalPrincipleSystemOfLPE} with 
$A=1$, $B=0$ and $C_t=0$ is invariant with respect to the transformation 
$\tilde t=-t$, $\tilde x=x$, $\tilde u=v$, $\tilde v=u$.

The gauge~$C=0$ determines the subclasses of the systems of Kolmogorov and Fokker--Planck 
equations associated with each other. The reduced form of such systems is given by the gauge $(A,C)=(1,0)$. 
Both the gauges break the normalization properties. 
Nevertheless, the admissible transformations in these subclasses can be described. 
They should satisfy the conditions $L^{b_a}(1/U^{ab_a})=0$.
The `essential' admissible transformations (for which $U^{a0}=0$) are obtained as the prolongations of the `essential'
admissible transformations of the Kolmogorov (or Fokker--Planck) equations
according to condition~\eqref{EqProlongationOfAdmTransOfAdjointVariationalPrincipleSystemOfLPE}. 
Conversely, the `essential' admissible transformations in the class of the Kolmogorov (or Fokker--Planck) equations
are the projections of the `essential' admissible transformations of the associated systems, 
which do not involve the transposition of the dependent variables and equations, 
to the corresponding sets of variables and arbitrary elements. 
Therefore, the sets of the `essential' admissible transformations 
of the Kolmogorov and Fokker--Planck equations are similar. 
This implies the following statement in view of Corollary~\ref{CorollaryOnClassificationOfLPEsC0} and 
Proposition~\ref{PropositionOnAdjointVariationalPrincipleForEssentialSymsOfLinEqs}. 

Consider the class of Fokker--Planck equations having the form
\begin{equation}\label{EqFokkerPlanck}
v_t=(A(t,x)v_x)_x-(B(t,x)v)_x.
\end{equation}
(For the form to be canonical, we alternate the sign of $t$.)

\begin{corollary}\label{CorollaryOnClassificationOfLPEsInFokkerPlankForm}
The kernel Lie algebra of class~\eqref{EqFokkerPlanck} is $\langle v\p_v\rangle$.
Any equation from this class is invariant with respect to the operators $f\p_v$,
where the parameter-function $f=f(t,x)$ runs through the solution set of this equation. 
Any equation from class~\eqref{EqFokkerPlanck} is reduced by a point transformation 
to an equation with $A=1$ from the same class.
Up to point transformations, all possible cases of extension of the maximal 
Lie invariance algebras in class~\eqref{EqFokkerPlanck} are exhausted by the following ones 
(in all the cases $A=1$; the values of~$B$ are given together with the corresponding maximal Lie invariance algebras):

\vspace{1.5ex}

$\makebox[6mm][l]{\rm 1.}
B=B(x) \colon\quad \langle\p_t,\,v\p_v,\,f\p_v\rangle;$

\vspace{1.5ex}

$\makebox[6mm][l]{\rm 2.}
B=\nu x^{-1},\ \nu\geqslant1,\ \nu\ne2 \colon\quad \langle\p_t,\, D,\, \Pi+2\nu tv\p_v,\, v\p_v,\, f\p_v\rangle;$

\vspace{1.5ex}

$\makebox[6mm][l]{\rm 3.}
B=x^{-1}\bigl(1-2\varkappa\tan(\varkappa\ln|x|)\bigr),\ \varkappa\ne0 \colon\quad 
\langle\p_t,\, 2D-(xB-2)v\p_v,\, \Pi+2txBv\p_v,\, v\p_v,\, f\p_v\rangle;$

\vspace{1.5ex}

$\makebox[6mm][l]{\rm 4.}
B=0 \colon\quad \langle\p_t,\, \p_x,\, G,\, D,\, \Pi,\, v\p_v,\, f\p_v\rangle$.

\vspace{1.5ex}

\noindent
Here $D=2t\p_t+x\p_x,\ \Pi=4t^2\p_t+4tx\p_x-(x^2+2t)v\p_v,\ G=2t\p_x-xv\p_v.$
\end{corollary}

Note that the kernel Lie algebra of~\eqref{EqFokkerPlanck} is different 
from the one of the subclass of~\eqref{EqGenLPE} with $C=0$ 
since the prolongation of the operators to~$v$ depends on the arbitrary elements.

Analogously to the `Kolmogorov' form, the group classification of class~\eqref{EqFokkerPlanck} 
with respect to its equivalence group can be derived from the classification presented 
in Corollary~\ref{CorollaryOnClassificationOfLPEsInFokkerPlankForm} 
through extending the classification cases by essential admissible transformations which are not 
generated by the equivalence group.

\section{Potential conservation laws}\label{SectionOnPotCLsLPEs}

As proved in Theorem~\ref{TheoremLocalCLsLPEs},
any (1+1)-dimensional linear second-order parabolic equation possesses an infinite series of local conservation laws.
Studying potential conservation laws of equations from class~\eqref{EqGenLPE},
in view of Note~\ref{NoteOnSufficiencyOfRedFormForInvestigationOfCL} and Proposition~\ref{PropositionOn2DConsLawEquivRelation} 
we may restrict ourselves to equations of the reduced form~\eqref{EqReducedLPE}, i.e., $u_t-u_{xx}+Vu=0$.
Fixing an arbitrary $p\in\mathbb{N}$ and
choosing $p$~linearly independent solutions $\bar\alpha=(\alpha^1,\ldots,\alpha^p)$ of
the associated adjoint equation
\begin{equation}\label{eqAdjEqOfLinParEqInRedForm}
\alpha_t+\alpha_{xx}-V\alpha=0,
\end{equation}
we obtain $p$~linearly independent conservation laws of equation~\eqref{EqReducedLPE} 
with the conserved vectors $(F^s,G^s)=(\alpha^s u,\alpha^s_xu-\alpha^s u_x)$.
(Hereafter the indices~$s$, $\sigma$ and $\varsigma$ run from~1 to $p$. 
Let us also recall that summation over repeated indices is assumed.)

The potentials $\bar v=(v^1,\ldots,v^p)$ introduced with these conservation laws by the formulas
\begin{equation}\label{eqGenPotSysOfParEqInRedForm}
v^s_x=\alpha^s u,\quad v^s_t=\alpha^s u_x-\alpha^s_xu
\end{equation}
are independent in the sense of Definition~\ref{DefinitionOfPotentialDependence} according to the following lemma.

\begin{lemma}\label{LemmaOnIndepPotOfEvolEqs}
For any equation~\eqref{EqReducedLPE} the potentials are locally dependent on the equation manifold
iff the corresponding conservation laws and, therefore, 
the corresponding solutions of equation~\eqref{eqAdjEqOfLinParEqInRedForm} are linearly dependent.
\end{lemma}

\begin{proof}
The `if' part of the claim being obvious, we only prove the converse statement, arguing by contradiction. 
Suppose that the potentials~$v^1$, \dots, $v^p$ introduced with the linearly independent solutions 
$\alpha^1$, \dots, $\alpha^p$ of~\eqref{eqAdjEqOfLinParEqInRedForm} are locally dependent. 
In case $p=1$ we have local triviality of $v^1$ as a potential,
i.e., $v^1$ can be expressed in terms of local variables and, hence, the corresponding conservation law is trivial.
Therefore it suffices to investigate the case where the number of independent conservation laws
is greater than~1. 

Without loss of generality we may assume that
there exist $r\in\mathbb N$ and a fixed function $P$ of~$t$, $x$, $\check v =(v^2,\ldots v^p)$ and $u_{(r)}$ such that
$v^1=P(t,x,\check v,u_{(r)})$ for any solution of the combined system determining
the whole set of potentials~$v^1$, \ldots, $v^p$ (up to gauge transformations, i.e., up to adding negligible constants to the potentials).
In view of equation~\eqref{EqReducedLPE} and its differential consequences,
we may suppose that $P$ depends only on $t$, $x$, $\check v$ and
$u_k=\partial^k u/\partial x^k$, $k=\overline{0,r'},$ where $r'\leqslant 2r$.
Let us apply the operator~$D_x$ to the condition $v^1=P(t,x,\check v,u,u_1,\ldots,u_{r'})$:
$v^1_x=P_x+P_{v^{s'\!}}v^{s'\!}_x+P_{u_k}u_{k+1}$. (The index~$s'$ runs from~2 to~$p$.)
Taking into account the equations $v^s_x=\alpha^su$,
we split the differentiated condition with respect to $u_k$ step-by-step in reverse order,
beginning with the highest derivative. As a result, we obtain $P_{u_k}=0$, $P_x=0$ and $\alpha^1=P_{v^{s'\!}}\alpha^{s'\!}$, i.e., 
the functions~$\alpha^1$, \dots, $\alpha^p$ are linearly dependent over the ring of smooth functions of~$t$.
In view of Corollary~\ref{CorollaryOnLinearDependenceOfSolutionsOfLEvolE} (see below), this means that 
the functions~$\alpha^1$, \dots, $\alpha^p$ are linearly dependent in the usual sense,
contradicting the assumed independence of the conservation laws.
\end{proof}

Let $W(\varphi^1,\ldots,\varphi^l)$ denote the Wronskian of the functions $\varphi^1$, \ldots, $\varphi^l$
with respect to the variable~$x$, i.e. $W(\varphi^1,\ldots,\varphi^l)=\det(\varphi^j_{i-1})_{i,j=1}^{\;l}$.

\begin{lemma}\label{LemmaOnLinearDependenceOfSolutionsOfLEvolE}
The solutions~$\varphi^1=\varphi^1(t,x)$, \ldots, $\varphi^l=\varphi^l(t,x)$ of
a $(1+1)$-dimensional linear evolution equation $L\varphi=0$ of arbitrary order
are linearly dependent iff $W(\varphi^1,\ldots,\varphi^l)=0$.
\end{lemma}

\begin{proof}
Since the equation $L\varphi=0$ is linear and evolutionary, the operator $L$ is the sum of $\p_t$ and
a linear differential operator with respect to $x$ whose coefficients depend on $t$ and $x$.
If the functions $\varphi^1$, \ldots, $\varphi^l$ are linearly dependent then
the equality $W(\varphi^1,\ldots,\varphi^l)=0$ is obvious.
Let us prove the converse statement.

In the case $l=2$ the condition $W(\varphi^1,\varphi^2)=0$ implies $\varphi^2=C\varphi^1$,
where $C$ is a smooth function of~$t$. Acting on the latter equality with the operator $L$,
we obtain $C_t\varphi^1=0$, i.e. $C=\const$ or $\varphi^1=0$.
In any case the functions~$\varphi^1$ and~$\varphi^2$ are linearly dependent.

Suppose $W(\varphi^1,\ldots,\varphi^l)=0$.
Without loss of generality we can assume $W(\varphi^1,\ldots,\varphi^{l-1})\ne 0$.
(Otherwise we consider a smaller value of $l$.) Then
$\varphi^l=C^k\varphi^k,$
where $C^k$ are smooth functions of $t$ and the superscript $k$ runs from 1 to $l-1$.
Acting on the latter equality with the operator $L$ results in the equation $C^k_t\varphi^k=0$
which implies, in view of the condition $W(\varphi^1,\ldots,\varphi^{l-1})\ne 0$,
$C^k=\const$, which concludes the proof.
\end{proof}

\begin{corollary}\label{CorollaryOnLinearDependenceOfSolutionsOfLEvolE}
Solutions of a $(1+1)$-dimensional linear evolution equation of an arbitrary order are linearly dependent 
in the usual sense iff 
they are linearly dependent over the ring of smooth functions of~$t$.
\end{corollary}

For equation~\eqref{EqReducedLPE} the complete set of first level potential conservation laws
is indeed the union of conservation laws of systems~\eqref{eqGenPotSysOfParEqInRedForm}
corresponding to all possible values of~$p$ and $p$-tuples~$\bar\alpha$.

\begin{theorem}\label{TheoremOnPotConsLawsOfLPEs}
Any local conserved vector of system~\eqref{eqGenPotSysOfParEqInRedForm} is equivalent on the manifold 
of system~\eqref{eqGenPotSysOfParEqInRedForm} to a local conserved vector of equation~\eqref{EqReducedLPE}.
\end{theorem}

\begin{corollary}\label{CorollaryOnPotConsLawsOfLPEs}
For any linear $(1+1)$-dimensional second-order parabolic equation 
potential conserved vectors of any level are equivalent to local ones
on the manifolds of the corresponding potential systems,
and potentials of any level can locally be expressed via local variables $t$, $x$, $u_{(r)}$ (for some $r$)
and potentials of the first level~only.
\end{corollary}

\begin{corollary}\label{CorollaryOnPotSystemsOfLPEs}
For any linear $(1+1)$-dimensional second-order parabolic equation 
any potential system of higher level is equivalent, 
with respect to point transformations nontrivially acting only on potential variables,
to a potential system of the first level. 
\end{corollary}

In other words, the set of potential conservation laws of any linear $(1+1)$-dimensional second-order parabolic equation 
is exhausted by its local conservation laws. 
The set of locally independent potentials is exhausted by potentials of the first level.

Following~\cite{Popovych&Ivanova2004ConsLawsLanl} where these statements were derived for 
the linear heat equation $u_t=u_{xx}$, we present the proof of Theorem~\ref{TheoremOnPotConsLawsOfLPEs} 
in the form of a sequence of basic lemmas. 
In comparison with~\cite{Popovych&Ivanova2004ConsLawsLanl}, only minor modifications of the proof 
are needed due to the reduction of equations from class~\eqref{EqGenLPE} to the simpler from~\eqref{EqReducedLPE} 
by equivalence transformations. 

\begin{lemma}\label{LemmaOnLocConsLawsOfPotSystemForLHE}
Any local conservation law of system~\eqref{eqGenPotSysOfParEqInRedForm} is equivalent to the one with the conserved vector
$(Ku,K_xu-Ku_x)$ where the function $K=K(t,x,\bar v)$ is determined by the system
\begin{equation}\label{SystemOnK}
K_t+K_{xx}-VK=0, \qquad \alpha^sK_{xv^s}-\alpha^s_xK_{v^s}=0.
\end{equation}
\end{lemma}

\begin{proof}
Consider a local conservation law of system~\eqref{eqGenPotSysOfParEqInRedForm} in the most general form,
where the conserved vector is a vector-function of $t$, $x$ and derivatives of the functions $u$ and $v^s$
from order zero up to some finite number.
Taking into account system~\eqref{eqGenPotSysOfParEqInRedForm} and its differential consequences, we can exclude
the dependence of the conserved vector on any of the (nonzero order) derivatives of $v^s$ and the derivatives of~$u$
containing differentiation with respect to~$t$.
Similarly to Lemma~\ref{LemmaOnOrderOfConsLawsOfLPEs} we can prove that
the reduced conserved vector~$(F,G)$ does not depend on (nonzero order) derivatives of~$u$ and, moreover,
$F=F(t,x,\bar v)$, $G=-\alpha^sF_{v^s}(t,x,\bar v)u+G^0(t,x,\bar v)$. The functions~$F$ and $G^0$
satisfy the system
\[
\alpha^s\alpha^\sigma F_{v^sv^\sigma}=0, \qquad
\alpha^sG^0_{v^s}=2\alpha^s_xF_{v^s}+\alpha^sF_{xv^s}, \qquad
F_t+G^0_x=0.
\]
Let us pass on to the equivalent conserved vector $(\widetilde F,\widetilde G)$, where
$\widetilde F=F+D_xH$, $\widetilde G=G-D_tH$ and $H=H(t,x,\bar v)$ is a solution of the equations
$H_x=-F$, $H_t=G$. (The variables $v^s$ are considered as parameters in the latter equations.) Then
$\widetilde F=Ku$, $\widetilde G=K_xu-Ku_x$. The function $K=\alpha^sH_{v^s}$ depends on $t$, $x$ and $\bar v$
and satisfies system~\eqref{SystemOnK}.
\end{proof}

\begin{lemma}\label{LemmaOnDiffConsOfPotSystemForLHE}
Let the solutions~$\alpha^s=\alpha^s(t,x)$ and $\beta^s=\beta^s(t,x)$ of equation~\eqref{eqAdjEqOfLinParEqInRedForm}
satisfy the additional condition $\alpha^s_x\beta^s-\alpha^s\beta^s_x=0$.
Then for any $i,j\in{\mathbb N}\cup\{0\}$
\begin{equation}\label{EqAlphaBetaIJ}
\alpha^s_i\beta^s_j-\alpha^s_j\beta^s_i=0.
\end{equation}
Here and it what follows the subscripts $i$ and $j$ denote the $i$-th and $j$-th order derivatives with respect to $x$.
\end{lemma}

\begin{proof}
We carry out an induction with respect to the value~$i+j$.

Equation~\eqref{EqAlphaBetaIJ} is trivial for $i+j=0$, coincides with the additional condition for $i+j=1$ and
is obtained from this condition by means of differentiation with respect to $x$ if $i+j=2$.

Let us suppose that the Lemma's statement is true if $i+j=m-1$ and $i+j=m$
and prove it for $i+j=m+1$. Acting on equation~\eqref{EqAlphaBetaIJ} where $i+j=m-1$ with the operator $\p_t+\p_{xx}$
and taking into account the conditions $\alpha^s_t+\alpha^s_{xx}-V\alpha^s=0$ and $\beta^s_t+\beta^s_{xx}-V\beta^s=0$,
we obtain the equation $\alpha^s_{i+1}\beta^s_{j+1}-\alpha^s_{j+1}\beta^s_{i+1}=0.$ 
Therefore, the statement is true for $i'+j'=m+1$, $1\leqslant i',j'\leqslant m$ (here $i'=i+1$, $j'=j+1$).
It~remains to prove the case $i'=m+1$, $j'=0$ (or equivalently $i'=0$, $j'=m+1$).
For these values of $i'$ and $j'$ the statement is obtained
by subtracting the above equation $\alpha^s_m\beta^s_1-\alpha^s_1\beta^s_m=0$ from the results of
differentiating equation~\eqref{EqAlphaBetaIJ} where $i=m$, $j=0$ with respect to $x$.
\end{proof}

\begin{lemma}\label{LemmaOnWronskianForLHE}
If $\alpha^s_i\beta^s_j-\alpha^s_j\beta^s_i=0$ for $0\leqslant i<j\leqslant p$
then $W(\alpha^1,\ldots,\alpha^p,\beta^\sigma)=0$ for any $\sigma$.
\end{lemma}

\begin{proof}
Let $M^\sigma_{ij}$ denote the $(p-1)$-st order minor of $W(\bar\alpha,\beta^\sigma)$,
which is obtained by deleting the $\sigma$-th and $(p+1)$-st columns
corresponding to the functions~$\alpha^\sigma$ and~$\beta^\sigma$ and the $i$-th and $j$-th rows, 
where $0\leqslant i<j\leqslant p$. 
(For convenience the rows of $W(\bar\alpha,\beta^\sigma)$ are enumerated from 0 to $p$.)
We multiply the equation~$\alpha^s_i\beta^s_j-\alpha^s_j\beta^s_i=0$ by
$(-1)^{i+j+\sigma+p+1}M^\sigma_{ij}$ and convolve with respect to the index pair~$(i,j)$.
In view of the Laplace theorem on determinant expansion we obtain
$W(\bar\alpha|_{\alpha^s\rightsquigarrow\alpha^\sigma},\beta^{s})=0$.
Here the symbol~``$\rightsquigarrow$'' means that
the function~$\alpha^s$ is substituted instead of the function~$\alpha^\sigma$,
and summation over the index~$s$ is carried out.
This latter equation immediately implies the lemma's statement since 
$W(\bar\alpha|_{\alpha^s\rightsquigarrow\alpha^\sigma},\beta^s)=0$ for any fixed $s\ne\sigma$.
\end{proof}

\begin{lemma}\label{LemmaGenSolutionOfSystemOnK}
The general solution of system~\eqref{SystemOnK} can be written in the form
\begin{equation}\label{ExpressionForK}
K=\alpha^sH_{v^s}+\beta^0,
\end{equation}
where $H$ is an arbitrary smooth function of $\bar v$,
and $\beta^0=\beta^0(t,x)$ is an arbitrary solution of equation~\eqref{eqAdjEqOfLinParEqInRedForm}.
\end{lemma}

\begin{proof}
In view of Lemma~\ref{LemmaOnLocConsLawsOfPotSystemForLHE}
the functions~$\alpha^s$ and $\beta^s=K_{v^s}$ satisfy the conditions of Lemma~\ref{LemmaOnDiffConsOfPotSystemForLHE}
and, therefore, the ones of Lemma~\ref{LemmaOnWronskianForLHE}, and the variables~$\bar v$ are assumed as parameters. 
Since the $\alpha^s$ are linearly independent, this implies $K_{v^{\sigma}}=C^{\sigma s}\alpha^s$, 
where $C^{\sigma s}$ are smooth functions of the variables~$\bar v$ only.
The expressions for the cross derivatives
$K_{v^{\sigma}v^{\varsigma}}=C^{\sigma s}_{v^{\varsigma}}\alpha^s=C^{\varsigma s}_{v^{\sigma}}\alpha^s$
result in the equation $C^{\sigma s}_{v^{\varsigma}}=C^{\varsigma s}_{v^{\sigma}}$ which can easily be integrated:
$C^{\sigma s}=P^s_{v^{\sigma}}$ for some smooth function~$P^s$ of the variables~$\bar v$.
Substituting the expressions for~$C^{\sigma s}$ in the equations on $K$ and integrating, we obtain
$K=\alpha^sP^s+\beta^0$, where $\beta^0=\beta^0(t,x)$ is a solution of equation~\eqref{eqAdjEqOfLinParEqInRedForm}.
The latter equality and the equation $\alpha^sK_{xv^s}-\alpha^s_xK_{v^s}=0$ together imply
the equation
$(\alpha^{\sigma}_x\alpha^{\varsigma}-\alpha^{\sigma}\alpha^{\varsigma}_x)
(P^{\varsigma}_{v^{\sigma}}-P^{\sigma}_{v^{\varsigma}})=0$.
Analogously to Lemma~\ref{LemmaOnWronskianForLHE} we can state for any $i,j\in{\mathbb N}\cup\{0\}$
\begin{equation}\label{EqAlphaWij}
(\alpha^{\sigma}_i\alpha^{\varsigma}_j-\alpha^{\sigma}_j\alpha^{\varsigma}_i)
(P^{\varsigma}_{v^{\sigma}}-P^{\sigma}_{v^{\varsigma}})=0.
\end{equation}

Let $M^{\sigma'\varsigma'}_{ij}$ denote the $(p-2)$-nd order minor of $W(\bar\alpha)$,
obtained deleting the $\sigma'$-th and $\varsigma'$-th columns
corresponding to the functions~\raisebox{0ex}[0ex][0ex]{$\alpha^{\sigma'}$} 
and~\raisebox{0ex}[0ex][0ex]{$\alpha^{\varsigma'}$} and the $i$-th and $j$-th rows, 
where $0\leqslant i<j\leqslant p-1$. 
(Enumeration of the rows from 0 is used again for convenience.)
We multiply the equation~\eqref{EqAlphaWij} by $(-1)^{i+j+\sigma'+\varsigma'}M^{\sigma'\varsigma'}_{ij}$
and convolve with respect to the index pair~$(i,j)$.
The Laplace theorem on determinant expansion implies that
\begin{equation}\label{EqAlphaWij1}
W\bigl(\bar\alpha|_{
\alpha^{\sigma}\rightsquigarrow\alpha^{\sigma'}\!\!,\;
\alpha^{\varsigma}\rightsquigarrow\alpha^{\varsigma'}}\!
\bigr)
(P^{\varsigma}_{v^{\sigma}}-P^{\sigma}_{v^{\varsigma}})=0.
\end{equation}
Here the symbol~``$\rightsquigarrow$'' means that
the functions~$\alpha^{\sigma}$ and $\alpha^{\varsigma}$ are substituted
instead of the functions~$\alpha^{\sigma'}$ and $\alpha^{\varsigma'}$ respectively
and we have summation over the indices~$\sigma$ and $\varsigma$.
For any fixed $\{\sigma,\varsigma\}\ne\{\sigma',\varsigma'\}$ we have
$W\bigl(\bar\alpha|_{
\alpha^{\sigma}\rightsquigarrow\alpha^{\sigma'}\!\!,\;
\alpha^{\varsigma}\rightsquigarrow\alpha^{\varsigma'}}\!
\bigr)=0$.
Since $W(\bar\alpha)\ne 0$ in view of the linear independence of the functions~$\alpha^s$,
equation~\eqref{EqAlphaWij1} results in $P^{\varsigma'}_{v^{\sigma'}}-P^{\sigma'}_{v^{\varsigma'}}=0$,
i.e. $P^s=H_{v^s}$ for some smooth function~$H$ of~$\bar v$.
\end{proof}

In view of Lemma~\ref{LemmaGenSolutionOfSystemOnK}, the conserved vector $(Ku,K_xu-Ku_x)$
from Lemma~\ref{LemmaOnLocConsLawsOfPotSystemForLHE} has the form 
$(\beta^0u+D_xH,\beta^0_xu-\beta^0u_x-D_tH)$. 
Hence, it is equivalent to the conserved vector $(\beta^0u,\beta^0_xu-\beta^0u_x)$ 
which is also a conserved vector of equation~\eqref{EqReducedLPE}.

This completes the proof of Theorem~\ref{TheoremOnPotConsLawsOfLPEs}.

\section{Simplest potential symmetries}\label{SectionOnSimplestPotSymsOfLPEs}

Since there do not exist any pure potential conservation laws of linear $(1+1)$-dimensional second-order parabolic equations,
the study of potential symmetries for class~\eqref{EqGenLPE} is exhausted by 
the investigation of Lie symmetries of potential systems constructed with local conservation laws only. 
Such systems are called \emph{potential systems of the first level}. 
At first we classify so-called \emph{simplest potential symmetries}~\cite{Popovych&Ivanova2004ConsLawsLanl},
which are the Lie symmetries of the simplest potential systems associated with a single local conservation law.
As a result, inequivalent linear $(1+1)$-dimensional second-order parabolic equations 
possessing nontrivial simplest potential symmetries are completely described.
In the part on connections between different linear equations our presentation is very close to 
the theory of Darboux transformations for linear evolution equations~\cite{Matveev&Salle1991}.
The potential systems obtained by employing an arbitrary finite number of local conservation laws 
are investigated in Sections~\ref{SectionOnGenPotSysForLPEs} and~\ref{SectionOnGeneralPotSymsOfLPEs}. 
A thorough study of the simplest potential systems is necessary for understanding the general case since 
such systems are components of more general potential systems.

Consider an arbitrary equation of the form~\eqref{EqGenLPE}. 
Introducing the \emph{potential}~$v$ by the single conservation law of~\eqref{EqGenLPE} with 
the characteristic~$\alpha$, we obtain the potential system 
\begin{equation}\label{EqPotSysOfLPE}
v_x=\alpha u,\quad v_t=\alpha Au_x-((\alpha A)_x-\alpha B)u.
\end{equation}
(Let us recall that the characteristic~$\alpha=\alpha(t,x)$ is a solution of the adjoint equation~\eqref{EqAdjLPE}.) 
More precisely, for the conservation law we take the associated conserved vector of the canonical form~\eqref{eqCVofLPEs}. 
The following statement justifies our choice of conserved vectors for the construction of potential systems.  

\begin{lemma}
In determining the simplest potential symmetries of equations from the class~\eqref{EqGenLPE}, 
it is sufficient to consider only conserved vectors of the form~\eqref{eqCVofLPEs}, 
which are canonical representatives of the corresponding conservation laws. 
\end{lemma}

\begin{proof}
Only conserved vectors of the minimal order (with fluxes of order 0 and flows of order 1) lead 
to nontrivial potential symmetries of equations from class~\eqref{EqGenLPE}. 
Indeed, consider a conserved vector of a higher order of an equation of the form~\eqref{EqGenLPE}. 
It necessarily possesses the representation  
\begin{equation}\label{EqEquivConseredVectorsOfLPE}
\bigl(\alpha u+D_xH,\, -\alpha Au_x+((\alpha A)_x-\alpha B)u-D_tH\bigr),
\end{equation}
where $\alpha=\alpha(t,x)$ is a solution of the adjoint equation~\eqref{eqCVofLPEs} and $H$ is a differential function. 
In view of equation~\eqref{EqGenLPE} the function~$H$ can be assumed to depend on $t$, $x$ and derivatives of~$u$ 
with respect to only~$x$ up to an order~$r$, $r>0$. 
Let $Q=\tau\p_t+\xi\p_x+\eta\p_u+\theta\p_v$ be a Lie symmetry operator of the corresponding potential system 
\[
v_x=\alpha u+D_xH,\quad v_t=\alpha Au_x-((\alpha A)_x-\alpha B)u-D_tH.
\]
The coefficients of~$Q$ are functions of $t$, $x$, $u$ and $v$. 
Applying the infinitesimal invariance condition~\cite{Olver1986,Ovsiannikov1982} to the first equation of the system
and collecting coefficients of the unconstrained variable $v_{r+1}=\p^{r+1}v/\p x^{r+1}$, 
we obtain the equation $\eta_v-\tau_vu_t-\xi_vu_x=0$. 
After splitting it with respect to derivatives of~$u$, we have the condition $\eta_v=\tau_v=\xi_v=0$, i.e., 
the operator~$Q$ is projectable to the space of $(t,x,u)$. 

Each conserved vector of the minimal order (equal to~1) of equation~\eqref{EqGenLPE} possesses the representation~\eqref{EqEquivConseredVectorsOfLPE}, 
where the function~$H$ depends only on $t$, $x$ and~$u$ and, therefore, can be neglected due to the point transformation 
$\tilde t=t$, $\tilde x=x$, $\tilde u=u$ and $\tilde v=v-H$, which has, up to similarity, no influence on 
Lie symmetries of the potential system. 
This finally gives the canonical form~\eqref{eqCVofLPEs} of conserved vectors of equations from class~\eqref{EqGenLPE}. 
\end{proof}

The initial equation~\eqref{EqGenLPE} for $u$ is a differential consequence of system~\eqref{EqPotSysOfLPE}.
Another differential consequence of~\eqref{EqPotSysOfLPE} is the equation
\begin{equation}\label{EqPotLPE}
v_t=Av_{xx}+\left(B-A_x-2A\frac{\alpha_x}\alpha\right)v_x
\end{equation}
on the potential dependent variable~$v$, which is called the \emph{potential equation} 
associated with the linear parabolic equation~\eqref{EqGenLPE} and the characteristic~$\alpha$.
There is a one-to-one correspondence between solutions of the potential system and the potential equation 
due to the projection $(u,v)\to v$ on the one hand and due to the formula $u=v_x/\alpha$ on the other. 
The correspondence between solutions of the initial equation and the potential system is one-to-one only up to a constant summand. 

Any linear system of the general form 
\begin{equation}\label{EqGenFormOfPotSysOfLPE}
v_x=\alpha u,\quad v_t=\beta u_x+\gamma u, 
\end{equation}
where $\alpha$, $\beta$ and $\gamma$ are functions of $t$ and $x$, $\alpha\beta\ne0$, 
is the potential system of an equation from class~\eqref{EqGenLPE}. 
Moreover, the corresponding initial and potential equations are uniquely determined as its differential consequences. 
The coefficients of the initial equation are defined by the formulas 
\[
A=\frac\beta\alpha,\quad 
B=\frac{\beta_x+\gamma}\alpha,\quad
C=\frac{\gamma_x-\alpha_t}\alpha
\]
and, therefore, $\alpha$ is a characteristic of the conservation law associated with the potential system~\eqref{EqGenFormOfPotSysOfLPE} 
since it satisfies the adjoint equation~\eqref{EqAdjLPE}. 
The potential equation is represented in terms of $(\alpha,\beta,\gamma)$ as 
\begin{equation}\label{EqGenFormOfPotEqOfLPE}
v_t=\frac\beta\alpha v_{xx}+\left(\frac\gamma\alpha-\beta\frac{\alpha_x}{\alpha^2}\right)v_x.
\end{equation}

\begin{lemma}\label{LemmaOn1to1CorrespondenceBetweenLieSymsOfPotSysAndPotEqOfLPE}
Given arbitrary smooth functions $\alpha$, $\beta$ and $\gamma$ of the variables $t$ and $x$, where $\alpha\beta\ne0$,
the maximal Lie invariance algebras of system~\eqref{EqGenFormOfPotSysOfLPE} and equation~\eqref{EqGenFormOfPotEqOfLPE} are isomorphic. 
Namely, for any infinitesimal Lie symmetry operator $Q=\tau\p_t+\xi\p_x+\theta\p_v+\eta\p_u$ of system~\eqref{EqGenFormOfPotSysOfLPE} 
its projection $Q'=\tau\p_t+\xi\p_x+\theta\p_v$ to the variables $(t,x,v)$ is an infinitesimal Lie symmetry operator 
of equation~\eqref{EqGenFormOfPotEqOfLPE}. 
The coefficient~$\eta$ is expressed via coefficients of the first prolongation of the operator~$Q'$ with respect to~$x$. 
\end{lemma}

\begin{proof}
A heuristic argument in favour of the lemma is the one-to-one correspondence between the solutions of the system and the equation. 
This argument is not sufficient since a point transformation on a set of solutions may induce a non-point transformation on an 
equivalent set. That is why the best way to proceed is by direct calculation. 

The infinitesimal invariance criterion~\cite{Olver1986,Ovsiannikov1982} implies the determining equations for the 
coefficients of a symmetry operator~$Q$ of system~\eqref{EqGenFormOfPotSysOfLPE}: 
\begin{gather}\nonumber
\tau_x=\tau_u=\tau_v=0,\quad \xi_u=\xi_v=0,\quad \theta_u=\theta_{vv}=0,
\\\nonumber
(2\xi_x-\tau_t)\frac\beta\alpha=\tau\Bigl(\frac\beta\alpha\Bigr)_t+\xi\Bigl(\frac\beta\alpha\Bigr)_x,\quad
\theta_t=\frac\beta\alpha \theta_{xx}+\left(\frac\gamma\alpha-\beta\frac{\alpha_x}{\alpha^2}\right)\theta_x,
\\\nonumber
\xi_t+(2\theta_{vx}-\xi_{xx})\frac\beta\alpha
+\tau\left(\frac\gamma\alpha-\beta\frac{\alpha_x}{\alpha^2}\right)_t
+\xi\left(\frac\gamma\alpha-\beta\frac{\alpha_x}{\alpha^2}\right)_x
+(\tau_t-\xi_x)\left(\frac\gamma\alpha-\beta\frac{\alpha_x}{\alpha^2}\right)=0,
\\
\eta=\left(\theta_v-\xi_x-\tau\frac{\alpha_t}\alpha-\xi\frac{\alpha_x}\alpha\right)u+\frac{\theta_x}\alpha.
\label{EqDetEqForSymOpsOfGenFormOfPotSysOfLPEforEta}
\end{gather}
The equations $\tau_u=\xi_u=\theta_u=0$ guarantee that the operator~$Q$ is projectable to the variables $(t,x,v)$. 
The other determining equations excluding the last one form the complete system of determining equations for 
symmetry operators of equation~\eqref{EqGenFormOfPotEqOfLPE}.
For any Lie symmetry operator $Q'=\tau\p_t+\xi\p_x+\theta\p_v$ of equation~\eqref{EqGenFormOfPotEqOfLPE} 
the operator $Q=Q'+\eta\p_u$, where $\eta$ is defined by~\eqref{EqDetEqForSymOpsOfGenFormOfPotSysOfLPEforEta}, 
is a Lie symmetry operator of system~\eqref{EqGenFormOfPotSysOfLPE}. 
Equation~\eqref{EqDetEqForSymOpsOfGenFormOfPotSysOfLPEforEta} is re-written as 
\[
\eta=\left(\frac1\alpha\theta^x-\tau\frac{\alpha_t}{\alpha^2}v_x-\xi\frac{\alpha_x}{\alpha^2}v_x\right)
\bigg|_{\alpha u\rightsquigarrow v_x},
\]
where $\theta^x=\theta_x+\theta_vv_x-\xi_xv_x$,
i.e.\ $\eta$ is expressed via coefficients of the first prolongation of the operator~$Q'$ with respect to~$x$ 
in view of the equation $u=v_x/\alpha$.
\end{proof}

In view of Lemma~\ref{LemmaOn1to1CorrespondenceBetweenLieSymsOfPotSysAndPotEqOfLPE} 
the investigation of Lie symmetries of any potential system associated with an equation from class~\eqref{EqGenLPE}
is reduced to the consideration the corresponding potential equation. 
The only possibility for pure potential symmetries to appear is connected with the coefficient~$\eta$. 
Namely, the condition $\eta_v\ne0$ should be satisfied. 
Therefore, the Lie symmetry operator $Q'$ of equation~\eqref{EqGenFormOfPotEqOfLPE} 
induces a potential symmetry of the associated equation from class~\eqref{EqGenLPE} iff $\theta_{vx}\ne0$.

It is convenient to use the other dependent variable $w=v/\alpha$ instead of~$v$ in the further consideration. 
The function~$w$ will be called the \emph{modified potential} associated with the characteristic~$\alpha$. 
In terms of~$w$ the potential equation~\eqref{EqPotLPE} takes the form
\begin{equation}\label{EqModifiedPotLPE}
w_t=\widehat Aw_{xx}+\widehat Bw_x+\widehat Cw,
\end{equation}
which will be called the modified potential equation, associated with characteristic~$\alpha$. Here 
\[
\widehat A=A, \qquad
\widehat B=B-A_x, \qquad
\widehat C=-\frac{\alpha_t}\alpha+A\frac{\alpha_{xx}}\alpha+(B-A_x)\frac{\alpha_x}\alpha
-2A\left(\frac{\alpha_x}\alpha\right)^2.
\]
Introducing the function $\psi=1/\alpha$, we rewrite the equality for~$\widehat C$ in the form
\[
\widehat C=\frac{\psi_t-A\psi_{xx}-(B-A_x)\psi_x}\psi,
\]
i.e. $\psi$ is a solution of the same equation~\eqref{EqModifiedPotLPE} as~$w$. 
(The value $w=1/\alpha$ obviously is a solution of~\eqref{EqModifiedPotLPE} 
since $v=1$ is a solution of~\eqref{EqPotLPE}.)
The potential system~\eqref{EqPotSysOfLPE} can also be written in terms of $w$ and $\psi$ instead of~$v$ and~$\alpha$:
\begin{equation}\label{EqModifiedPotSysOfLPE}
w_x-\frac{\psi_x}\psi w=u,\quad w_t-\frac{\psi_t}\psi w=Au_x+\left(A\frac{\psi_x}\psi+B-A_x\right)u.
\end{equation}
The above representations of the potential equation and potential system are more suitable 
for the classification of Lie symmetries. 
In fact, the first equation of~\eqref{EqModifiedPotSysOfLPE} is the Darboux transformation~\cite{Matveev&Salle1991} 
of~\eqref{EqModifiedPotLPE} to~\eqref{EqGenLPE}. 
The Darboux transformation possesses the useful property of duality. 
We formulate this in a way slightly different from~\cite{Matveev&Salle1991}. 
Denote the Darboux transformation constructed with the nonzero function~$\psi$ by ${\rm DT}[\psi]$, i.e. 
\[
{\rm DT}[\psi](w)=w_x-\frac{\psi_x}\psi w.
\]

\begin{lemma}\label{LemmaOnDualDarbouxTrans}
Let $w^0$ be a fixed nonzero solution of~\eqref{EqModifiedPotLPE} and let
the Darboux transformation ${\rm DT}[w^0]$ map~\eqref{EqModifiedPotLPE} to equation~\eqref{EqGenLPE}. 
Then $\alpha^0=1/w^0$ is a solution of the equation~\eqref{EqAdjLPE} adjoint to equation~\eqref{EqGenLPE} 
and ${\rm DT}[\alpha^0]$ maps~\eqref{EqAdjLPE} to the equation adjoint to~\eqref{EqModifiedPotLPE}, i.e.,
\begin{gather*}\begin{array}{rcl}
u_t=Au_{xx}+Bu_x+Cu & \xleftarrow{{\rm DT}[w^0]} & w_t=\widehat Aw_{xx}+\widehat Bw_x+\widehat Cw
\\[-.4ex]
&\Updownarrow&
\\[.4ex]
\alpha_t+(A\alpha)_{xx}-(B\alpha)_x+C\alpha=0 & \xrightarrow{{\rm DT}[\alpha^0]} & 
\hat\alpha_t+(\widehat A\hat\alpha)_{xx}-(\widehat B\hat\alpha)_x+\widehat C\hat\alpha=0.
\end{array}
\end{gather*}
\end{lemma}

\begin{note}
${\rm DT}[\alpha^0]$ will be called \emph{dual} to the Darboux transformation ${\rm DT}[w^0]$. 
Since the twice adjoint equation coincides with the initial one, 
the twice dual Darboux transformation is nothing but the initial Darboux transformation. 
Moreover, `then' in the lemma can be replaced by `if and only if'. 
\end{note}

Lemma~\ref{LemmaOnDualDarbouxTrans} can also be reformulated in terms of characteristics of 
conservation laws. Denote equations~\eqref{EqGenLPE} and~\eqref{EqModifiedPotLPE} by~$\mathcal L$ 
and~$\widehat{\mathcal L}$ for convenience.

{\addtocounter{lemma}{-1}\renewcommand{\thelemma}{\arabic{lemma}$'$}
\begin{lemma}\label{LemmaOnDualDarbouxTransInTermsOfChars}
If $w^0$ is a nonzero solution of~$\widehat{\mathcal L}$ and ${\rm DT}[w^0](\widehat{\mathcal L}\,)=\mathcal L$ 
then $\alpha^0=1/w^0\in\Ch_{\rm f}(\mathcal L)$ and 
${\rm DT}[\alpha^0]\colon\Ch_{\rm f}(\mathcal L)\to\Ch_{\rm f}(\widehat{\mathcal L}\,)$. 
\end{lemma}}

\begin{note}\label{NoteOnDarbouxTransAsLinearMappingOfSolutionSpaceOfLPEs}
For any nonzero solution~$\psi$ of~\eqref{EqModifiedPotLPE} the Darboux transformation ${\rm DT}[\psi]$
is a linear mapping of the solution space of~\eqref{EqModifiedPotLPE} to the solution space of~\eqref{EqGenLPE}.
The kernel of this mapping coincides with the linear span~$\langle\psi\rangle$. 
Its image is the whole solution space of~\eqref{EqGenLPE}. Indeed, for any solution~$u$ of~\eqref{EqGenLPE} 
we can find a solution~$w$ of~\eqref{EqModifiedPotLPE}, mapped to~$u$, by integrating
system~\eqref{EqModifiedPotSysOfLPE} with respect to~$w$. 
System~\eqref{EqModifiedPotSysOfLPE} is compatible in view of equation~\eqref{EqGenLPE}.
Therefore, ${\rm DT}[\psi]$ generates a one-to-one linear mapping between 
the solution space of~\eqref{EqModifiedPotLPE}, factorized by~$\langle\psi\rangle$, 
and the solution space of~\eqref{EqGenLPE}. 
\end{note}

In view of Proposition~\ref{PropositionOnTransOfCharsOfCLsForLPEs} any point equivalence transformation in 
class~\eqref{EqGenLPE} is prolonged to characteristics of conservation laws of equations from this class. 
Then it can be prolonged to the first derivatives of potentials due to potential system~\eqref{EqPotSysOfLPE}, 
and under this prolongation $\tilde v_x=v_x$, $\tilde v_t=v_t$. 
Excluding trivial translations of the potential variable~$v$ with constant summands, we assume that 
$v$ is not transformed under prolongation of the equivalence transformations. 
The extension to the functions~$w$ and~$\psi$ is obvious. 
As a result, we obtain simultaneous equivalence transformations between 
the initial, potential and modified potential equations:
\begin{gather}\nonumber
\tilde t=T(t),\quad \tilde x=X(t,x), \quad \tilde u=U^1(t,x)u, \\ 
\tilde\alpha=\frac{\alpha}{X_xU^1}, \quad \tilde v=v, \quad \tilde w=X_xU^1w, \quad \tilde \psi=X_xU^1\psi.
\label{EqProlongedEquivTransToPotsForLPE}
\end{gather}
where $T_tX_xU^1\ne0$.

Transformations~\eqref{EqProlongedEquivTransToPotsForLPE} preserve the determining equations derived in the proof of 
Lemma~\ref{LemmaOn1to1CorrespondenceBetweenLieSymsOfPotSysAndPotEqOfLPE} as well as
both the conditions $\theta_{vx}\ne0$ and $\theta_{vx}=0$, where $\theta$ is the coefficient of~$\p_v$ in 
symmetry operators of potential systems or potential equations. 
This is why Lie and pure potential symmetries of any equation from class~\eqref{EqGenLPE} are not mixed 
under transformations~\eqref{EqProlongedEquivTransToPotsForLPE}, 
and the dimension of the factor-space of potential symmetry operators 
constructed with a single characteristic with respect to the subspace of Lie symmetry operators is not changed. 
Therefore, potential symmetries of equations from class~\eqref{EqGenLPE} can be studied up to the equivalence relation 
generated by transformations from~$G^\sim$.
In particular, it is sufficient to consider only the reduced form~\eqref{EqReducedLPE} of linear parabolic equations.
The corresponding potential system and potential equation and their modifications also are simplified to
\begin{gather} \label{EqPotSysOfReducedLPE}
v_x=\alpha u,\quad v_t=\alpha u_x-\alpha_xu,
\\ \nonumber 
v_t-v_{xx}+2\frac{\alpha_x}\alpha v_x=0,
\\ \label{EqModifiedReducedPotLPE}
w_t-w_{xx}+Pw=0,
\\ \label{EqModifiedReducedPotSysOfLPE}
w_x-\frac{\psi_x}\psi w=u,\quad w_t-\frac{\psi_t}\psi w=u_x+\frac{\psi_x}\psi u.
\end{gather}
where $\alpha=\alpha(t,x)$ is a fixed solution of the reduced adjoint equation $\alpha_t+\alpha_{xx}-V\alpha=0$, 
$\psi=\psi(t,x)$ is a fixed solution of~\eqref{EqModifiedReducedPotLPE}, $V=V(t,x)$ and 
\[
P=V-2\left(\frac{\alpha_x}\alpha\right)_x=V+2\Bigl(\frac{\psi_x}\psi\Bigr)_x.
\]

The equivalence transformations~\eqref{EqTransFromEquivGroupOfReducedLPEs} in the reduced class~\eqref{EqReducedLPE} 
form a subgroup in the equivalence group~$G^\sim$ under special restrictions 
on the parameter-functions~$X$ and $U^1$. 
They are prolonged by formulas~\eqref{EqProlongedEquivTransToPotsForLPE} to equivalence transformations of 
the whole reduced potential frame. 
Hence, the classification of potential symmetries of class~\eqref{EqReducedLPE} follows from 
the group classification of the same class in terms of $(w,P)$ instead of~$(u,V)$.

Consider a Lie symmetry operator $Q'=\tau\p_t+\xi\p_x+\zeta\p_w$ of equation~\eqref{EqModifiedReducedPotLPE}. 
The coefficients of~$Q'$ are functions of $t$, $x$ and $w$.
The infinitesimal invariance criterion~\cite{Olver1986,Ovsiannikov1982} applied to equation~\eqref{EqModifiedReducedPotLPE}
implies the conditions 
\[
\tau=\tau(t), \quad \xi=\frac12\tau_tx+\sigma(t), \quad \zeta=\zeta^1(t,x)w+\zeta^0(t,x),\quad
\zeta^1=-\frac18\tau_{tt}x^2-\frac12\sigma_tx+\varkappa(t)
\]
which do not involve the arbitrary element~$P$, and the classifying equations
\[
\zeta^1_t-\zeta^1_{xx}+\tau P_t+\xi P_x+\tau_t P=0,\quad \zeta^0_t-\zeta^0_{xx}+P\zeta^0=0.
\]
Due to the first equation of system~\eqref{EqModifiedReducedPotSysOfLPE}, 
the coefficient~$\eta$ of~$\p_u$ of the corresponding Lie symmetry operator~$Q$ of system~\eqref{EqModifiedReducedPotSysOfLPE}
can be expressed via coefficients of the first prolongation of the operator~$Q'$ with respect to~$x$:
\[
\eta=\left(\zeta^1_x-\frac12\tau_t\frac{\psi_x}\psi-\tau\Bigl(\frac{\psi_x}\psi\Bigr)_t
-\xi\Bigl(\frac{\psi_x}\psi\Bigr)_x\right)w
+\left(\zeta^1-\frac12\tau_t\right)u+\zeta^0_x-\frac{\psi_x}\psi\zeta^0.
\]
It is obvious that the operator~$Q'$ induces a pure potential symmetry of equation~\eqref{EqReducedLPE} iff $\eta_w\ne0$, i.e., 
\begin{equation}\label{EqExistanceOfPotSymOfReducedLPE}
\tau\varrho_t+\xi\varrho_x+\frac12\tau_t\varrho\ne\zeta^1_x,
\end{equation}
where $\varrho=\psi_x/\psi$. This implies that trivial Lie invariance operators of the potential equations 
generate no pure potential symmetry operators of the initial equations. 
More precisely, the Darboux transformation ${\rm DT}[\psi]$ maps 
the trivial part of the Lie invariance algebra of~\eqref{EqModifiedPotLPE}/\eqref{EqModifiedReducedPotSysOfLPE} 
onto the trivial part of the Lie invariance algebra of~\eqref{EqGenLPE}/\eqref{EqReducedLPE}. 
We therefore obtain the following result: 

\begin{proposition}\label{PropositionOnNecCondOfExistOfPotSymOfReducedLPE}
Equation~\eqref{EqReducedLPE} admits a pure potential symmetry associated with the characteristic~$\alpha$ only 
if the corresponding equation~\eqref{EqModifiedReducedPotLPE} has a nontrivial Lie symmetry, i.e.,
$P$ is equivalent to a stationary function.
\end{proposition}

It may happen that nontrivial Lie invariance operators of the potential equations do not generate 
any pure potential symmetry operators of the initial equations. 
We provide the symmetry criterion for Lie invariance operators and solutions of 
the potential equation leading to pure potential symmetries. 
Its formulation is independent of the equivalence transformations. 

\begin{proposition}\label{PropositionOnSymCondOfExistOfPotSymOfLPE}
Suppose that $\,Q'=\tau\p_t+\xi\p_x+\zeta^1w\p_w$ is a nontrivial Lie invariance operator of equation~\eqref{EqModifiedPotLPE}, 
$\psi$ is a solution of~\eqref{EqModifiedPotLPE} and ${\rm DT}[\psi]$ transforms equation~\eqref{EqModifiedPotLPE} to~\eqref{EqGenLPE}. 
The operator~$Q'$ induces a pure potential symmetry operator of equation~\eqref{EqGenLPE} 
iff the following obviously equivalent conditions are satisfied, where $Q'[\psi]=\zeta^1\psi-\tau\psi_t-\xi\psi_x$:

1) $Q'[\psi]$ and $\psi$ are linearly independent;

2) $\psi$ is not an eigenfunction of~$Q'$;

3) ${\rm DT}[\psi](Q'[\psi])$ is a non-vanishing solution of equation~\eqref{EqGenLPE};

4) $\psi$ is not invariant with respect to the operator $Q'-\lambda w\p_w$ for some constant~$\lambda$. 
\end{proposition}

\begin{proof}
All conditions in the proposition are stable under equivalence transformations. 
It is enough to consider the reduced equations~\eqref{EqReducedLPE} and~\eqref{EqModifiedReducedPotSysOfLPE}
instead of~\eqref{EqModifiedPotLPE} and~\eqref{EqGenLPE}. 
Since $Q'[\psi]$ is the result of the action of the Lie invariance operator~$Q$ of equation~\eqref{EqModifiedReducedPotSysOfLPE} on 
the solution~$\psi$ of the same equation then $Q'[\psi]$ also is a solution of~\eqref{EqModifiedReducedPotSysOfLPE} 
and hence ${\rm DT}[\psi](Q'[\psi])$ is a solutions of~\eqref{EqReducedLPE}.
Let $Q$ denote the prolongation of~$Q'$ to~$u$.
The expression for $\eta_w$, where $\eta$ is the coefficient of~$\p_u$ in~$Q$, is rewritten in terms of $Q'[\psi]$:
\[
\eta_w=\left(\frac{Q'[\psi]}{\psi}\right)_x=\frac1\psi{\rm DT}[\psi](Q'[\psi]).
\]
Therefore, $\eta_w=0$ iff $W(\psi,Q'[\psi])=0$, i.e.\ ${\rm DT}[\psi](Q'[\psi])=0$. 
In view of Lemma~\ref{LemmaOnLinearDependenceOfSolutionsOfLEvolE}, the condition $W(\psi,Q'[\psi])=0$ is equivalent to the fact that 
the functions $Q'[\psi]$ and $\psi$ are linearly dependent, i.e.\ $Q'[\psi]=\lambda\psi$ for some constant~$\lambda$ or 
$\hat Q'_\lambda[\psi]=0$, where $Q'_\lambda=Q'-\lambda w\p_w$. 
For any $\lambda$ the operator $Q'_\lambda$ belongs to the Lie invariance algebra of equation~\eqref{EqModifiedReducedPotSysOfLPE}.
\end{proof}

In view of Lemma~\ref{LemmaOn1to1CorrespondenceBetweenLieSymsOfPotSysAndPotEqOfLPE}
pure potential symmetries obviously exist if the dimension of the essential Lie invariance algebra of the potential equation 
is greater than the corresponding dimension for the initial equation. As shown by the statements below and the example 
considered in Section~\ref{SectionOnSimpestPotSymsOfLHE}, this condition is not necessary. 

\begin{theorem}\label{TheoremOnPotSymOfLPEs}
A linear second-order parabolic equation admits simplest pure potential symmetries 
iff it is equivalent with respect to point equivalence transformations to an equation from class~\eqref{EqReducedLPE} in which 
\[
V=P(x)-2\Bigl(\frac{\psi_x}\psi\Bigr)_x, 
\]
where $\psi=\psi(t,x)$ is a nonzero solution of the equation $\psi_t-\psi_{xx}+P(x)\psi=0$ and 
either $P=\mu x^{-2}$, $\mu=\const$, or $(\psi_x/\psi)_t\ne0$ (this latter condition is equivalent to $V_t\ne0$) 
if $P$ is inequivalent to~$\mu x^{-2}$ with respect to point equivalence transformations. 
In the case $P=0$ the initial equation possesses at least two potential symmetry operators 
which are linearly independent up to Lie symmetries. 
The associated characteristic equals $1/\psi$.
\end{theorem}

\begin{proof}
Let us briefly recall the main results of the above consideration which form the basis of our proof. 
The transformations from~$G^\sim$ generate 
a well-defined equivalence relation on the whole frame of simplest potential symmetries 
of equations from class~\eqref{EqGenLPE}. 
The investigation of Lie symmetries of potential systems is reduced to 
the group classification of modified potential equations forming the same class 
as the initial equations and are connected with them via the Darboux transformation. 
To complete the proof, we have to verify condition~\eqref{EqExistanceOfPotSymOfReducedLPE} 
for all the cases of the Lie--Ovsiannikov classification (see Theorem~\ref{TheoremOnGroupClassificationOfLPEs}). 

Let $P=\mu x^{-2}$. Suppose that any of the operators $\p_t$, $D$ and~$\Pi$ does not satisfy 
condition~\eqref{EqExistanceOfPotSymOfReducedLPE} for some~$\psi$. 
Then $\varrho_t=0$, $x\varrho_x+\varrho=0$ and $4tx\varrho_x+4t\varrho=-2x$, which implies the contradiction $0=x$. 
Hence for any solution of equation~\eqref{EqModifiedReducedPotLPE} 
at least one operator among $\p_t$, $D$ and~$\Pi$ satisfies condition~\eqref{EqExistanceOfPotSymOfReducedLPE} 
and, therefore, induces a pure potential symmetry operator of~\eqref{EqReducedLPE}.

In the case $\mu=0$ at least one operator among $\p_x$ and~$G$ 
induces a pure potential symmetry operator of~\eqref{EqReducedLPE} 
since otherwise the equations $\varrho_x=0$ and $2t\varrho_x=-1$ would imply the contradiction $0=-1$. 
Therefore, for $P=0$ we have at least two independent pure potential symmetry operator of~\eqref{EqReducedLPE}.

There is only one independent nontrivial Lie symmetry operator~$\p_t$ for the general value $P=P(x)$. 
If it does not induce a pure potential symmetry operator of~\eqref{EqReducedLPE} then 
$\varrho_t=0$ and, therefore, $V_t=P_t+\varrho_{xt}=0$. 
Conversely, suppose that $P_t=V_t=0$ and $\varrho_{t}\ne0$ simultaneously. Then $\varrho_{xt}=0$. 
If equation~\eqref{EqModifiedReducedPotLPE} has a solution~$\psi$ satisfying the conditions 
$(\psi_x/\psi)_t\ne0$ and $(\psi_x/\psi)_{tx}=0$ then it is equivalent to the linear heat equation 
(the case $P=0$).
\end{proof}

Based on Theorem~\ref{TheoremOnPotSymOfLPEs}, we can formulate symmetry criteria on the existence 
of simplest potential symmetries without involving equivalence transformations. 

\begin{corollary}
A linear second-order parabolic equation admits simplest pure potential symmetries 
associated with its characteristic~$\alpha$ only if 
the corresponding potential equation possesses nontrivial Lie symmetry operators.
If the potential equation has more than one (three) independent nontrivial Lie symmetry operators then 
the initial equation admits at least one (two) independent simplest pure potential symmetry operators.
\end{corollary}

\section{Simplest potential symmetries of the linear heat equation}\label{SectionOnSimpestPotSymsOfLHE}

Theorem~\ref{TheoremOnPotSymOfLPEs} gives description of equations from class~\eqref{EqGenLPE}, 
having nontrivial simplest potential symmetries.
At the same time, there is another problem concerning simplest potential symmetries: 
Given an equation from class~\eqref{EqGenLPE}, to describe all its characteristics leading to 
its nontrivial simplest potential symmetries.
In what follows we consider this problem in detail for the linear heat equation
\begin{equation}\label{EqLHE}
u_t=u_{xx}.
\end{equation}

We recall that the maximal Lie invariance algebra of the linear heat equation is
\begin{gather*}
\mathfrak g_0=\langle
\p_t,\ \p_x,\ 2t\p_t+x\p_x,\ 2t\p_x-xu\p_u,\ 4t^2\p_t+4tx\p_x-(x^2+2t)u\p_u,\ u\p_u,\ f\p_u\rangle.
\end{gather*}
Here the function $f=f(t,x)$ runs through the solution set of this equation. 
The operators from~$\mathfrak g_0$ generate the continuous symmetry transformations of~\eqref{EqLHE}.
See, e.g., \cite{Olver1986} for their explicit form. 
Equation~\eqref{EqLHE} also possesses a nontrivial group of discrete symmetry transformations 
generated by two involutive transformations of alternate sign $(t,x,u)\to(t,-x,u)$ and $(t,x,u)\to(t,x,-u)$.
The point symmetry group~$G_0$ of~\eqref{EqLHE} is generated by both the continuous and discrete symmetries. 
The most general solution obtainable from a given solution $u=u(t,x)$ by transformations from~$G_0$ is of the form
\[
\tilde u=\frac{\varepsilon_3}{\sqrt{1+4\varepsilon_6t}}e^{-\frac{\varepsilon_5x+\varepsilon_6x^2-\varepsilon_5^2t}{1+4\varepsilon_6t}}
u\left(\frac{\varepsilon_4^2t}{1+4\varepsilon_6t}-\varepsilon_2,\frac{\varepsilon_4(x-2\varepsilon_5t)}{1+4\varepsilon_6t}-\varepsilon_1\right)+f(t,x),
\]
where $\varepsilon_1,\ldots,\varepsilon_6$ are arbitrary constants, 
$\varepsilon_3\varepsilon_4\not=0$ and $f=f(t,x)$ is an arbitrary solution of the linear heat equation.
The essential part~$G_0^{\rm ess}$ of $G_0$ is formed by the transformations with $f\equiv0$.

In view of Corollary~\ref{CorollaryOn2DPotSystemsEquivRelGeneratedBySymGroup} 
the transformations from~$G_0^{\rm ess}$ prolonged by formulas~\eqref{EqProlongedEquivTransToPotsForLPE}
form equivalence groups on the sets of potential systems and potential equations 
associated with equation~\eqref{EqLHE} and single characteristics. 
We identify these equivalence groups with~$G_0^{\rm ess}$. 
It is natural to investigate problems on potential symmetries of~\eqref{EqLHE} 
up to the equivalence relation generated by~$G_0^{\rm ess}$.

\begin{lemma}\label{LemmaOnCharsEquiv1OfLHE}
The characteristic~$\alpha$ of the linear heat equation gives an equation of the form~\eqref{EqModifiedReducedPotLPE} with $P_x=0$ 
iff $\alpha=1\bmod G_0^{\rm ess}$. 
\end{lemma}
\begin{proof}
Let $P_x=0$ in the equation~\eqref{EqModifiedReducedPotLPE} corresponding to the characteristic~$\alpha$. 
Integrating the equation $-2(\alpha_x/\alpha)_x=P$ as an ordinary differential equation on~$\alpha$ with the parameter~$t$, 
we obtain the expression $\alpha=\theta\exp(Px^2/4+\zeta)$, where $P$, $\theta$ and~$\zeta$ are functions of~$t$ and $\theta\ne0$. 
Substituting of the expression for~$\alpha$ to the the backward linear heat equation $\alpha_t+\alpha_{xx}=0$ 
and the subsequently splitting with respect to~$x$ result in the system of ordinary differential equations
\[
P_t=-P^2, \quad \zeta_t=-P\zeta, \quad \theta_t=\left(\zeta^2+\frac12P\right)\theta.
\]
In integrating this system, two cases are to be distinguished, corresponding to two cases for the characteristics:
\[
V=0,\quad \alpha=\delta_2\exp(\delta_1x-\delta_1^2t)
\quad\mbox{and}\quad 
V=\frac1{t+\delta_0},\quad \alpha=\delta_2\exp\frac{(x+\delta_1)^2}{4(t+\delta_0)},
\]
where $\delta_0$, $\delta_1$ and $\delta_2$ are constants. In both these cases the characteristics 
are $G_0^{\rm ess}$-equivalent to $\alpha=1$. 

The converse statement is obvious.
\end{proof}

\begin{theorem}\label{TheoremOnCharsGivenPotSymsOfLHE}
Lie symmetries of system~\eqref{EqPotSysOfReducedLPE}, 
where $\alpha=\alpha(t,x)$ is a (fixed) solution of the backward linear heat equation $\alpha_t+\alpha_{xx}=0$, induce
pure potential symmetries of the linear heat equation iff $\,\alpha\in\{1,x\}\bmod G_0^{\rm ess}$.
\end{theorem}

\begin{proof}
In view of Proposition~\ref{PropositionOnNecCondOfExistOfPotSymOfReducedLPE}, a characteristic~$\alpha$
leads to pure potential symmetries of~\eqref{EqLHE} only if 
the function $P=-2(\alpha_x/\alpha)_x$ is $G^\sim_1$-equivalent to a function not depending on~$t$. 
Let us apply a transformation~$\mathcal T$ from~$G^\sim_1$, prolonged to the whole frame under consideration, 
such that $\tilde P_t=0$. 
Equation~\eqref{EqLHE} is mapped by~$\mathcal T$ to the equation $\tilde u_t-\tilde u_{xx}+\tilde V\tilde u=0$, where 
$\tilde V_{xxx}=0$. (We omit tildes over the transformed variables~$t$ and~$x$ for convenience.) 
The functions~$\tilde V$ and~$\tilde P$ are connected by the relation 
$\tilde P-\tilde V=-2(\tilde\alpha_x/\tilde\alpha)_x=2(\tilde\psi_x/\tilde\psi)_x$, 
where $\tilde\alpha$ is the transformed characteristic and $\tilde\psi=1/\tilde\alpha$.
Considering~$t$ as a parameter, we integrate this relation as an ordinary differential equation on~$\tilde\psi$:
$\tilde\psi=\pm e^{(\int H\,dx-R)/2}$, where $H$ is a smooth function of~$x$ such that $H'=\tilde P$ 
and $R$ is a fourth-degree polynomial of~$x$ with coefficients depending on~$t$ such that $R_{xx}=\tilde V$.
Substituting the obtained expression for~$\tilde\psi$ into the equation 
$\tilde\psi_t-\tilde\psi_{xx}+\tilde P\tilde\psi=0$ gives the following relation between~$H$ and~$R$:
$
H'-\frac12H^2+R_xH=R_t-R_{xx}+\frac12R_x{}^2.
$
The differential consequence 
\[
R_{tx}H=R_{tt}-R_{txx}+R_xR_{tx}
\]
of this relation with respect to~$\p_t$ is essential for the further proof.
There are two possible cases which should be considered separately. 

If $R_{tx}\ne0$, we express~$H$ from the differential consequence: 
\[
H=\frac{R_{tt}}{R_{tx}}-\frac{R_{txx}}{R_{tx}}+R_x,
\] 
and then differentiate once more with respect to~$t$. 
Since $H_t=0$, we derive an equation only in~$R$, which can be split with respect to~$x$ 
since $R$ is a polynomial of~$x$. The resulting equations on the coefficients of~$R$ imply 
that $R_{txx}=0$ and $(R_{ttx}/R_{tx})_t=0$. 
Therefore, $H$ is a third-order polynomial of~$x$ and 
\[
\tilde P-\tilde V=H'-R_{xx}=\frac{R_{ttx}}{R_{tx}}=\const.
\]
This means that $P$ is a function only of~$t$ in the old variables. 
Hence, $\alpha=1\bmod G_0^{\rm ess}$ in view of Lemma~\ref{LemmaOnCharsEquiv1OfLHE}.

The condition $R_{tx}=0$ implies $R_{tt}=0$ and $\tilde\psi=\tilde\psi^1(x)e^{-\nu t}$, where $\nu$ is a constant 
and the function~$\tilde\psi^1$ depends only on~$x$. 
Since $(\tilde\psi_x/\tilde\psi)_t=0$, 
in view of Theorem~\ref{TheoremOnPotSymOfLPEs} the function~$\tilde P$ should be $G^\sim_1$-equivalent to 
the function~$\mu x^{-2}$ with a constant~$\mu$. Otherwise there are no nontrivial potential symmetries associated 
with the characteristic~$\alpha$. Therefore, 
\begin{gather*}
\tilde V=\kappa_2x^2+\kappa_1x+\kappa_0,\qquad
\tilde P=\frac\mu{x^2}+\mu_2x^2+\mu_1x+\mu_0,
\\
\tilde\alpha=\frac1\psi=\lambda_0|x|^{-\mu/2}\exp\left(
\frac{\kappa_2-\mu_2}{24}x^4+\frac{\kappa_1-\mu_1}{12}x^3+\frac{\kappa_0-\mu_0}{4}x^2+\lambda_1x+\nu t\right), 
\end{gather*}
where $\mu,\mu_i,\kappa_i,\lambda_1,\lambda_0=\const$, $i=1,2,3$. 
Then the equation $\tilde\alpha_t+\tilde\alpha_{xx}-\tilde V\tilde\alpha=0$ implies that 
\begin{gather*}
\kappa_2=\mu_2,\quad \kappa_1=\mu_1,\quad \mu(\mu+2)=0,\quad \mu\lambda=0, \quad 
(1-\mu)(\kappa_0-\mu_0)=2(\kappa_0-\nu-\lambda^2),\\
\kappa_1=\lambda(\kappa_0-\mu_0),\quad 4\kappa_2=(\kappa_0-\mu_0)^2.
\end{gather*}
 The condition $\mu=0$ leads to the case $\tilde P-\tilde V=\const$ considered above. 
If $\mu\ne0$ then it follows from the above equations that $\mu=-2$, $\lambda=\kappa_1=\mu_1=0$, $\nu=(\kappa_0+\mu_0)/2$, i.e., 
\[
\tilde\alpha=\hat\lambda_0x\exp\left(\frac14(\kappa_0-\mu_0)^2x^2+\frac{\kappa_0+\mu_0}2t\right).
\]
After returning to the old variables, we have $\alpha=x$. 
(The explicit form of the variable transformation depends on the value of $\kappa_0-\mu_0$.)
\end{proof}

Potential symmetries associated with the characteristic $\alpha=1$ were studied by a number 
of authors~\cite{Bluman&Kumei1989,Sophocleous1996,Popovych&Ivanova2003PETs}.
The corresponding potential~$v^1$ is defined by the system $v^1_x=u$, $v^1_t=u_x$.
Its maximal Lie invariance algebra is
\begin{gather*}
\mathfrak p_1=\langle
\p_t,\ \p_x,\ 2t\p_x-(xu+v^1)\p_u-xv^1\p_{v^1},\ 2t\p_t+x\p_x-u\p_u, \\
\qquad\ 4t^2\p_t+4tx\p_x-((x^2+6t)u+2xv^1)\p_u-(x^2+2t)v^1\p_{v^1},\ u\p_u+v^1\p_{v^1},\ f_x\p_u+f\p_{v^1}\rangle.
\end{gather*}
The potential equation $v^1_t=v^1_{xx}$ has the same form as the initial equation~\eqref{EqLHE}.
That is why the algebras~$\mathfrak g_0$ and $\mathfrak p_1$ are isomorphic~\cite{Popovych&Ivanova2003PETs}. 
The basis operators of~$\mathfrak p_1$ are obtained from the basis operators of~$\mathfrak g_0$ 
by re-denoting $u\to v^1$ and then carrying out the first prolongation with respect to~$x$ in view of $u=v^1_x$. 
Any linear combination of operators from $\mathfrak p_1$ which contains the third or fifth basis operator is
a pure potential symmetry operator of the linear heat equation. 

The case of the simplest nonconstant characteristic $\alpha=x$ was studied in~\cite{Ivanova&Popovych2007CommentOnMei}.
The corresponding potential system $v^2_x=u$, $v^2_t=xu_x-u$ possesses 
the maximal Lie invariance algebra
\begin{gather*}
\mathfrak p_2=\langle
\p_t,\ 2t\p_t+x\p_x-2u\p_u,\ 4t^2\p_t+4tx\p_x-((x^2+6t)u+2v^2)\p_u-(x^2-2t)v^2\p_{v^2},\\
\qquad\ u\p_u+v^2\p_{v^2},\ x^{-1}h_x\p_u+h\p_{v^2}\rangle,
\end{gather*}
where the function $h=h(t,x)$ runs through the set of solutions of the potential equation 
$v^2_t+2x^{-1} v^2_x-v^2_{xx}=0$.
Any linear combination of operators from $\mathfrak p_2$ which contains the third basis operator is
a pure potential symmetry operator of the linear heat equation.

Note that the pure potential symmetry operators from the algebra~$\mathfrak p_2$ 
differ from the ones from the algebra~$\mathfrak p_1$ by
both the explicit form and the nature of the potential variable~$v^2$ associated with the characteristic $\alpha=x$
in contrast to~$v^1$ associated with the characteristic $\alpha=1$.

\begin{corollary}
All $G_0^{\rm ess}$-equivalent simplest pure potential symmetries of the linear heat equation~\eqref{EqLHE}
are exhausted by the operators from~$\mathfrak p_1$ and~$\mathfrak p_2$ satisfying the condition $\eta_v\ne0$, 
where $\eta$ is the coefficient of~$\p_u$ and $v$ is the corresponding potential ($v^1$ or $v^2$).
\end{corollary}

Let us demonstrate how to extend the results obtained in this section 
to equations which are equivalent to the heat equation with respect to point transformations. 
Consider the Fokker--Planck equation 
\begin{equation}\label{EqFPx}
\tilde u_{\tilde t}=\tilde u_{\tilde x\tilde x}+(\tilde x\tilde u)_{\tilde x}. 
\end{equation}
Its maximal Lie invariance algebra is
\begin{gather*}
\tilde{\mathfrak g}_0=\langle\p_t,\ e^{-t}\p_x,\,e^{-2t}\p_t-e^{-2t}x\p_x+e^{-2t}u\p_u,\ e^t\p_x-e^txu\p_u,\\
\qquad\ e^{2t}\p_t+e^{2t}x\p_x-e^{2t}x^2u\p_u,\ u\p_u,\ f\p_u \rangle, 
\end{gather*}
where the function $f=f(t,x)$ runs through the solution set of the same equation.
(We omit tildes over variables when it is understandable that they originate from the Fokker--Planck equation.)
Equation~\eqref{EqFPx} is reduced to equation~\eqref{EqLHE} by the point transformation
\[
\mathcal T\colon\quad t=\frac12e^{2\tilde t},\quad x=e^{\tilde t}\tilde x,\quad u=e^{-\tilde t}\tilde u.
\]
According to Proposition~\ref{PropositionOnTransOfCharsOfCLsForLPEs}, the transformation~$\mathcal T$ 
acts identically on values of characteristics: $\alpha=\tilde\alpha$. 
That is why in view of Theorem~\ref{TheoremOnCharsGivenPotSymsOfLHE}  
only two $\tilde G_0^{\rm ess}$-inequivalent characteristics $\tilde\alpha^1=1$ and $\tilde\alpha^2=e^{\tilde t}\tilde x$
of equation~\eqref{EqFPx} lead to nontrivial potential symmetries of this equation. 
Here $\tilde G_0^{\rm ess}$ is the essential part of the point symmetry group of equation~\eqref{EqFPx}. 
The corresponding potential symmetry algebras $\tilde{\mathfrak p}_1$ and $\tilde{\mathfrak p}_2$ can be obtained in two ways. 
The first way is the direct calculation of the maximal Lie invariance algebras of the associated potential systems 
$\tilde v^1_{\tilde x}=\tilde u$, 
$\tilde v^1_{\tilde t}=\tilde u_{\tilde x}+\tilde x\tilde u$ 
and 
$\tilde v^2_{\tilde x}=e^{\tilde t}\tilde x\tilde u$, 
$\tilde v^2_{\tilde t}=e^{\tilde t}\tilde x\tilde u_{\tilde x}+e^{\tilde t}(\tilde x^2-1)\tilde u$. 
The second way is to map the potential symmetry algebras $\mathfrak p_1$ and $\mathfrak p_2$ of the linear heat equation 
by the transformation inverse to~$\mathcal T$ and trivially prolonged to the potentials $\tilde v^1$ and $\tilde v^2$, respectively.
Finally, the algebras $\tilde{\mathfrak p}_1$ and $\tilde{\mathfrak p}_2$ have the form 
\begin{gather*}
\tilde{\mathfrak p}_1=\langle\ \p_t,\ e^{-t}\p_x,\ e^{-2t}\p_t-e^{-2t}x\p_x+e^{-2t}u\p_u,\ e^t\p_x-e^t(xu+v^1)\p_u-e^txv^1\p_{v^1},
\\ \phantom{\tilde{\mathfrak p}_1=\langle\ }
e^{2t}\p_t+e^{2t}x\p_x-e^{2t}(x^2u+2xv^1+2u)\p_u-e^{2t}(x^2+1)v^1\p_{v^1},\ u\p_u+v^1\p_{v^1}, 
\\ \phantom{\tilde{\mathfrak p}_1=\langle\ }
g_x\p_u+g\p_{v^1}\, \rangle,
\\[1ex]
\tilde{\mathfrak p}_2=\langle\ 
\p_t-u\p_u,\ 
e^{-2t}\p_t-e^{-2t}x\p_x+e^{-2t}u\p_u,\ 
\\ \phantom{\tilde{\mathfrak p}_2=\langle\ }
e^{2t}\p_t+e^{2t}x\p_x-e^{2t}(x^2u+2e^{-t}v^2+2u)\p_u-e^{2t}(x^2-1)v^2\p_{v^2},\ 
u\p_u+v^2\p_{v^2}, 
\\ \phantom{\tilde{\mathfrak p}_2=\langle\ }
e^tx^{-1}h_x\p_u+h\p_{v^2}\, \rangle,
\end{gather*}
where the functions $g=g(t,x)$ and $h=h(t,x)$ run through the solution set of the associated potential equations 
$\tilde v^1_{\tilde t}-\tilde v^1_{\tilde x\tilde x}-\tilde x\tilde v^1_{\tilde x}=0$ 
and 
$\tilde v^2_{\tilde t}-\tilde v^2_{\tilde x\tilde x}+(2\tilde x^{-1}-\tilde x)\tilde v^2_{\tilde x}=0$, 
respectively. 

As a result, we obtain the following statement.

\begin{corollary}\label{CorollaryOnPotSymsOfFPxE}
All $\tilde G_0^{\rm ess}$-equivalent simplest pure potential symmetries of the Fokker--Planck equation~\eqref{EqFPx}
are exhausted by the operators from~$\tilde{\mathfrak p}_1$ and~$\tilde{\mathfrak p}_2$ 
satisfying the condition $\tilde\eta_{\tilde v}\ne0$, 
where $\tilde\eta$ is the coefficient of~$\p_{\tilde u}$ and $\tilde v$ is the corresponding potential ($\tilde v^1$ or $\tilde v^2$).
\end{corollary}

Corollary~\ref{CorollaryOnPotSymsOfFPxE} essentially generalizes results 
of~\cite{Ivanova&Popovych2007CommentOnMei,Pucci&Saccomandi1993,Saccomandi1997} 
on simplest potential symmetries of equation~\eqref{EqFPx}, associated with the characteristic~$1$.

\section{Preliminary analysis of general potential systems}\label{SectionOnGenPotSysForLPEs}

The investigation of general potential symmetries can be carried out in a way similar to that of simplest potential symmetries 
but all calculations are much more complicated. 
The idea is to find an analogue of the potential equations~\eqref{EqModifiedPotLPE} in the case of general potential systems 
and to construct potential symmetries of the initial equations via prolongation of Lie symmetries of the potential equations. 
It is necessary also to prove statements on the behavior of symmetries of the potential frame under the action of equivalence transformations. 
This allows us to replace the study of class~\eqref{EqGenLPE} by the study of the simpler class~\eqref{EqReducedLPE}. 
Since the problem is quite complicated, in this section we only find a convenient form of potential systems 
and construct potential equations each of which is equivalent, in a certain sense, to a whole potential systems. 
Along the way, we make intensive use of multiple Darboux transformations. 
These results create a basis for the symmetry analysis of the potential frame in the next section. 

Let us fix an equation $u_t=Au_{xx}+Bu_x+Cu$ from class~\eqref{EqGenLPE}, an arbitrary $p\in\mathbb{N}$ and 
$p$~linearly independent solutions $\alpha^1$, \ldots, $\alpha^p$ of the adjoint equation~\eqref{EqAdjLPE}. 
Then the conservation laws corresponding to these characteristics are linearly independent.
For any~$s$ we introduce the potential~$v^s$ using the conserved vector of the canonical form~\eqref{eqCVofLPEs}, 
associated with $\alpha^s$. 
Our choice of the conserved vectors will be justified below by Corollary~\ref{CorollaryOnChiceOfConservedVertorsForConstructionOfPotSymsOfLPE}.
As a result, we obtain the potential system 
\begin{equation}\label{EqGenPotSysOfLPEs}
v^s_x=\alpha^s u,\quad v^s_t=\alpha^sAu_x-((\alpha^sA)_x-\alpha^sB)u 
\end{equation}
corresponding to the characteristic tuple $\bar\alpha=(\alpha^1,\ldots,\alpha^p)$. 
Let us recall that the indices~$s$, $\sigma$ and~$\varsigma$ run at most from~1 to~$p$. 
Additional constraints on indices are indicated explicitly when needed. 
The summation convention over repeated indices is used unless otherwise stated 
or it is obvious from the context that indices are fixed.

If the initial equation runs through class~\eqref{EqGenLPE} and 
$\alpha^1$, \ldots, $\alpha^p$ runs through all linearly independent solutions of the adjoint equation 
then the associated potential systems form \emph{the potential frame of order}~$p$ (and the first level) 
over the class~\eqref{EqGenLPE}. In general, by the (potential) order of an object we will mean the number 
of independent first-level potentials appearing in this object. 
Below we extend the potential frame with other objects. 

\looseness=1
System~\eqref{EqGenPotSysOfLPEs} is homogeneous with respect to the index~$s$. 
It is formed by $p$ similar blocks. Each of them consists of a pair of equations in a potential and 
the initial unknown function~$u$ and possesses the structure of a simplest potential system. 
All the potentials $v^1$, \ldots, $v^p$ are on an equal footing. 
At first glance these features seem to be advantages of this representation of potential systems 
but after careful consideration a number of drawbacks become apparent. 
The total number of unknown functions in system~\eqref{EqGenPotSysOfLPEs} equals $p+1$ 
and the total number of equations is $2p$. 
At the same time, the system is not ``too'' overdetermined since it has no nontrivial differential consequences. 
For any fixed~$s$ the corresponding pair of equations 
implies an equation only with respect to~$v^s$ and a differential consequence equivalent to the initial equation. 
In fact, system~\eqref{EqGenPotSysOfLPEs} contains only $p+1$ independent equations 
but the number of equations cannot be reduced to the minimal one in a symmetric way. 
It is not clear what a potential equation corresponding to the whole system~\eqref{EqGenPotSysOfLPEs} should be. 
Another argument in favour of modifying system~\eqref{EqGenPotSysOfLPEs} comes from group analysis. 
Consider a Lie symmetry operator $Q=\tau\p_t+\xi\p_x+\eta\p_u+\theta^s\p_{v^s}$ of system~\eqref{EqGenPotSysOfLPEs}.
The coefficients of~$Q$ are functions of $t$, $x$, $u$ and $v^\sigma$.
The infinitesimal invariance criterion applied to system~\eqref{EqGenPotSysOfLPEs} implies, in particular,
the following determining equations on the coefficients of~$Q$:
$\tau_u=\tau_x=\tau_{v^s}=0$, $\xi_u=\xi_{v^s}=0$, $\theta^s_u=0$.
In contrast to the case of single characteristics, 
the deduction of these simplest determining equations is much more involved. 
Certain tricks involving linear independence of $\alpha^1$, \ldots, $\alpha^p$ have to be used.  
Finding other simple determining equations which are typical for linear systems (e.g., 
$\eta_{uu}=0$, $\eta_{uv^s}=0$, $\theta^s_{v^\sigma v^{\varsigma}}=0$) demands still more calculations and tricks. 
A complete analysis of the whole system of determining equations seems impossible. 

Below by an iteration procedure we obtain another potential system associated with 
the characteristic tuple $\bar\alpha=(\alpha^1,\ldots,\alpha^p)$, 
which is equivalent to system~\eqref{EqGenPotSysOfLPEs} but is appropriate 
for the investigation of potential symmetries of the initial equation. 
Each step of the iteration procedure is similar to the consideration in the beginning of 
Section~\ref{SectionOnSimplestPotSymsOfLPEs}.
For convenience an equation for an unknown function~$\varphi$ will be denoted by~$\lceil \varphi\rfloor$.

\emph{Step 1}. 
Let us re-denote $u\to w^0$, $B\to B^0$, $C\to C^0$,
$v^1\to f^1$, $\alpha\to \beta^0$ and $\alpha^s\to \beta^{0,s}$. 
(We will use the old and new notations simultaneously.)
Consider the (first-level) potential system 
\[
f^1_x=\beta^{0,1}w^0,\quad f^1_t=\beta^{0,1}Aw^0_x-((\beta^{0,1}A)_x-\beta^{0,1}B^0)w^0
\] 
associated with the single characteristic~$\alpha^1=\beta^{0,1}$. 
The tuple~$(w^0,f^1)$ is a solution of this system iff the modified potential $w^1=f^1/\beta^{0,1}$ satisfies
the equation $w^1_t=Aw^1_{xx}+B^1w^1_x+C^1w^1$, where $B^1=B^0-A_x$ and 
\[
C^1=C-B_x+A_{xx}+A_x\frac{\beta^{0,1}_x}{\beta^{0,1}}+2A\biggl(\frac{\beta^{0,1}_x}{\beta^{0,1}}\biggr)_x
=C-B_x+A_{xx}+A_x\frac{W^1_x}{W^1}+2A\left(\frac{W^1_x}{W^1}\right)_x.
\]
Here and occasionally below the notation $W^s$ for the Wronskian $W(\alpha^1,\dots,\alpha^s)$ is used. 
The function $w^{1,1}=1/\beta^{0,1}$ is a solution of~$\lceil w^1\rfloor$.
The Darboux transformation~${\rm DT}[w^{1,1}]$ maps $\lceil w^1\rfloor$ to $\lceil w^0\rfloor=\lceil u\rfloor$.
In view of Lemma~\ref{LemmaOnDualDarbouxTrans} the dual Darboux transformation~${\rm DT}[\beta^{0,1}]$
maps $\lceil \alpha\rfloor=\lceil \beta^{0,1}\rfloor$ to the equation 
$\beta^1_t=(A\beta^1)_{xx}+(B^1\beta^1)_x+C^1\beta^1$ 
adjoint to~$\lceil w^1\rfloor$. Therefore, the functions 
\[
\beta^{1,s}={\rm DT}[\beta^{0,1}](\beta^{0,s})=\alpha^s_x-\frac{\alpha^1_x}{\alpha^1}\alpha^s=
\frac{W(\alpha^1,\alpha^s)}{W(\alpha^1)}
\]
satisfy the equation~$\lceil \beta^1\rfloor$ and $\beta^{1,s}\in\Ch_{\rm f}(\lceil w^1\rfloor)$, 
i.e. they are characteristics of conservation laws of~$\lceil w^1\rfloor$. 
Note that $\beta^{1,1}=0$.

\emph{Step 2}.
Using the conservation law of~$\lceil w^1\rfloor$, having the characteristic~$\beta^{1,2}$, 
we introduce the potential~$f^2$ and obtain the potential system 
\[
f^2_x=\beta^{1,2}w^1,\quad f^2_t=\beta^{1,2}Aw^1_x-((\beta^{1,2}A)_x-\beta^{1,2}B^1)w^1.
\] 
(Its union with the constructed first-level potential system results in a second-level potential system.)
The tuple~$(w^1,f^2)$ satisfies this system iff the modified potential $w^2=f^2/\beta^{1,2}$ is a solution of 
the equation $w^2_t=Aw^2_{xx}+B^2w^2_x+C^2w^2$, where $B^2=B^1-A_x=B-2A_x$ and 
\[
C^2=C^1-B^1_x+A_{xx}+A_x\frac{\beta^{1,2}_x}{\beta^{1,2}}+2A\biggl(\frac{\beta^{1,2}_x}{\beta^{1,2}}\biggr)_x
=C-2B_x+3A_{xx}+A_x\frac{W^2_x}{W^2}+2A\left(\frac{W^2_x}{W^2}\right)_x.
\] 
since $\beta^{0,1}\beta^{1,2}=W^2$.
The function $w^{2,2}=1/\beta^{1,2}$ is a solution of~$\lceil w^2\rfloor$.
${\rm DT}[w^{2,2}]$ maps $\lceil w^2\rfloor$ in $\lceil w^1\rfloor$.
Then the dual Darboux transformation~${\rm DT}[\beta^{1,2}]$
maps $\lceil \beta^1\rfloor$ to the equation $\beta^2_t=(A\beta^2)_{xx}+(B^2\beta^2)_x+C^2\beta^2$ 
adjoint to~$\lceil w^2\rfloor$. 
Therefore, the functions $\beta^{2,s}={\rm DT}[\beta^{1,2}](\beta^{1,s})$
satisfy the equation~$\lceil \beta^2\rfloor$ and $\beta^{2,s}\in\Ch_{\rm f}(\lceil w^2\rfloor)$, 
i.e. they are characteristics of conservation laws of~$\lceil w^2\rfloor$.
Since $\beta^{1,s}={\rm DT}[\alpha^1](\alpha^s)$ then in view of the Crum theorem~\cite{Crum1955,Matveev&Salle1991}
\[
\beta^{2,s}={\rm DT}[\beta^{1,2}](\beta^{1,s})=\beta^{1,s}_x-\frac{\beta^{1,2}_x}{\beta^{1,2}}\beta^{1,s}=
\frac{W(\alpha^1,\alpha^2,\alpha^s)}{W(\alpha^1,\alpha^2)}.
\]
This formula implies, in particular, that $\beta^{2,1}=\beta^{2,2}=0$.

The next iteration is obvious. 

\emph{Step s}.
Using the conservation law with the characteristic~$\beta^{s-1,s}$
of the modified potential equation~$\lceil w^{s-1}\rfloor$ from the previous step, 
we introduce the potential~$f^s$ and obtain the potential system 
\begin{equation}\label{EqIteratedGenPotSysOfLPEs}
f^s_x=\beta^{s-1,s}w^{s-1}, \quad f^s_t=\beta^{s-1,s}Aw^{s-1}_x-((\beta^{s-1,s}A)_x-\beta^{s-1,s}B^{s-1})w^{s-1}.
\end{equation}
(Its union with the $(s-1)$-level potential system constructed during the previous iterations
results in an $s$-level potential system of~\ref{EqGenLPE}.)
The tuple~$(w^{s-1},f^s)$ satisfies system~\eqref{EqIteratedGenPotSysOfLPEs} 
iff the modified potential $w^s=f^s/\beta^{s-1,s}$ is a solution of 
the equation $w^s_t=Aw^s_{xx}+B^sw^s_x+C^sw^s$, where $B^s=B^{s-1}-A_x=B-sA_x$ and 
\begin{gather*}
C^s=C^{s-1}-B^{s-1}_x+A_{xx}+A_x\frac{\beta^{s-1,s}_x}{\beta^{s-1,s}}
+2A\biggl(\frac{\beta^{s-1,s}_x}{\beta^{s-1,s}}\biggr)_x\\
\phantom{C^s}=C-sB_x+\frac{s(s-1)}2A_{xx}+A_x\frac{W^s_x}{W^s}+2A\left(\frac{W^s_x}{W^s}\right)_x.
\end{gather*}
since $\beta^{0,1}\dots\beta^{s-1,s}=W^s$. 
The function $w^{s,s}=1/\beta^{s-1,s}$ is a solution of~$\lceil w^s\rfloor$.
${\rm DT}[w^{s,s}]$ maps $\lceil w^s\rfloor$ in $\lceil w^{s-1}\rfloor$.
Then the dual Darboux transformation~${\rm DT}[\beta^{s-1,s}]$
maps $\lceil \beta^{s-1}\rfloor$ to the equation $\beta^s_t=(A\beta^s)_{xx}+(B^s\beta^s)_x+C^s\beta^s$ 
adjoint to~$\lceil w^s\rfloor$. 
Therefore, the functions $\beta^{s,\sigma}={\rm DT}[\beta^{s-1,s}](\beta^{s-1,\sigma})$
satisfy the equation~$\lceil \beta^s\rfloor$ and $\beta^{s,\sigma}\in\Ch_{\rm f}(\lceil w^s\rfloor)$, 
i.e.\ they are characteristics of conservation laws of~$\lceil w^s\rfloor$.
Since $\beta^{s-1,\sigma}$ are constructed by iteration of the Darboux transformation 
from the characteristics $\alpha^1$, \ldots, $\alpha^p$, in view of the Crum theorem we obtain
\[
\beta^{s,\sigma}={\rm DT}[\beta^{s-1,s}](\beta^{s-1,\sigma})
=\beta^{s-1,\sigma}_x-\frac{\beta^{s-1,s}_x}{\beta^{s-1,s}}\beta^{s-1,\sigma}
=\frac{W(\alpha^1,\dots,\alpha^s,\alpha^\sigma)}{W(\alpha^1,\dots,\alpha^s)}.
\]
This formula implies, in particular, that $\beta^{s,\sigma}=0$ if $\sigma\leqslant s$ and $\beta^{s,\sigma}\ne0$ if $\sigma>s$.

The iteration procedure is stopped on step~$p$ after the construction of the equation $\lceil w^p\rfloor$ 
since there are no nonzero functions $\beta^{p,\sigma}$.

For any~$s<p$ the second equation in system~\eqref{EqIteratedGenPotSysOfLPEs} is a differential consequence 
of system~\eqref{EqIteratedGenPotSysOfLPEs} with $s+1$ replacing $s$.
Therefore, the `minimal' combined potential system consists of the first equations of the potential systems from 
all steps and the second equation of the potential system constructed on the last, $p$-th, step. 

Let $f^0=w^0=u$, $W^0=W^{-1}=1$ and $\beta^{-1,s}=1$ by definition.
Excluding~$w^s$ due to the formula $w^s=f^s/\beta^{s-1,s}$, we obtain the combined potential system 
in terms of only $f^s$:
\begin{equation}\label{EqMinIteratedGenPotSysOfLPEs}
f^s_x=H^sf^{s-1},
\quad 
f^p_t=H^pAf^{p-1}_x-G^pf^{p-1},
\end{equation}
where
\[
H^s=\frac{\beta^{s-1,s}}{\beta^{s-2,s-1}}=\frac{W^sW^{s-2}}{(W^{s-1})^2}, \quad
G^s=(H^sA)_x-H^sB^{s-1}+2H^sA\frac{\beta^{s-2,s-1}_x}{\beta^{s-2,s-1}}.
\]
It can be proved that $H^s_t+G^s_x=0$ for any~$s$.

System~\eqref{EqMinIteratedGenPotSysOfLPEs} is the $p$-level form of the potential system of equation~\eqref{EqGenLPE},
associated with the characteristic tuple $\bar\alpha=(\alpha^1,\ldots,\alpha^p)$.
The equations $f^s_t=H^sAf^{s-1}_x-G^sf^{s-1}$, $s=1,\dots,p-1$, are differential consequences 
of~\eqref{EqMinIteratedGenPotSysOfLPEs}.
In the case $p>1$ the derivatives with respect to~$x$ can be excluded from these equations with $s>1$ as well 
from the last equation of~\eqref{EqMinIteratedGenPotSysOfLPEs}. The resulting equations are 
$
f^s_t=H^sH^{s-1}Af^{s-2}-G^sf^{s-1},\ s=2,\dots,p.
$ 

In view of Corollary~\ref{CorollaryOnPotSystemsOfLPEs}, system~\eqref{EqMinIteratedGenPotSysOfLPEs} should be 
equivalent with respect to point transformations to a potential system of the first level. 
Below we explicitly construct a point transformation from system~\eqref{EqGenPotSysOfLPEs} to system~\eqref{EqMinIteratedGenPotSysOfLPEs}.

Let us define the functions $g^{s,\sigma}$, $\sigma\geqslant s$, by the recursive formula 
\[
g^{1,\sigma}=v^\sigma, \quad g^{s+1,\sigma}=\frac{\beta^{s-1,\sigma}}{\beta^{s-1,s}}g^{s,s}-g^{s,\sigma}.
\]
For convenience it can be assumed that $g^{s,\sigma}=0$ if $\sigma<s$. 

\begin{lemma}\label{LemmaOnRepresentationOfGenPotsViaSimplestPotsForLPEs1}
For any fixed $s$ and $\sigma$ the function $g^{s\sigma}$ is the potential of the equation $\lceil w^{s-1}\rfloor$, 
associated with the characteristic~$\beta^{s-1,\sigma}$. In particular, $g^{s,s}=f^s$ up to a trivial constant summand.
\end{lemma}

\begin{proof}
The lemma is proved by induction with respect to~$s$. 
The statement of the lemma for~$s=1$ is obvious in view of the definitions of~$g^{1,\sigma}$ and~$f^1$. 
Suppose that the statement is true for a fixed~$s$ and $\sigma=s,\dots,p$. Let us prove it for $s+1$ and $\sigma=s+1,\dots,p$.
By assumption, the functions $g^{s,\sigma}$, $\sigma=s,\dots,p$, satisfy the conditions
\[
g^{s,\sigma}_x=\beta^{s-1,\sigma}w^{s-1}, \quad 
g^{s,\sigma}_t=\beta^{s-1,\sigma}Aw^{s-1}_x-((\beta^{s-1,\sigma}A)_x-\beta^{s-1,\sigma}B^{s-1})w^{s-1}.
\]
and $g^{s,s}=f^s$. Then the first derivatives of $g^{s+1,\sigma}$, $\sigma=s+1,\dots,p$, are 
\begin{gather*}
g^{s+1,\sigma}_x=\left(\frac{\beta^{s-1,\sigma}}{\beta^{s-1,s}}\right)_xf^s=\beta^{s,\sigma}w^s, 
\\ 
g^{s+1,\sigma}_t=\left(\frac{\beta^{s-1,\sigma}}{\beta^{s-1,s}}\right)_tf^s+
\left(\beta^{s-1,\sigma}_x-\frac{\beta^{s-1,\sigma}}{\beta^{s-1,s}}\beta^{s-1,s}_x\right)Aw^{s-1}
=\left(\frac{\beta^{s-1,\sigma}}{\beta^{s-1,s}}\right)_tf^s+\frac{\beta^{s,\sigma}}{\beta^{s-1,s}}Af^s_x
\\
\phantom{g^{s+1,\sigma}_x}
=\left(\beta^{s,\sigma}(B^{s-1}\!-\!2A_x)-\beta^{s,\sigma}_xA-\left(\frac{\beta^{s-1,\sigma}}{\beta^{s-1,s}}\right)_x\beta^{s-1,s}_xA
\right)w^s
+\beta^{s,\sigma}A\biggl(w^s_x+\frac{\beta^{s-1,s}_x}{\beta^{s-1,s}}w^s\biggr)
\\[1ex]
\phantom{g^{s+1,\sigma}_x}
=\beta^{s,\sigma}Aw^s_x-((\beta^{s,\sigma}A)_x-\beta^{s,\sigma}B^s)w^s.
\end{gather*}
Therefore, the function $g^{s+1,\sigma}$ is a potential of the equation $\lceil w^s\rfloor$, 
associated with the characteristic~$\beta^{s,\sigma}$. 
Since $g^{s+1,s+1}_x=f^{s+1}_x$ and $g^{s+1,s+1}_t=f^{s+1}_t$ then $g^{s+1,s+1}=f^{s+1}$ up to a trivial constant summand 
which can be neglected.
\end{proof}

For our further considerations we need a simple but useful property of matrix minors.
Let $\mathcal M_{i_1\dots i_q}^{j_1\dots j_q}$ denote the submatrix of the square matrix~$\mathcal M\in{\rm M}_{n,n}$, 
obtained by deletion of the rows with (different) numbers $i_1$, \dots, $i_q$ 
and the columns with (different) numbers $j_1$, \dots, $j_q$, $q\leqslant n$. 
$M_{i_1\dots i_q}^{j_1\dots j_q}=\det\mathcal M_{i_1\dots i_q}^{j_1\dots j_q}$ is the corresponding minor. 
In particular, $M_{i_1\dots i_n}^{j_1\dots j_n}=1$ by definition as the determinant of the empty matrix, 
$\det\mathcal M=M$ ($q=0$, i.e., there are no deleted rows and columns). 
Here and in the next lemma all indices run from 1 to $n$.

\begin{lemma}\label{LemmaOnMatrixEqulityForCrumTheorem}
For any~$n\geqslant2$, $\mathcal M\in{\rm M}_{n,n}$, $i\ne j$, $k\ne l$:
\[
M^i_kM^j_l-M^i_lM^j_k=\sign(j-i)\sign(l-l)\,M^{ij}_{kl}M.
\]
\end{lemma}

\begin{proof}
The lemma is also proved by induction over~$n$. 
For $n=2$ the statement reduces to the formula for $\det\mathcal M$.
Suppose that the statement is true for a fixed~$n$. Let us prove it for $n+1$. 
Consider $\mathcal M\in{\rm M}_{n+1,n+1}$. 
After permuting rows and columns, without loss of generality we can assume that $i=k=n$, $j=l=n+1$. 
Thus, it should be proved that $K:=M^n_nM^{n+1}_{n+1}-M^n_{n+1}M^{n+1}_n-M^{n,n+1}_{n,n+1}M=0$.
We expand $M^n_n$, $M^n_{n+1}$ and $M$ with respect to the elements of the $(n+1)$-st column of~$\mathcal M$:
\[
K=\sum_{i=1}^{n-1}(-1)^{i+n}\mu^{n+1}_i
\bigl(M^{n,n+1}_{in}M^{n+1}_{n+1}-M^{n,n+1}_{i,n+1}M^{n+1}_n+M^{n,n+1}_{n,n+1}M^{n+1}_i\bigr).
\]
(The terms with $i=n$ and $i=n+1$ cancel.) 
Then for each $i\in\{1,\dots,n-1\}$ we expand $M^{n+1}_{n+1}$, $M^{n+1}_n$ and $M^{n,n+1}_{n,n+1}$ 
in the corresponding term of the sum with respect to the elements of the $i$-th row:
\[
K=\sum_{i=1}^{n-1}\sum_{j=1}^{n-1}(-1)^{j+n}\mu^{n+1}_i\mu^j_i
\bigl(M^{n,n+1}_{in}M^{j,n+1}_{i,n+1}-M^{n,n+1}_{i,n+1}M^{j,n+1}_{in}+M^{jn,n+1}_{in,n+1}M^{n+1}_i\bigr).
\]
(The terms with $j=n$ cancel.) 
The coefficients in the latter sum vanish in view of the induction hypothesis applied to the 
matrices~$\mathcal M^{n+1}_i$. 
\end{proof}

\begin{lemma}\label{LemmaOnRepresentationOfGenPotsViaSimplestPotsForLPEs2}
$\displaystyle g^{s,\sigma}=(-1)^{s-1}
\frac{W(\alpha^1,\dots,\alpha^{s-1},\alpha^\sigma)_{\bar v\rightsquigarrow\bar\alpha_{s-1}}}{W(\alpha^1,\dots,\alpha^{s-1})}$. 
\end{lemma}

Here the notations introduced in Section~\ref{SectionOnPotCLsLPEs} are used.
Namely, the subscript $s-1$ denotes the $(s-1)$-st order derivative with respect to $x$, 
and a function is considered as its zero-order derivative.
The notation ``$\bar v\rightsquigarrow\bar\alpha_{s-1}$'' means that the derivatives~$\alpha^\varsigma_{s-1}$ 
are replaced by the function~$v^\varsigma$ for the range of subscripts in the corresponding Wronskian, 
i.e., $\varsigma=1,\dots,s-1,\sigma$ in our case. 
Note that Lemma~\ref{LemmaOnRepresentationOfGenPotsViaSimplestPotsForLPEs2} gives 
significant values of~$g^{s,\sigma}$ only for $\sigma\geqslant s$ and $g^{s,\sigma}=0$ if $\sigma<s$, 
in agreement with the definition of~$g^{s,\sigma}$. 

\begin{proof}
We again use induction on~$s$. 
The statement of the lemma for~$s=1$ is obvious since $g^{1,\sigma}=v^\sigma$, 
$W(\alpha^\sigma)_{\bar v\rightsquigarrow\bar\alpha}=v^\sigma$ and the Wronskian of the empty tuple equals 1 by definition. 
Suppose that the statement is true for a fixed~$s$ and $\sigma=s,\dots,p$. Let us prove it for $s+1$ and $\sigma=s+1,\dots,p$.
In view of the assumption and the recursive formula for $g^{s+1,\sigma}$, the statement for $s+1$ is equivalent to the formula
\begin{gather*}
W(\alpha^1,\dots,\alpha^s)\,W(\alpha^1,\dots,\alpha^{s-1},\alpha^\sigma)_{\bar v\rightsquigarrow\bar\alpha_{s-1}}
-W(\alpha^1,\dots,\alpha^{s-1},\alpha^\sigma)\,W(\alpha^1,\dots,\alpha^s)_{\bar v\rightsquigarrow\bar\alpha_{s-1}}
\\
=W(\alpha^1,\dots,\alpha^{s-1})\,
W(\alpha^1,\dots,\alpha^s,\alpha^\sigma)_{\bar v\rightsquigarrow\bar\alpha_s}
\end{gather*}
The latter is a pure matrix equality and does not depend on the specific structure of Wronski matrices. 
It follows from Lemma~\ref{LemmaOnMatrixEqulityForCrumTheorem}.
\end{proof}

\begin{corollary}
Systems~\eqref{EqGenPotSysOfLPEs} and~\eqref{EqMinIteratedGenPotSysOfLPEs} are equivalent with respect to 
a point transformation of only the potential dependent variables, being linear in these variables. 
In other words, the $p$-level potential frame over class~\eqref{EqGenLPE} is equivalent to 
the first-level potential frame of order~$p$ over the same class. 
They can be simultaneously considered in the framework of the general $p$-order potential frame. 
\end{corollary}

\begin{corollary}
On any level $s$ the functions $g^{s\sigma}$, $\sigma\geqslant s$, can be expressed via 
the functions $g^{\varsigma\smash{\sigma'}}$, $\sigma'\geqslant\varsigma$, of any lower level~$\varsigma$:
\[
g^{s,\sigma}=(-1)^{s-\varsigma}
\frac{W(\beta^{\varsigma-1,\varsigma},\dots,\beta^{\varsigma-1,s-1},\beta^{\varsigma-1,\sigma})
_{\bar g^{\varsigma}\rightsquigarrow\bar\beta^{\varsigma-1}_{s-\varsigma}}}
{W(\beta^{\varsigma-1,\varsigma},\dots,\beta^{\varsigma-1,s-1})}.
\]
\end{corollary}
\begin{proof}
We apply Lemma~\ref{LemmaOnRepresentationOfGenPotsViaSimplestPotsForLPEs2} assuming that the iteration procedure is started from level~$\varsigma$. 
\end{proof}

Since $w^s=f^s/\beta^{s-1,s}$ then in view of the formula for~$\beta^{s-1,s}$ we have one more corollary 
from Lemma~\ref{LemmaOnRepresentationOfGenPotsViaSimplestPotsForLPEs2}. 

\begin{corollary}\label{CorollaryOnIndependenceOfPotsOnOrderOfCharsForLPEs}
For any~$s$ the final result of the $s$-th iteration (i.e., the expression of the potential $w^s$ 
via the potentials $v^1$, \dots, $v^s$ and the form of the equation~$\lceil w^s\rfloor$)
is invariant with respect to nondegenerate linear transformations 
$\tilde\alpha^\sigma=\sum_{\varsigma=1}^s\alpha^\varsigma c_{\varsigma\sigma}$, $\sigma=1,\dots,s$,
of the characteristics from the tuple~$(\alpha^1,\dots,\alpha^s)$. 
Here $c_{\varsigma\sigma}$, $\sigma,\varsigma=1,\dots,s$, are constants such that $\det(c_{\varsigma\sigma})\not=0$.
In~particular, 
\[
w^s=(-1)^{s-1}
\frac{W(\alpha^1,\dots,\alpha^s)_{\bar v\rightsquigarrow\bar\alpha_{s-1}}}{W(\alpha^1,\dots,\alpha^s)}.
\]
\end{corollary}

\begin{note}\label{NoteOnSimultaneousLinearCombiningOfCharsAndPots}
A nondegenerate linear transformation 
$\tilde\alpha^\sigma=\sum_{\varsigma=1}^s\alpha^\varsigma c_{\varsigma\sigma}$, $\sigma=1,\dots,s$, 
in a characteristic tuple~$(\alpha^1,\dots,\alpha^s)$ necessarily implies 
the simultaneous linear transformation $\tilde v^\sigma=\sum_{\varsigma=1}^sv^\varsigma c_{\varsigma\sigma}$ 
in the associated potential tuple~$(v^1,\dots,v^s)$ with the same coefficients. 
\end{note}

\begin{note}\label{NoteOnCorrespodenceOfCharSubspaceForLPEs}
We can say that the modified $s$-level potential $w^s$ and the equation~$\lceil w^s\rfloor$ correspond to 
the $s$-dimensional characteristics subspace $\langle\alpha^1,\dots,\alpha^s\rangle$ 
instead of the characteristic tuple~$(\alpha^1,\dots,\alpha^s)$
since the choice of a subspace basis is inessential in view of Corollary~\ref{CorollaryOnIndependenceOfPotsOnOrderOfCharsForLPEs}.
\end{note}

\begin{note}
The Crum theorem can be proved in a way similar to Lemma~\ref{LemmaOnRepresentationOfGenPotsViaSimplestPotsForLPEs2} 
by using Lemma~\ref{LemmaOnMatrixEqulityForCrumTheorem}.
\end{note}

Below the notation ``$(\alpha^1,\dots,\lefteqn{\smash{\diagdown}}\alpha^\varsigma,\dots,\alpha^s)$'' means that 
$\alpha^\varsigma$ is absent in the corresponding tuple of~$\alpha$'s. 

\begin{lemma}\label{LemmaOnSolutionsOfPotEqsforLPEs}
For any fixed~$s$ the functions 
\[
w^{s,\varsigma}=(-1)^{\varsigma-1}
\frac{W(\alpha^1,\dots,\lefteqn{\smash{\diagdown}}\alpha^\varsigma,\dots,\alpha^s)}{W(\alpha^1,\dots,\alpha^s)}, 
\quad \varsigma=1,\dots,s,
\]
are linearly independent solutions of~$\,\lceil w^s\rfloor$. 
Moreover, $W(w^{1,s},\dots,w^{s,s})\,W(\alpha^1,\dots,\alpha^s)=1$. 
\end{lemma}

\begin{proof}
Since for any $\sigma$ $v^\sigma=\const$ is a solution of the corresponding potential equation, 
in view of Corollary~\ref{CorollaryOnIndependenceOfPotsOnOrderOfCharsForLPEs}
the tuples $(v^1,\dots,v^s)=(\delta_{1\sigma},\dots,\delta_{s\sigma})$, $\sigma=1,\dots,s$, 
where $\delta_{\varsigma\sigma}$ is the Kronecker delta, give $s$ solutions~$w^{s,\varsigma}$ of~$\,\lceil w^s\rfloor$.
By the formula on determinant expansion, 
$w^{s,\sigma}\alpha^\sigma_{\varsigma-1}=0$ if $\varsigma=1,\dots,s-1$ and $w^{s,\sigma}\alpha^\sigma_{\varsigma-1}=(-1)^{s-1}$ if $\varsigma=s$.
Combining differential consequences of these formulas, we derive that 
\[
w^{s,\sigma}_{\smash{\varsigma'\!}-1}\alpha^\sigma_{\varsigma-1}=\left\{
\begin{array}{cl}0,\ &\varsigma=1,\dots,s-\varsigma',\\[.5ex]
(-1)^{s-\smash{\varsigma'\!}},\ & \varsigma=s-\varsigma'+1.
\end{array}\right.  
\]
Therefore, the product of the Wronski matrix of $(\alpha^1,\dots,\alpha^s)$ and 
the transposed Wronski matrix of $(w^{s,1},\dots,w^{s,s})$ is a matrix with zeros above 
the anti-diagonal (the diagonal going from the lower left corner to the upper right corner). 
Its anti-diagonal entries equal $(-1)^{s-\varsigma}$, $\varsigma=1,\dots,s$, starting from the lower left corner.
The determinant of such a matrix equals $1$. Therefore, $W(w^{1,s},\dots,w^{s,s})\,W(\alpha^1,\dots,\alpha^s)=1$.
\end{proof}

\begin{corollary}\label{CorollaryOnPresentationOfPotsViaFixedSolutionsOfPotEqForLPEs}
The $s$-level potential~$w^s$ is a linear combination of the potentials $v^1$, \dots, $v^s$ 
with functional coefficients which are fixed solutions of the equation~$\,\lceil w^s\rfloor$:
$w^s=\sum_{\sigma=1}^s w^{s\sigma}v^\sigma.$ 
\end{corollary}

\begin{lemma}\label{LemmaOnDTOfWsForLPEs2}
${\rm DT}[w^{s,s}](w^{s,\varsigma})=w^{s-1,\varsigma}$, $\varsigma=1,\dots,s-1$.
\end{lemma}
\begin{proof}
We substitute the expressions for $w^{s,s}$ and $w^{s,\varsigma}$ given by Lemma~\ref{LemmaOnSolutionsOfPotEqsforLPEs} 
into~${\rm DT}[w^{s,s}](w^{s,\varsigma})$, simplify the obtain result and then apply Lemma~\ref{LemmaOnMatrixEqulityForCrumTheorem}.
\end{proof}

The multiple Darboux transformation constructed with a tuple of linearly independent functions~$(\psi^1,\dots,\psi^s)$ 
is denoted by ${\rm DT}[\psi^1,\dots,\psi^s]$, i.e., 
\[
{\rm DT}[\psi^1,\dots,\psi^s](w)=\frac{W(\psi^1,\dots,\psi^s,w)}{W(\psi^1,\dots,\psi^s)}.
\]

As a result of the above considerations we obtain the following theorem:

\begin{theorem}\label{TheoremOnDualMultipleDarbouxTrans}
Let $\,\alpha^1$, \dots, $\alpha^p$ be linearly independent solutions of the equation~$\lceil\alpha\rfloor=\lceil u\rfloor^*$ 
adjoint to the equation~$\lceil u\rfloor$ and ${\rm DT}[\alpha^1,\dots,\alpha^p]\,\lceil\alpha\rfloor=\lceil\beta^p\rfloor$.
Then the adjoint equation $\lceil\beta^p\rfloor^*=\lceil w^p\rfloor$ to $\lceil\beta^p\rfloor$ is 
the $p$-level potential equation of~$\lceil u\rfloor$, constructed from the characteristic tuple~$(\alpha^1,\dots,\alpha^p)$, 
the functions 
\[
w^{p,\varsigma}=(-1)^{\varsigma-1}
\frac{W(\alpha^1,\dots,\lefteqn{\smash{\diagdown}}\alpha^\varsigma,\dots,\alpha^p)}{W(\alpha^1,\dots,\alpha^p)}, 
\quad \varsigma=1,\dots,p,
\]
are its linearly independent solutions and ${\rm DT}[w^{p,1},\dots,w^{p,p}]\,\lceil w^p\rfloor=\lceil u\rfloor$, i.e.,
\begin{gather*}\begin{array}{rcl}
u_t=Au_{xx}+Bu_x+Cu & \xleftarrow{{\rm DT}[w^{p,1},\dots,w^{p,p}]} & w^p_t=Aw^p_{xx}+B^pw^p_x+C^pw^p
\\[-.4ex]
&\Updownarrow&
\\[.4ex]
\alpha_t+(A\alpha)_{xx}-(B\alpha)_x+C\alpha=0 & \xrightarrow{{\rm DT}[\alpha^1,\dots,\alpha^p]} & 
\beta^p_t+(A\beta^p)_{xx}-(B^p\beta^p)_x+C^p\beta^p=0.
\end{array}
\end{gather*}
\end{theorem}

In view of reflexiveness in the duality of linear equations, 
Theorem~\ref{TheoremOnDualMultipleDarbouxTrans} may also be reformulated, analogously to Lemma~\ref{LemmaOnDualDarbouxTrans}, 
in terms of characteristics of conservation laws. 

\begin{corollary}\label{CorollaryOnDualMultipleDarbouxTrans}
If $\,\psi^1$, \dots, $\psi^p$ are linearly independent solutions of the equation~$\widehat{\mathcal L}$ from class~\eqref{EqGenLPE}
and ${\rm DT}[\psi^1,\dots,\psi^p](\widehat{\mathcal L}\,)=\mathcal L$ 
then $\mathcal L$ belongs to the class~\eqref{EqGenLPE}, 
\[
\alpha^\varsigma=(-1)^{\varsigma-1}
\frac{W(\psi^1,\dots,\lefteqn{\smash{\diagdown}}\psi^\varsigma,\dots,\psi^p)}{W(\psi^1,\dots,\psi^p)}\in\Ch_{\rm f}(\mathcal L), 
\quad \varsigma=1,\dots,p,
\]
and 
${\rm DT}[\alpha^1,\dots,\alpha^p]\colon\Ch_{\rm f}(\mathcal L)\to\Ch_{\rm f}(\widehat{\mathcal L}\,)$. 
The tuples $\,(\psi^1,\dots,\psi^p)$ and $\,(\alpha^1,\dots,\alpha^p)$ will be called dual to each other.
\end{corollary}

\begin{note}
A statement on a connection of the same kind between the equations from two arbitrary steps of the iteration procedure 
can be formulated similarly to Theorem~\ref{TheoremOnDualMultipleDarbouxTrans}. 
It is sufficient to assume that the step of lower number is the start of the iteration 
and the step of greater number is the end of the iteration. In particular, 
\[
{\rm DT}[w^{s,\sigma},\dots,w^{s,s}]\,\lceil w^s\rfloor=\lceil w^{\sigma-1}\rfloor\quad\Leftrightarrow\quad
{\rm DT}[\beta^{\sigma-1,\sigma},\dots,\beta^{\sigma-1,s}]\,\lceil \beta^{\sigma-1}\rfloor=\lceil \beta^s\rfloor.
\]
\end{note}

\begin{note}\label{NoteOnMultipleDarbouxTransAsLinearMappingOfSolutionSpaceOfLPEs}
Analogously to simple Darboux transformations (see Note~\ref{NoteOnDarbouxTransAsLinearMappingOfSolutionSpaceOfLPEs}),
for any linearly independent solutions $w^{p,1}$, \dots, $w^{p,p}$ of~$\lceil w^p\rfloor$ 
the multiple Darboux transformation ${\rm DT}[w^{p,1},\dots,w^{p,p}]$
is a linear mapping from the solution space of~$\lceil w^p\rfloor$ into the solution space of~$\lceil u\rfloor$.
The kernel of this mapping coincides with the linear span~$\langle w^{p,1},\dots,w^{p,p}\rangle$. 
Its image is the whole solution space of~$\lceil u\rfloor$ since 
it is the composition of the simple Darboux transformations ${\rm DT}[w^{s,s}]$. 
For any~$s$ ${\rm DT}[w^{s,s}]$ is a linear mapping from the solution space of~$\lceil w^s\rfloor$ 
onto the solution space of~$\lceil w^{s-1}\rfloor$, where $w^0:=u$.
Therefore, ${\rm DT}[w^{p,1},\dots,w^{p,p}]$ generates a one-to-one linear mapping between 
the solution space of~$\lceil w^p\rfloor$, factorized by the subspace~$\langle w^{p,1},\dots,w^{p,p}\rangle$, 
and the solution space of~$\lceil u\rfloor$. 
\end{note}

The equations~$\lceil u\rfloor$ and~$\lceil w^p\rfloor$ are differential consequences 
of the $p$-level potential system~\eqref{EqMinIteratedGenPotSysOfLPEs}.
There is a one-to-one correspondence between solutions of the potential system and the equation~$\lceil w^p\rfloor$
due to the projection $(f^0,\dots,f^p)\to f^p$ and the transformation~$w^p=f^p/\beta^{p-1,p}$ in one direction
and due to the inverse transformation $f^p=\beta^{p-1,p}w^p$ and the backward recursive formula $f^{s-1}=f^s_x/H^s$ 
in the other. 
The correspondence between solutions of the initial equation~$\lceil u\rfloor$ and the potential system 
is one-to-one only up to arbitrary linear combinations of $p$ fixed solutions of system~\eqref{EqMinIteratedGenPotSysOfLPEs}. 
This follows, e.g., from the fact that every $v^s$ is determined via~$u$ up to an arbitrary constant summand and 
$(f^1,\dots,f^p)$ is the product of~$(v^1,\dots,v^p)$ and a matrix-function with coefficients depending on~$t$ and~$x$.
Taking into account the above arguments and Note~\ref{NoteOnCorrespodenceOfCharSubspaceForLPEs}, 
we will call the equation~$\lceil w^p\rfloor$ the \emph{modified potential equation, 
associated with the equation~$\lceil u\rfloor$ and the characteristics subspace 
$\,\langle\alpha^1,\dots,\alpha^p\rangle$}.

A system of the general form~\eqref{EqMinIteratedGenPotSysOfLPEs} with $H^1\dots H^pA\ne0$ is a $p$-level potential system 
of an equation from class~\eqref{EqGenLPE} only under special restrictions on the coefficients. 
Namely, the following conditions are necessary and sufficient:
\[
H^p_t+G^p_x=0,\qquad 
H^s_t=(AH^s)_{xx}-\left(\frac{G^{s+1}-AH^{s+1}_x}{H^{s+1}}H^s\right)_x,\quad s<p.
\]
Here the coefficients $G^s$, $s<p$, are calculated from $G^p$ and $H^\sigma$, $\sigma\geqslant s$, 
by the recursive formula 
\[
G^s=\frac{G^{s+1}-(AH^{s+1})_x}{H^{s+1}}H^s-AH^s_x.
\]  
Therefore, for any fixed~$s<p$ the function~$H^s$ should satisfy the Fokker--Planck equations with 
diffusion coefficient~$A$ and drift coefficient expressed via $A$, $G^p$ and $H^\sigma$, $\sigma>s$.

\section{General potential symmetries}\label{SectionOnGeneralPotSymsOfLPEs}

Consider a system of the general form~\eqref{EqMinIteratedGenPotSysOfLPEs} with $H^1\dots H^pA\ne0$, 
which is a $p$-level potential system of an equation from class~\eqref{EqGenLPE}. 
Such a system possesses $s-1$ algebraically independent nontrivial differential consequences 
$f^s_t=H^sAf^{s-1}_x-G^sf^{s-1}$, $s=1,\dots,p-1$ having, as differential equations, order one. 
System~\eqref{EqMinIteratedGenPotSysOfLPEs} implies the second-order partial differential equation 
\begin{equation}\label{EqPthLevelPotEqInFTermsForLPE}
f^p_t=Af^p_{xx}-\frac{G^p+AH^p_x}{H^p}f^p_x
\end{equation}
with respect to only the function~$f^p$.

\begin{definition}
The Lie invariance algebra of a $p$-order (or $p$-level) potential system of an equation~$\mathcal L$ 
from class~\eqref{EqGenLPE} is called a \emph{$p$-order potential symmetry algebra} of~$\mathcal L$. 
Any operator from this algebra is called a \emph{$p$-order potential symmetry operator} of~$\mathcal L$. 
\end{definition}

\begin{lemma}\label{LemmaOn1to1CorrespondenceBetweenLieSymsOfGenPotSysAndPthLevelPotEqOfLPE}
Let system~\eqref{EqMinIteratedGenPotSysOfLPEs} be a $p$-level potential system of an equation 
from class~\eqref{EqGenLPE}.
Then the maximal Lie invariance algebras of system~\eqref{EqMinIteratedGenPotSysOfLPEs} 
and equation~\eqref{EqPthLevelPotEqInFTermsForLPE} are isomorphic. 
Namely, for any Lie invariance operator $Q=\tau\p_t+\xi\p_x+\eta\p_u+\theta^s\p_{f^s}$ 
of system~\eqref{EqMinIteratedGenPotSysOfLPEs} 
its projection $Q'=\tau\p_t+\xi\p_x+\theta^p\p_{f^p}$ to the variables $(t,x,f^p)$ 
is a Lie invariance operator of equation~\eqref{EqPthLevelPotEqInFTermsForLPE}. 
The~coefficient~$\theta^{s-1}$ of~$Q$, where $\theta^0:=\eta$, 
is expressed via coefficients of the $(p-s+1)$-th prolongation of the operator~$Q'$ with respect to~$x$
in a backward recursive way in accordance with the equations $f^{\sigma-1}=f^\sigma_x/H^\sigma$. 
\end{lemma}

\begin{proof}
Similarly to Lemma~\ref{LemmaOn1to1CorrespondenceBetweenLieSymsOfPotSysAndPotEqOfLPE}, 
the one-to-one correspondence between the solutions of the system and the equation gives 
an empiric argument in favour of the lemma. 
Since this argument is not sufficient, the best way again is to make direct calculations. 
An application of the infinitesimal invariance criterion~\cite{Olver1986,Ovsiannikov1982} 
to system~\eqref{EqMinIteratedGenPotSysOfLPEs} has specific features due to the arbitrariness of~$p$ 
and the recursive form of the equations $f^\sigma_x=H^\sigma f^{\sigma-1}$ and needs a usage of tricks.
That is why we provide some explanations on the derivation of the determining equations 
of the coefficients of the Lie invariance operators. 

The first trick is to consider also the determining equations obtained via the application 
of the infinitesimal invariance criterion to the nontrivial first-order differential consequences 
of~\eqref{EqMinIteratedGenPotSysOfLPEs}. This allows us to simplify calculations at an early stage. 
The trick is admissible since the system extended by differential consequences possesses 
the same Lie invariance algebra as the initial system. 
(This is why usually differential consequences are used only under confining to the system manifold.)

Let $Q=\tau\p_t+\xi\p_x+\eta\p_u+\theta^s\p_{f^s}$ be a Lie invariance operator 
of system~\eqref{EqMinIteratedGenPotSysOfLPEs}. 
Here the coefficients $\tau$, $\xi$, $\eta$ and $\theta^s$ are smooth functions of the variables $t$, $x$, $u$ and $f^s$. 
Splitting the infinitesimal invariance conditions of the extended system with respect to the derivatives 
$u_t$ and $u_x$ (there are no other possibilities for the initial split), we obtain, in particular, the equations 
$
\tau_u=\xi_u=\theta^s_u=0,\ \eta_{uu}=0,\ \tau_x+\tau_{f^\sigma}H^\sigma f^{\sigma-1}=0.
$
They allow us to further split with respect to~$u$. 
Thus, splitting the already derived condition $\tau_x+\tau_{f^\sigma}H^\sigma f^{\sigma-1}=0$, 
we have $\tau_{f^1}=0$ and, therefore, allows to also split with respect to~$f^1$ in this condition. 
Iterating the splitting of this condition with respect to the dependent variables, we derive $\tau_x=\tau_{f^\sigma}=0$.
Other infinitesimal invariance conditions imply, under splitting with respect to~$u$, 
the equations $\xi_{f^1}=0$ and $\theta^1_{f^1\!f^1}=0$.

Since $\tau_u=\xi_u=\theta^s_u=0$ and $u$ appears only in the equations corresponding to the value $s=1$, 
these two equations with $s=1$ have no influence on the further splitting in the infinitesimal invariance conditions for 
the other equations. We can consider the subsystem formed by the equations with $s\geqslant2$ (this is the second trick). 
Here~$f^1$ and~$f^2$ play the same role as $f^0=u$ and $f^1$ in the whole system. 
Therefore, analogously to the whole system, for this subsystem we have the conditions 
$\theta^\sigma_{f^1}=0$, $\sigma\leqslant2$, $\xi_{f^2}=0$ and $\theta^2_{f^2\!f^2}=0$.
As a result of iterating the procedure of system contraction, we obtain 
$\theta^s_{f^\sigma}=0$, $\sigma<s$, $\xi_{f^s}=0$ and $\theta^s_{f^s\!f^s}=0$.

The third trick is that we can neglect the infinitesimal invariance conditions of the differential consequences 
from this moment on. Under the derived restrictions
the infinitesimal invariance conditions for the equations $f^{\sigma-1}=f^\sigma_x/H^\sigma$ 
imply the backward recursive formula
\[
\theta^{s-1}=\left(\theta^s_{f^s}-\xi_x-\tau\frac{H^s_t}{H^s}-\xi\frac{H^s_x}{H^s}\right)f^{s-1}
+\sum_{\sigma>s}\frac{H^\sigma}{H^s}\theta^s_{f^\sigma}f^{\sigma-1}+\frac{\theta^s_x}{H^s}
\]
with the start in $s=p$.
Therefore, $\theta^s_{f^\sigma\!f^\varsigma}=0$ for any $s$, $\sigma$ and $\varsigma$ 
in view of $\theta^p_{f^\sigma}=0$, $\sigma<p$ and $\theta^p_{f^p\!f^p}=0$.
Finally, the infinitesimal invariance conditions for the last equation of system~\eqref{EqMinIteratedGenPotSysOfLPEs} 
result in the equations
\begin{gather*}
(2\xi_x-\tau_t)A=\tau A_t+\xi A_x,\quad
\theta^p_t=A\theta^p_{xx}-\frac{G^p+AH^p_x}{H^p}\theta_x,
\\
\xi_t+(2\theta^p_{xf^p}-\xi_{xx})A=
\tau\left(\frac{G^p+AH^p_x}{H^p}\right)_t
+\xi\left(\frac{G^p+AH^p_x}{H^p}\right)_x
+(\tau_t-\xi_x)\frac{G^p+AH^p_x}{H^p},
\end{gather*}

The equations $\tau_u=\xi_u=\theta_u=0$, $\tau_{f^s}=\xi_{f^s}=0$ and $\theta^s_{f^\sigma}=0$, $\sigma<s$,
guarantee that the operator~$Q$ is projectable to the variables $(t,x,f^s,\dots,f^p)$ for any~$s$. 
The obtained system of determining equations contains, as a subsystem, the complete system of determining equations for 
Lie symmetry operators of equation~\eqref{EqPthLevelPotEqInFTermsForLPE}. 
Therefore, the projection $Q'=\tau\p_t+\xi\p_x+\theta^p\p_{f^p}$ of~$Q$ to the variables $(t,x,f^p)$ 
is a Lie invariance operator of equation~\eqref{EqPthLevelPotEqInFTermsForLPE}.
And vice versa, for any Lie invariance operator $Q'=\tau\p_t+\xi\p_x+\theta^p\p_{f^p}$ 
of equation~\eqref{EqPthLevelPotEqInFTermsForLPE}
the operator \[Q=Q'+\eta\p_u+\sum_{\sigma<p}\theta^\sigma\p_{f^\sigma},\] 
where the additional coefficients 
$\eta=\theta^0$ and $\theta^\sigma$, $\sigma<p$, are determined by the above backward recursive formula, 
is a Lie symmetry operator of system~\eqref{EqMinIteratedGenPotSysOfLPEs}. 
The formula can be re-written as 
\[
\theta^{s-1}=\frac1{H^s}\left(\theta^{s,x}-\tau\frac{H^s_t}{H^s}f^s_x-\xi\frac{H^s_x}{H^s}f^s_x\right)
\bigg|_{H^sf^{s-1}\rightsquigarrow f^s_x},
\]
where there is no summation over~$s$ (since it is assumed fixed) and 
$\theta^{s,x}$ is the coefficient of $\p_{f^s_x}$ in the first prolongation of the operator~$\tau\p_t+\xi\p_x+\theta^s\p_{f^s}$.
Therefore the coefficient~$\theta^{s-1}$ is expressed via the coefficients of the standard $(p-s+1)$-st prolongation 
of the operator~$Q'$ with respect to~$x$ in accordance with the backward recursive equations 
$f^{\sigma-1}=f^\sigma_x/H^\sigma$. 
\end{proof}

\begin{corollary}\label{CorollaryOnReducibilityOfPotSymsOfLPE}
If $\,\theta^{s-1}_{f^\sigma}=0$ for a fixed value of~$s$ and any $\sigma>s-1$ 
then $\,\theta^{\varsigma-1}_{f^\sigma}=0$ for any $\varsigma\leqslant s$ and any $\sigma>s-1$.
\end{corollary}

\begin{proof}
If $\theta^{s-1}_{f^\sigma}=0$, $\sigma>s-1$, then
the coefficient $\theta^{\varsigma-1}$ for any $\varsigma<s$ is expressed via 
the coefficients of the standard $(s-\varsigma+1)$-st prolongation 
of the operator~$\tau\p_t+\xi\p_x+\theta^{s-1}\p_{f^{s-1}}$ with respect to~$x$ ($s$ is fixed!) and 
the coefficients of the prolongation do not depend on $f^\sigma$, $\sigma>s-1$.
\end{proof}

In other words, if $Q=\tau\p_t+\xi\p_x+\eta\p_u+\theta^\varsigma\p_{f^\varsigma}$ is a Lie invariance operator 
of system~\eqref{EqMinIteratedGenPotSysOfLPEs} and $\theta^s_{f^\sigma}=0$ for a fixed~$s$ and any $\sigma>s$ 
then the truncated operator $\check Q=\tau\p_t+\xi\p_x+\eta\p_u+\sum_{\varsigma=1}^s{\theta^\varsigma}\p_{f^\varsigma}$ 
is a Lie invariance operator of the $s$-level potential system 
\[
f^\varsigma_x=H^\varsigma f^{\varsigma-1},\quad \varsigma\leqslant s,\qquad f^s_t=H^sAf^{s-1}_x-G^sf^{s-1}, 
\]
which is a subsystem of system~\eqref{EqMinIteratedGenPotSysOfLPEs} extended by its differential consequences. 
This means that the corresponding potential symmetry of the initial equation is also induced 
by the truncated operator $\check Q$ and, therefore, can be assumed to have order less than~$p$. 

We also need to prove the statement on generalized potential symmetries of equations from class~\eqref{EqGenLPE}, 
which is similar to Lemma~\ref{LemmaOn1to1CorrespondenceBetweenLieSymsOfGenPotSysAndPthLevelPotEqOfLPE}. 
This finally justifies the choice of conserved vectors in the canonical form~\eqref{eqCVofLPEs} 
for the construction of potential systems. 

\begin{lemma}\label{LemmaOn1to1CorrespondenceBetweenGeneralizedSymsOfGenPotSysAndPthLevelPotEqOfLPE}
Let system~\eqref{EqMinIteratedGenPotSysOfLPEs} be a $p$-level potential system of an equation 
from class~\eqref{EqGenLPE}.
Up to the equivalence of generalized symmetries, 
every generalized symmetry operator of system~\eqref{EqMinIteratedGenPotSysOfLPEs} is obtained via 
the $p$-th order prolongation, with respect to only the variable~$x$, of a generalized symmetry operator of 
the corresponding $p$-level modified potential equation~\eqref{EqPthLevelPotEqInFTermsForLPE} and 
expressing the derivatives of~$f^p$ by $f^s$ and derivatives of~$u$ according to system~\eqref{EqMinIteratedGenPotSysOfLPEs}.  
In particular, up to the equivalence of generalized symmetries, 
the coefficients of every generalized symmetry operator of system~\eqref{EqMinIteratedGenPotSysOfLPEs} 
depends at most on $t$, $x$, $f^s$ and derivatives of~$u$ with respect to~$x$ 
and are linear with respect to the dependent variables and the derivatives of~$u$.
\end{lemma}

\begin{proof}
Suppose that $Q=\eta\p_u+\theta^s\p_{f^s}$ is a generalized symmetry operator of system~\eqref{EqMinIteratedGenPotSysOfLPEs}. 
Here the coefficients $\theta^0:=\eta$ and $\theta^s$ are functions of $t$, $x$ and derivatives of~$u$ and $f^s$. 
Due to system~\eqref{EqMinIteratedGenPotSysOfLPEs}, its differential consequences and the equivalence relation of generalized symmetries, 
we can exclude the derivatives of $f^s$ of nonzero orders and derivatives of~$u$ containing differentiations
with respect to~$t$ from the coefficients of~$Q$. 
Therefore, they are assumed to depend, at most, on $t$, $x$, $f^s$ and derivatives of~$u$ with respect to~$x$. 

We temporarily introduce the notations $u_k:=\p^ku/\p x^k$, $k\geqslant1$, $u_0:=f^0=u$, $u_{-s}:=f^s$, 
$\mathop{\rm ord}\nolimits \theta^\mu=\max\{k\mid \p\theta^\mu/\p u_k\ne0,\,k\geqslant-p\}$, 
$r:=\mathop{\rm ord}\nolimits \theta^p$, $r\geqslant-p$. 
The infinitesimal invariance condition~\cite{Olver1986,Ovsiannikov1982} applied to the equation $f^s_x=H^sf^{s-1}$ 
implies the formula 
\begin{equation}\label{EqExpressionForCoeefsOfGenSymOpsOfLPEs}
\theta^{s-1}=\frac1{H^s}\biggl(\theta^s_x+\theta^s_{f^\sigma}H^\sigma f^{\sigma-1}+\sum_{k\geqslant0}\theta^s_{u_k}u_{k+1}\biggr).
\end{equation}
Iterating formula~\eqref{EqExpressionForCoeefsOfGenSymOpsOfLPEs} backward starting from $s=p$, 
we obtain the expression of $\theta^{s-1}$ via 
the total derivatives of $\theta^p$ with respect to~$x$ up to order $p-s+1$ according to the equation 
$f^{s-1}=(H^s)^{-1}\p_x\bigl((H^{s+1})^{-1}\p_x\bigl(\dots\bigl((H^p)^{-1}\p_xf^p\bigr)\dots\bigr)\bigr).$ 
In particular, $\mathop{\rm ord}\nolimits \theta^{s-1}=r+p-s+1$.

In view of the other equations of system~\eqref{EqMinIteratedGenPotSysOfLPEs}, the last equation $f^p_t=H^pAf^{p-1}_x-G^pf^{p-1}$ 
of the system is equivalent to equation~\eqref{EqPthLevelPotEqInFTermsForLPE}. 
The equations on $\theta^p$ derived from the infinitesimal invariance condition for equation~\eqref{EqPthLevelPotEqInFTermsForLPE} 
as a differential consequence of system~\eqref{EqMinIteratedGenPotSysOfLPEs} 
coincide with the determining equations for the generalized symmetry operator $Q'=\theta^p\p_{f^p}$ 
of the single equation~\eqref{EqPthLevelPotEqInFTermsForLPE} 
under the above relations between the $x$-derivatives of~$f^p$ and the functions $u_k$, $k\geqslant-p$. 
Therefore, there is a one-to-one correspondence between generalized symmetries of system~\eqref{EqMinIteratedGenPotSysOfLPEs} 
and equation~\eqref{EqPthLevelPotEqInFTermsForLPE}, established via the projection to the prolongation space over $(t,x,f^p)$ 
in the forward direction and the prolongation by formula~\eqref{EqExpressionForCoeefsOfGenSymOpsOfLPEs} in the reverse direction.
It is implicitly assumed that the $x$-derivatives of~$f^p$ can be expressed in terms of the functions~$u_k$, $k\geqslant-p$, and vice versa. 

We can split the infinitesimal invariance condition for equation~\eqref{EqPthLevelPotEqInFTermsForLPE} with respect to $u_{r+1}$.
(This is equivalent to splitting with respect to $f^p_{p+r+1}$ when \eqref{EqPthLevelPotEqInFTermsForLPE} is considered as a single equation.) 
Collecting the coefficient of~$u_{r+1}^2$, we derive the equation $\theta^p_{u_ru_r}=0$. 
In view of the last equation, we can collect the coefficient of~$u_{r+1}u_r$, yielding $\theta^p_{u_ru_{r-1}}=0$.
Iterating the procedure, in the $l$-th step ($l\leqslant r+p+1$) we can collect the coefficient of~$u_{r+1}u_{r+2-l}$ 
and obtain the equation $\theta^p_{u_ru_{r+1-l}}=0$. 
The above procedure can be repeated for the terms containing $u_r$ and so on. 
Finally we derive the equations $\theta^p_{u_ku_{k'\!}}=0$, $-p\leqslant k, k'\leqslant r$.
Then in view of the iterative formula~\eqref{EqExpressionForCoeefsOfGenSymOpsOfLPEs} we have 
$\theta^{s-1}_{u_ku_{k'\!}}=0$, $-p\leqslant k, k'\leqslant r+p-s+1$.
\end{proof}

\begin{corollary}\label{CorollaryOnChiceOfConservedVertorsForConstructionOfPotSymsOfLPE}
To exhaustively investigate potential symmetries of equations from class~\eqref{EqGenLPE}, 
it suffices to only consider tuples of conserved vectors of the form~\eqref{eqCVofLPEs}, 
which are canonical representatives of the corresponding conservation laws. 
\end{corollary}

\begin{proof}
Consider $p$ conserved vectors of an equation of the form~\eqref{EqGenLPE}, 
corresponding to $p$ linearly independent conservation laws. 
They necessarily possess the representation  
\begin{equation}\label{EqTupleOfEquivConseredVectorsOfLPE}
\bigl(\alpha^s u+D_x\Phi^s,\, -\alpha^s Au_x+((\alpha^s A)_x-\alpha^s B)u-D_t\Phi^s\bigr),
\end{equation}
where $\alpha^s=\alpha^s(t,x)$, $s=1,\dots,p$, are linearly independent solutions of the adjoint equation~\eqref{eqCVofLPEs} 
and $\Phi^s$ are functions of $t$, $x$ and derivatives of~$u$. 
Due to the initial equations we can assume that the derivatives are only with respect to~$x$.
The corresponding potential systems are obtained by the substitution $v^s=\tilde v^s-\Phi^s$
from the system~\eqref{EqGenPotSysOfLPEs} associated with 
the tuple of equivalent conserved vectors in the canonical form. 
Here $\tilde v^s$ are the potentials generated by the conserved vectors~\eqref{EqTupleOfEquivConseredVectorsOfLPE}. 
In view of Lemmas~\ref{LemmaOnRepresentationOfGenPotsViaSimplestPotsForLPEs1} 
and~\ref{LemmaOnRepresentationOfGenPotsViaSimplestPotsForLPEs2} 
the induced substitution in terms of the potentials~$f^s$ has the form $f^s=\tilde f^s-\Psi^s$. 
For each value of~$s$ the function $\Psi^s$ is a linear combination of $\Phi^\sigma$, $\sigma=1,\dots,s$, 
with coefficients depending on~$t$ and~$x$, and the coefficient of $\Phi^s$ does not vanish.
Denote the maximal order of derivatives in $\Psi^s$ by~$\rho$ and a value of~$s$ with $\Psi^s_{u_\rho}\ne0$ by $s_0$. 
We assume $\rho\geqslant1$ since otherwise  the functions~$\Psi^s$ and, therefore, the functions~$\Phi^s$
can be neglected due to the point transformation, which has, up to similarity, no influence on 
Lie symmetries of the involved potential systems. 

Every Lie symmetry of the system in~$\tilde f^s$ and~$u$ corresponds to 
a generalized symmetry of system~\eqref{EqMinIteratedGenPotSysOfLPEs} in~$f^s$ and~$u$.
The question is when a generalized symmetry operator~$Q=\eta\p_u+\theta^s\p_{f^s}$ of system~\eqref{EqMinIteratedGenPotSysOfLPEs} induces 
a Lie symmetry operator~$\tilde Q=\tilde\eta\p_u+\tilde\theta^s\p_{f^s}$ 
(in evolution form) of the system in~$\tilde f^s$ and~$u$.

We use the notations of Lemma~\ref{LemmaOn1to1CorrespondenceBetweenGeneralizedSymsOfGenPotSysAndPthLevelPotEqOfLPE}. 
Suppose that $r+p\geqslant1$, where $r=\mathop{\rm ord}\nolimits \theta^p$. Since 
\[
\tilde\theta^{s_0}=Q\tilde f^{s_0}\bigl|_{\tilde f^s-\Psi^s\rightsquigarrow f^s}=
\biggl(\theta^{s_0}+\sum_{k=0}^\rho D_x^k(\eta)\Psi^s_{u_k}\biggl)\biggl|_{\tilde f^s-\Psi^s\rightsquigarrow f^s}, 
\]
the term  $\eta_{u_{r+p}}\Psi^s_{u_\rho}u_{r+p+\rho}$ cannot be canceled with other terms of~$\tilde\theta^{s_0}$. 
Therefore, $\mathop{\rm ord}\nolimits\tilde\theta^{s_0}=r+p+\rho\geqslant2$. 
At the same time, if $\tilde Q$ would be a Lie symmetry operator, 
$\mathop{\rm ord}\nolimits\tilde\theta^s\leqslant\max(2-s,0)\leqslant1.$ 

In the case $r+p=0$ the coefficient $\theta^p$ has the form $\theta^p=\theta^{p1}(t,x)f^p+\theta^{p0}(t,x)$. 
The determining equations for the generalized symmetry of~\eqref{EqMinIteratedGenPotSysOfLPEs} imply that $\theta^{p1}(t,x)=C=\const$.  
Then $\theta^{s-1}=Cf^{s-1}+\theta^{s-1,0}(t,x)$ in view of formula~\eqref{EqExpressionForCoeefsOfGenSymOpsOfLPEs}. 
Since $\tilde\eta=\eta=Cu+\theta^{00}(t,x)$, the operator~$\tilde Q$ gives a trivial potential symmetry of the initial equation, 
which corresponds to a trivial Lie symmetry of this equation.
\end{proof}

In view of Lemma~\ref{LemmaOn1to1CorrespondenceBetweenLieSymsOfGenPotSysAndPthLevelPotEqOfLPE}, 
Lie symmetry analysis of any $p$-level potential system associated with an equation from class~\eqref{EqGenLPE}
is reduced to similar investigation for the corresponding $p$-level potential equation. 
In fact, we will investigate the modified $p$-level potential equation~$\lceil w^p\rfloor$ 
instead of~$\lceil f^p\rfloor$. Such way has two advantages. 
Firstly, the potential~$w^p$ does not depend on nondegenerate linear combining 
(including a rearrangement, cf. Corollary~\ref{CorollaryOnIndependenceOfPotsOnOrderOfCharsForLPEs}) 
in the characteristic tuple. 
Secondly, the modified $p$-level potential equation is connected with the initial one via 
the multiple Darboux transformation expressed in terms of the corresponding characteristics. 
In terms of~$w$'s the $p$-level potential system has the form
\begin{gather}\nonumber
w^s_x+\frac{\beta^{s-1,s}_x}{\beta^{s-1,s}}w^s=w^{s-1}, \\ 
w^p_t+\frac{\beta^{p-1,p}_t}{\beta^{p-1,p}}w^p=Aw^{p-1}_x-
\biggl(A_x+\frac{\beta^{p-1,p}_x}{\beta^{p-1,p}}A-B^{p-1}\biggr)w^{p-1}\label{EqMinModifiedIteratedGenPotSysOfLPEs}
\end{gather}
which will be called t\emph{he modified $p$-level potential system} associated with an equation from class~\eqref{EqGenLPE} 
and the characteristic tuple $(\alpha^1,\dots,\alpha^p)$.

Analogously to the case of simplest potential symmetries, 
the only possibility of obtaining pure potential symmetries is connected with the coefficient~$\eta$. 
Namely, the condition $\eta_{f^s}\ne0$ (or $\eta_{w^s}\ne0$ in terms of $w$'s) should be satisfied for some~$s$. 
In fact, we are interested in the investigation $p$-order potential symmetries only in the case when 
it is not reduced to the consideration of potential symmetries of a smaller order. 
The irreducibility is naturally defined in terms of the~$w$'s. 
One of the reasons for this again is the independence of~$w^p$ of nondegenerate linear combining 
in the characteristic tuple. 
Another reason is the following. 
Let $Q=\tau\p_t+\xi\p_x+\eta\p_u+\theta^\varsigma\p_{f^\varsigma}$ and $Q=\tau\p_t+\xi\p_x+\eta\p_u+\zeta^\varsigma\p_{w^\varsigma}$ 
be the representations of the same operator~$Q$ in terms of $f$'s and $w$'s. 
In view of Corollary~\ref{CorollaryOnPresentationOfPotsViaFixedSolutionsOfPotEqForLPEs}
the conditions `$\theta^{s-1}_{f^\sigma}=0$ for a fixed value of~$s$ and any $\sigma>s-1$' 
and `$\smash{\theta^{\varsigma-1}_{f^\sigma}=0}$ for any $\varsigma\leqslant s$ and any $\sigma>s-1$' are equivalent to
the similar conditions in terms of the~$w$'s, i.e., 
`$\zeta^{s-1}_{w^\sigma}=0$ for a fixed value of~$s$ and any $\sigma>s-1$' 
and `$\zeta^{\varsigma-1}_{w^\sigma}=0$ for any $\varsigma\leqslant s$ and any $\sigma>s-1$'. 
Therefore, Corollary~\ref{CorollaryOnReducibilityOfPotSymsOfLPE} can be completely reformulated in terms of the~$w$'s.

{\addtocounter{corollary}{-1}\renewcommand{\thecorollary}{\arabic{corollary}$'$}
\begin{corollary}\label{CorollaryOnReducibilityOfPotSymsOfLPE2}
If $\,\zeta^{s-1}_{w^\sigma}=0$ for a fixed value of~$s$ and any $\sigma>s-1$ 
then $\,\zeta^{\varsigma-1}_{w^\sigma}=0$ for any $\varsigma\leqslant s$ and any $\sigma>s-1$.
\end{corollary}}

\begin{definition}\label{DefinitionOfStrictlyPthOrderPorSyms}
Let $Q'=\tau\p_t+\xi\p_x+\zeta^p\p_{w^p}$ be a Lie invariance operator of the modified $p$-level potential equation
associated with an equation from class~\eqref{EqGenLPE} and the characteristic subspace $\langle\alpha^1,\dots,\alpha^p\rangle$. 
We will say that \emph{$Q'$ generates a strictly $p$-th order potential symmetry} of the initial equation if 
for any basis $(\tilde\alpha^1,\dots,\tilde\alpha^p)$ in $\langle\alpha^1,\dots,\alpha^p\rangle$ 
the prolongation $Q=Q'+\tilde\zeta^{s-1}\p_{\tilde w^{s-1}}$ of~$Q'$ to the corresponding 
potentials $\tilde w^{s-1}$, where $\tilde w^0:=u$, satisfies the following condition. 
For any~$s$ there exists $\sigma>s-1$ such that $\tilde\zeta^{s-1}_{\tilde w^\sigma}\ne0$.
\end{definition}

\begin{definition}\label{DefinitionOfStrictlyPthOrderPorSymAlgebra}
Let $\mathfrak g$ be a Lie invariance algebra of the modified $p$-level potential equation
associated with an equation from class~\eqref{EqGenLPE} and the characteristic subspace 
$\langle\alpha^1,\dots,\alpha^p\rangle$. 
We will say that \emph{$\mathfrak g$ generates a strictly $p$-th order potential symmetry algebra} 
of the initial equation if for any basis $(\tilde\alpha^1,\dots,\tilde\alpha^p)$ in 
$\langle\alpha^1,\dots,\alpha^p\rangle$ and for any~$s$
there exists $Q'\in\mathfrak g$ whose prolongation $Q=Q'+\tilde\zeta^{s-1}\p_{\tilde w^{s-1}}$ to the corresponding 
potentials $\tilde w^{s-1}$, where $\tilde w^0:=u$, satisfies the following condition. 
There exists $\sigma>s-1$ such that $\tilde\zeta^{s-1}_{\tilde w^\sigma}\ne0$.
\end{definition}

Roughly speaking, a potential symmetry operator (resp. algebra) is strictly of order $p$
if it cannot be obtained from a smaller number of conservation laws and potentials. 

Let us recall that $w^0=u$ and $\zeta^0=\eta$ by definition. 
It follows from Lemma~\ref{LemmaOn1to1CorrespondenceBetweenLieSymsOfGenPotSysAndPthLevelPotEqOfLPE} 
in view of the formula $w^s=f^s/\beta^{s-1,s}$ that 
\[
\zeta^s=\sum_{\sigma=s}^p\zeta^{s\sigma}(t,x)w^\sigma+\varrho^s(t,x),\quad
\zeta^0=\zeta^{00}(t,x)w^0+\sum_{\sigma=1}^p\zeta^{0\sigma}(t,x)w^\sigma+\varrho^0(t,x)
\]
and a more precise form of the coefficients $\zeta^{s-1}$ is calculated by an above backward recursive formula 
involving $\tau$, $\xi$ and $\zeta^p$, in accordance with equations of~\eqref{EqMinModifiedIteratedGenPotSysOfLPEs}. In particular, 
\begin{gather*}
\zeta^{p-1,p}=\frac1{w^{p,p}}{\rm DT}[w^{p,p}](Q'[w^{p,p}])=
\frac{W(w^{p,p},Q'[w^{p,p}])}{w^{p,p}\,W(w^{p,p})},
\\
\zeta^{s-1,p}=\frac1{w^{p,p}}{\rm DT}[w^{s,s}](w^{p,p}\zeta^{sp})=
\frac{W(w^{p,s+1},\dots,w^{p,p},Q'[w^{p,p}])}{w^{p,p}\,W(w^{p,s+1},\dots,w^{p,p})}, \quad s<p,
\end{gather*}
by the Crum theorem. Here $Q'[w^{p,p}]=\zeta^{pp}w^{p,p}-\tau w^{p,p}_t-\xi w^{p,p}_x$.

Consider the prolongation of transformations from the equivalence group~$G^\sim$ to the general potential frame. 
Any point equivalence transformation~$\mathcal T$ in class~\eqref{EqGenLPE} acts on the variables $(t,x,u)$
by the formulas $\tilde t=T(t)$, $\tilde x=X(t,x)$, $\tilde u=U^1(t,x)u$, where $T_tX_xU^1\ne0$.
In view of Proposition~\ref{PropositionOnTransOfCharsOfCLsForLPEs} it is prolonged to characteristics $\alpha^s$
of conservation laws of an equation from this class: $\tilde\alpha^s=\alpha^s/(X_xU^1)$.
The corresponding potential~$v^s$ can be assumed to be transformed identically under prolongation of~$\mathcal T$. 
(See Section~\ref{SectionOnSimplestPotSymsOfLPEs}.) 
Taking into account the constructed representations for $\beta^{s,\sigma}$, $g^{s,\sigma}$, $\sigma\geqslant s$, 
$f^s$ and $w^s$ via $\alpha^s$ and~$v^s$, we obtain the following statements.

\begin{lemma}\label{LemmaOnEquivTransOfGenPotFrameOfLPEs}
For any $p\in\mathbb N$ any transformation~$\mathcal T\in G^\sim$ is prolonged to the $p$-order potential frame 
over class~\eqref{EqGenLPE}. 
The prolonged transformation $\mathcal T^p$ 
\begin{gather*}
\tilde t=T(t),\quad \tilde x=X(t,x), \quad \tilde u=U^1(t,x)u, \quad  
\tilde\alpha^s=\frac{\alpha^s}{X_xU^1}, \quad \tilde v^s=v^s, \\
\tilde\beta^{s,\sigma}=\frac{\beta^{s,\sigma}}{X_x{}^{s+1}U^1}, \quad \tilde g^{s,\sigma}=g^{s,\sigma}, \quad 
\tilde f^s=f^s, \quad \tilde w^s=X_x{}^sU^1w^s, \quad \tilde w^{s,\sigma}=X_x{}^sU^1w^{s,\sigma},\\
\tilde A=\frac{X_x^2}{T_t}A,\quad 
\tilde B=\frac{X_x}{T_t}\left(B-2\frac{U^1_x}{U^1}A\right)-\frac{X_t-AX_{xx}}{T_t},\quad 
\tilde C=-\frac{U^1}{T_t}L\frac1{U^1},
\end{gather*}
where $T_tX_xU^1\ne0$, realizes a simultaneous transformation between the tuples consisting of
the initial, potential and modified potential equations and the potential systems in terms of $v^s$ and~$f^s$.
The transformations from~$G^\sim$, prolonged to the $p$-order potential frame over the class~\eqref{EqGenLPE}, 
form the group~$G^\sim_{\smash{[p]}}$ called the equivalence group of this frame. 
\end{lemma}

\begin{note}
In general, under the transformation~$\mathcal T$: 
$\tilde t=T(t)$, $\tilde x=X(t,x)$, $\tilde\psi=\Phi(t,x)\psi$, $\tilde\psi^s=\Phi(t,x)\psi^s$ we have 
\[
\bigl({\rm DT}[\tilde\psi^1,\dots,\tilde\psi^s]\tilde\psi\bigr)(\tilde x)
=\frac{\Phi(t,x)}{(X_x(t,x))^s}
\bigl({\rm DT}[\psi^1,\dots,\psi^s]\psi\bigr)(x).
\]
\end{note}

\begin{note}
The prolonged transformations from~$G^\sim$ do not exhaust all possible equivalence transformations of 
the $p$-order potential frame over the class~\eqref{EqGenLPE}. 
They can be extended, e.g., with linear combining of characteristics. 
If the variables $(t,x,u)$ (and, therefore, the arbitrary elements $(A,B,C)$) are not transformed and 
$\tilde\alpha^s=\alpha^\sigma c_{\sigma s}$, where 
$c_{\sigma s}=\const$, $\det(c_{\sigma s})\ne0$ and $c_{\sigma s}=0$, $\sigma>s$, 
then the corresponding transformation of the other functions appearing in the potential frame is easily constructed:
\[
\tilde v^s=v^\sigma c_{\sigma s},\quad 
\tilde\beta^{s,\sigma}=\beta^{s,\varsigma}c_{\varsigma\sigma},\quad 
\tilde g^{s,\sigma}=g^{s,\varsigma}c_{\varsigma\sigma},\quad 
\tilde f^s=f^sc_{ss},\quad 
\tilde w^s=w^s,\quad 
\tilde w^{s,\sigma}=w^{s,\varsigma}\hat c_{\sigma\varsigma}.
\] 
Here $(\hat c_{\varsigma\sigma})$ is the inverse matrix to~$(c_{\varsigma\sigma})$.
\end{note}

The transformations from~$G^\sim_{\smash{[p]}}$ preserve the determining equations derived in the proof of 
Lemma~\ref{LemmaOn1to1CorrespondenceBetweenLieSymsOfGenPotSysAndPthLevelPotEqOfLPE}. 
Let $Q=\tau\p_t+\xi\p_x+\eta\p_u+\theta^s\p_{f^s}$ be a Lie invariance operator 
of system~\eqref{EqMinIteratedGenPotSysOfLPEs}. 
The coefficients of~$Q$ are transformed under the operator mapping generated 
by~$\mathcal T^p\in G^\sim_{\smash{[p]}}$ 
by the formulas 
\[
\tilde\tau=\tau T_t,\quad \tilde\xi=\tau X_t+\xi X_x,\quad \tilde\eta=\tau U^1_tu+\xi U^1_xu+U^1\eta,\quad \tilde\theta^s=\theta^s.
\] 
Therefore, both the conditions $\eta_{f^\sigma}\ne0$ and $\eta_{f^\sigma}=0$ 
as well as the conditions $\theta^s_{f^\sigma}\ne0$ and $\theta^s_{f^\sigma}=0$
are preserved by the transformations from~$G^\sim_{\smash{[p]}}$ for any~$s$ and~$\sigma$.
This means that pure $p$-order potential symmetries of any equation from class~\eqref{EqGenLPE} are not mixed
under the transformations from~$G^\sim_{\smash{[p]}}$ with either potential symmetries of lesser orders or Lie symmetries.
The dimension of the factor-space of $p$-order potential symmetry operators 
corresponding to a $p$-tuple of characteristics with respect to the subspaces of potential symmetry operators 
of lesser orders are also not changed. 
Therefore, $p$-order potential symmetries of equations from class~\eqref{EqGenLPE} can be studied up to 
the equivalence relation generated by transformations from~$G^\sim_{\smash{[p]}}$.

There are different ways of employing this equivalence relation. 
One of them is to emphasize the simplification of the form of the equations under consideration. 
In particular, we can put $A=1$ and $B=0$ (and re-denote~$C$ by $-V$)
in~\eqref{EqGenLPE}, which implies $B^s=B-sA_x=0$. 
As a result, the symmetry analysis of the $p$-order potential frame over class~\eqref{EqGenLPE} is reduced to 
the symmetry analysis of the $p$-order potential frame over class~\eqref{EqReducedLPE}.
The equivalence group~$G^\sim_1$ of the reduced class~\eqref{EqReducedLPE} is canonically isomorphic to 
a subgroup of the equivalence group~$G^\sim$ with special restrictions of the parameter-functions~$X$ and $U^1$.
In view of Lemma~\ref{LemmaOnEquivTransOfGenPotFrameOfLPEs} the transformations from~$G^\sim_1$ 
are prolonged to equivalence transformations of the whole reduced potential frame. 
The prolonged transformations form the group $\hat G^\sim_{\smash{[p]}}$ 
which is canonically isomorphic via projection 
to the equivalence group of the class of modified $p$-level potential equations for 
the equations of the form~\eqref{EqReducedLPE}. 
Hence, the classification of potential symmetries of the class~\eqref{EqReducedLPE} follows from 
the group classification of the same class in terms of $(w^p,V^p)$ instead of~$(u,V)$.

Another way to proceed is to simplify the form of the operators. 
A Lie invariance operator $Q'$ of an equation from class~\eqref{EqGenLPE} with a nonvanishing coefficient of~$\p_t$ 
(or a vanishing coefficient of~$\p_t$ and a nonvanishing coefficient of~$\p_x$) can 
be reduced by transformations from~$G^\sim$ to the form $Q'=\p_{\tilde t}$ ($Q'=\p_{\tilde x}$).
Note that an appropriate simplification of the form of the equations leads to a
simplification of the form of their symmetry operators and vice versa. 
The choice of how to employ the above equivalence relation depends on the problems under consideration. 

At first, we formulate a criterion when a Lie invariance operator of the modified $p$-level potential equation 
generates a strictly $p$-th order potential symmetry of the initial equation. 
A simplification of the form of symmetry operators will be very helpful in this problem.

\begin{theorem}\label{TheoremOnCriterionOfStrictlyPthOrderPotSymsForLPE}
Suppose that $\lceil w^p\rfloor$ is an equation from class~\eqref{EqGenLPE}, 
$\psi^1$, \dots, $\psi^p$ are its linearly independent solutions and 
${\rm DT}[\psi^1,\dots,\psi^p]\,\lceil w^p\rfloor=\lceil u\rfloor$. 

1) Trivial Lie invariance operators of~$\lceil w^p\rfloor$ generate only trivial Lie invariance operators 
of~$\lceil u\rfloor$. 
Namely, $w^p\p_{w^p}$ generates $u\p_u$ and $\psi(t,x)\p_{w^p}$ generates ${\rm DT}[\psi^1,\dots,\psi^p](\psi)\p_u$. 
Here the function $\psi=\psi(t,x)$ runs through the solution set of~$\lceil w^p\rfloor$ and, 
therefore, ${\rm DT}[\psi^1,\dots,\psi^p](\psi)$ runs through the solution set of~$\lceil u\rfloor$. 

2) Let $Q'$ be an essential Lie invariance operator of~$\lceil w^p\rfloor$.
Then $Q'$ generates a strictly $p$-th order potential symmetry of~$\lceil u\rfloor$ iff 
any subspace of~$\langle\psi^1,\dots,\psi^p\rangle$ is not invariant under the action 
of the associated first-order differential operator~$\widehat Q'$. 
\end{theorem}

\begin{proof}
Item~1 is obvious in view of the properties of the Darboux transformation. 
It is significant only as a complement to item~2. 

To prove item~2 we show the equivalence of the respective negations.
Namely, $Q'$ does not generate a strictly $p$-th order potential symmetry of~$\lceil u\rfloor$ iff 
there exists a subspace of~$\langle\psi^1,\dots,\psi^p\rangle$ which is invariant under the action 
of the associated first-order differential operator~$\widehat Q'$.

Let us recall that any (nontrivial) essential Lie invariance operator~$Q'$ of~$\lceil w^p\rfloor$ 
has the form \[Q'=\tau\p_t+\xi\p_x+\zeta^{p1}w^p\p_{w^p},\] 
where $\tau=\tau(t)$, $\xi=\xi(t,x)$ and $\zeta^{p1}=\zeta^{p1}(t,x)$ 
are smooth functions of their arguments and $(\tau,\xi)\ne(0,0)$. 
The first-order differential operator~$\widehat Q'$ associated with $Q'$ acts on functions of~$t$ and~$x$ 
according to the formula $\widehat Q'\psi=Q'[\psi]=\zeta^{p1}-\tau\psi_t-\xi\psi_x$. 
Due to the equivalence relation generated by transformations from~$G^\sim_{\smash{[p]}}$ 
we can assume without loss of generality that $Q'\in\{\p_t,\p_x\}$. 
Then the standard prolongation~$Q'_{(q)}$ of~$Q'$ of any order~$q$ formally coincides with~$Q'$.

Suppose that a $q$-dimensional subspace~$\mathcal I$ in the linear space~$\langle\psi^1,\dots,\psi^p\rangle$ is invariant
under action of~$\widehat Q'$. Let $s=p-q$.
Without loss of generality we can choose a basis $\{w^{p,1},\dots,w^{p,p}\}$ in~$\langle\psi^1,\dots,\psi^p\rangle$ 
in such a way that the functions $w^{p,s+1}$, \dots, $w^{p,p}$ form a basis of~$\mathcal I$. Therefore, 
\[
Q'[w^{p,\sigma}]=\sum_{\varsigma=s+1}^p\kappa_{\sigma\varsigma}w^{p,\varsigma},\quad \sigma=s+1,\dots,p, 
\] 
where $\kappa_{\sigma\varsigma}$, $\sigma,\varsigma=s+1,\dots,p$, are constants.
Since $w^{s}={\rm DT}[w^{p,s+1},\dots,w^{p,p}](w^p)$ then the coefficient of $\p_{w^{s}}$ in 
the corresponding Lie invariance operator~$Q$ of system~\eqref{EqMinModifiedIteratedGenPotSysOfLPEs} is 
$\zeta^{s}=Q'_{(q)}w^{s}$, i.e., 
\begin{gather*}
\zeta^{s}=\sum_{\sigma=s+1}^p\left(
\frac{W(w^{p,s+1},\dots,w^{p,p},w^p)_{Q'[w^{p,\sigma}]\rightsquigarrow w^{p,\sigma}}}{W(w^{p,s+1},\dots,w^{p,p})}-
\frac{W(w^{p,s+1},\dots,w^{p,p})_{Q'[w^{p,\sigma}]\rightsquigarrow w^{p,\sigma}}}{W(w^{p,s+1},\dots,w^{p,p})}w^s\right)
\\
\phantom{\zeta^{s}}=\sum_{\sigma=s+1}^p(\kappa_{\sigma\sigma}-\kappa_{\sigma\sigma})w^{s}=0.
\end{gather*}
Corollary~\ref{CorollaryOnReducibilityOfPotSymsOfLPE2} implies 
$\zeta^\varsigma_{w^\sigma}=0$ for any $\varsigma\leqslant s$ and any $\sigma>s$. 
This means that $Q'$ does not generate a strictly $p$-th order potential symmetry of~$\lceil u\rfloor$.

Conversely, suppose that $Q'$ does not generate a strictly $p$-th order potential symmetry of~$\lceil u\rfloor$.
Therefore, there exists a basis $\{w^{p,1},\dots,w^{p,p}\}$ in~$\langle\psi^1,\dots,\psi^p\rangle$ such that 
the corresponding Lie invariance operator $Q=Q'+\tilde\zeta^{\varsigma-1}\p_{\tilde w^{\varsigma-1}}$ 
of system~\eqref{EqMinModifiedIteratedGenPotSysOfLPEs} satisfies, for some fixed value~$s$, the condition 
$\zeta^{\varsigma-1}_{w^\sigma}=0$ for any $\varsigma\leqslant s$ and any $\sigma>s-1$. 
In particular, 
\[
\zeta^{s-1}_{w^p}=
\frac{W(w^{p,s+1},\dots,w^{p,p},Q'[w^{p,p}])}{w^{p,p}\,W(w^{p,s+1},\dots,w^{p,p})}=0,
\]
i.e., $Q'[w^{p,p}]\in\langle w^{p,s+1},\dots,w^{p,p}\rangle$.
In view of Corollary~\ref{CorollaryOnPresentationOfPotsViaFixedSolutionsOfPotEqForLPEs}
the condition $\zeta^{\varsigma-1}_{w^\sigma}=0$ for any $\varsigma\leqslant s$ and any $\sigma>s-1$ is equivalent to 
$\zeta^{\varsigma-1}_{v^\sigma}=0$ for any $\varsigma\leqslant s$ and any $\sigma>s-1$. 
Thus the order in the tuple $(w^{p,s+1},\dots,w^{p,p})$ is inessential for the condition on the $\zeta$'s and, 
therefore, $Q'[w^{p,\sigma}]\in\langle w^{p,s+1},\dots,w^{p,p}\rangle$, $\sigma=s+1,\dots,p$.
It follows that the subspace $\langle w^{p,s+1},\dots,w^{p,p}\rangle$ of the space~$\langle\psi^1,\dots,\psi^p\rangle$ 
is invariant under the action of~$\widehat Q'$.
\end{proof}

\begin{corollary}\label{CorollaryOnCriterionOfStrictlyPthOrderPotSymsForLPE}
An essential Lie invariance operator $Q'$ of~$\lceil w^p\rfloor$
generates a strictly $p$-th order potential symmetry 
of~$\lceil u\rfloor={\rm DT}[\psi^1,\dots,\psi^p]\,\lceil w^p\rfloor$ iff 
any one- or two-dimensional subspace of~$\langle\psi^1,\dots,\psi^p\rangle$ is not invariant under the action 
of the associated first-order differential operator~$\widehat Q'$. 
\end{corollary}

\begin{proof}
In view of Theorem~\ref{TheoremOnCriterionOfStrictlyPthOrderPotSymsForLPE}, it is enough to prove 
that a finite-dimensional space~$\mathcal I$ invariant under the action of~$\widehat Q'$ 
has a one- or two-dimensional invariant subspace. Let $q=\dim\mathcal I$ and 
$\varphi^1$,~\dots,~$\varphi^q$ form a basis of~$\mathcal I$. 
Then $Q'[\varphi^\sigma]=\kappa_{\sigma\varsigma}\varphi^\varsigma$, where $\kappa_{\sigma\varsigma}$ are constants. 
Hereafter the indices $\sigma$ and $\varsigma$ run from $1$ to $q$.
If $(c_1,\dots,c_q)$ is an eigenvector of the matrix $(\kappa_{\sigma\varsigma})$ then 
$\langle c_\sigma\varphi^\sigma\rangle$ is a one-dimensional invariant subspace of~$\mathcal I$. 
In the real case, if the matrix $(\kappa_{\sigma\varsigma})$ has no real eigenvalues, 
we can take an eigenvector $(c_1,\dots,c_q)$ of the complexification. Then  
$\langle \varphi^\sigma\mathop{\rm Re}\nolimits c_\sigma, \varphi^\sigma\mathop{\rm Im}\nolimits c_\sigma\rangle$ 
is a two-dimensional invariant subspace of~$\mathcal I$ without proper invariant subspaces.
It is obvious that in the complex case it is enough to consider only one-dimensional subspaces. 
\end{proof}

\begin{note}
Similarly to Proposition~\ref{PropositionOnSymCondOfExistOfPotSymOfLPE}, 
Theorem~\ref{TheoremOnCriterionOfStrictlyPthOrderPotSymsForLPE} can be re-formulated in different terms. 
Thus, a subspace $\langle w^{p,1},\dots,w^{p,q}\rangle$ of the solution space of~$\lceil w^p\rfloor$
is invariant under the action of the operator~$\widehat Q'$ and 
$Q'[w^{p,\sigma}]=\kappa_{\sigma\varsigma}w^{p,\varsigma}$ 
iff $(w^{p,1},\dots,w^{p,q})$ is an invariant solution of the uncoupled system of $q$ copies of~$\lceil w^p\rfloor$ 
with respect to the operator \[\bar Q'=\tau\p_t+\xi\p_x+\zeta^{p1}w^{p,\sigma}\p_{w^{p,\sigma}}
-\kappa_{\sigma\varsigma}w^{p,\varsigma}\p_{w^{p,\sigma}}.\] 
Here $\kappa_{\sigma\varsigma}$ are constants.
The indices $\sigma$ and $\varsigma$ again run from $1$ to $q$.
\end{note}

The characterization of linear second-order parabolic equations possessing potential symmetry algebras 
of a fixed order is given by the following theorem.

\begin{theorem}\label{TheoremOnGenPotSymsOfLPEs}
A linear $(1+1)$-dimensional second-order parabolic equation admits a strictly $p$-th order potential symmetry algebra 
iff it is equivalent with respect to point equivalence transformations to an equation from class~\eqref{EqReducedLPE} 
in which 
\begin{equation}\label{EqFormOfVWithPotSyms}
V=P(x)-2\left(\frac{\bigl(W(\psi^1,\dots,\psi^p)\bigr)_x}{W(\psi^1,\dots,\psi^p)}\right)_x, 
\end{equation}
where $\psi^s=\psi^s(t,x)$ are linearly independent solutions of the equation $\psi_t-\psi_{xx}+P(x)\psi=0$ and 
either $P=\mu x^{-2}$, $\mu=\const$, 
or no subspaces of $\langle\psi^1,\dots,\psi^p\rangle$ are invariant under the action of~$\p_t$
if $P$ is inequivalent to~$\mu x^{-2}$ with respect to point equivalence transformations. 
The associated characteristic subspace is $\langle\alpha^1,\dots,\alpha^p\rangle$, where
\[
\alpha^\varsigma=(-1)^{\varsigma-1}
\frac{W(\psi^1,\dots,\lefteqn{\smash{\diagdown}}\psi^\varsigma,\dots,\psi^p)}{W(\psi^1,\dots,\psi^p)}, 
\quad \varsigma=1,\dots,p,
\]
In the case $P=0$ ($\sim\mu=0$) the potential symmetry algebra contains at least one (two) independent operators which 
are linearly independent up to Lie symmetries and essentially involve the $p$-th order potential. 
For the general value $P=P(x)$ a sufficient condition for
the potential symmetry algebra corresponding to $\langle\psi^1,\dots,\psi^p\rangle$ having strictly $p$-th order is 
the following. For an arbitrary choice of basis in $\langle\psi^1,\dots,\psi^p\rangle$ 
all the modified potential equations~$\lceil w^{s-1}\rfloor$ including $\lceil w^0\rfloor=\lceil u\rfloor$ 
have non-stationary functions as values of the arbitrary elements~$V^{s-1}$.
\end{theorem}

\begin{proof}
The equivalence relation generated by transformations from~$G^\sim_{\smash{[p]}}$ reduces
the $p$-th order potential frame over class~\eqref{EqGenLPE} to 
the $p$-th order potential frame over class~\eqref{EqReducedLPE} 
in which there is an analogous equivalence relation generated by transformations from~$\hat G^\sim_{\smash{[p]}}$.
In view of Lemma~\ref{LemmaOn1to1CorrespondenceBetweenLieSymsOfGenPotSysAndPthLevelPotEqOfLPE}, 
the investigation of $p$-th order potential symmetries can be replaced by the investigation of 
Lie symmetries of $p$-th order potential equations. 
The projection of the group $\hat G^\sim_{\smash{[p]}}$ to $(t,x,w^p,P)$, where $P:=V^p$, coincides with 
the equivalence group of the class of modified $p$-level potential equations for 
the equations of the form~\eqref{EqReducedLPE}. 
This class is a copy of the class~\eqref{EqReducedLPE} written in terms of $(t,x,w^p,P)$, i.e., 
both these classes have the same (up to the notation of variables, arbitrary elements and group parameters) 
equivalence group. 
Therefore, the Lie--Ovsiannikov classification described in Theorem~\ref{TheoremOnGroupClassificationOfLPEs} 
can be used here. 
Since the initial equations of form~\eqref{EqReducedLPE} and the corresponding modified $p$-level potential equations
are connected via the multiple Darboux transformation, the Crum theorem~\cite{Crum1955,Matveev&Salle1991} provides 
a connection between their arbitrary elements.  
To complete the proof, we have to examine all the inequivalent cases of extension of Lie symmetries 
from Theorem~\ref{TheoremOnGroupClassificationOfLPEs}, re-written in terms of $(t,x,w^p,P)$.
The general case without extension is eliminated in view of item~1 of 
Theorem~\ref{TheoremOnCriterionOfStrictlyPthOrderPotSymsForLPE}.

Let $P=\mu x^{-2}$, i.e., the modified $p$-level potential equation has the form $w^p_t\!-w^p_{xx}\!+\mu x^{-2}w^p\!=\!0.$
We prove by contradiction that there are no finite dimensional subspaces in the solution space 
of~$\lceil w^p\rfloor$ which are simultaneously invariant under the action of the differential operators 
$\widehat{\p_t}$, $\widehat D$ and $\widehat\Pi$ corresponding to $\p_t$, $D$ and $\Pi$. 
Suppose that such a subspace~$\mathcal I$ exists. 
Let $q=\dim\mathcal I$ and $\varphi^1$,~\dots,~$\varphi^q$ form a basis of~$\mathcal I$. 
Then 
\begin{gather*}
\p_t[\varphi^\sigma]=-\varphi^\sigma_t=-\kappa^1_{\sigma\varsigma}\varphi^\varsigma, 
\\[1ex]
D[\varphi^\sigma]=-2t\varphi^\sigma_t-x\varphi^\sigma_x=-\kappa^2_{\sigma\varsigma}\varphi^\varsigma, 
\\[1ex]
\Pi[\varphi^\sigma]=-4t^2\varphi^\sigma_t-4tx\varphi^\sigma_x-(x^2+t)\varphi^\sigma=
-\kappa^3_{\sigma\varsigma}\varphi^\varsigma, 
\end{gather*}
where $\kappa^1_{\sigma\varsigma}$, $\kappa^2_{\sigma\varsigma}$ and $\kappa^3_{\sigma\varsigma}$ are constants. 
In what follows the indices $\sigma$ and $\varsigma$ run from $1$ to $q$. 
The above equality implies that $(K^{\sigma\varsigma}-x^2\delta_{\sigma\varsigma})\varphi^\varsigma=0$, 
where $\delta_{\sigma\varsigma}$ is the Kronecker delta and 
\[K^{\sigma\varsigma}=4t^2\kappa^1_{\sigma\varsigma}-t(4\kappa^2_{\sigma\varsigma}+\delta_{\sigma\varsigma})+
\kappa^3_{\sigma\varsigma}.\] 
For any fixed~$t$, the determinant of the matrix $(K^{\sigma\varsigma}-x^2\delta_{\sigma\varsigma})$ 
vanishes only for a finite number of values of~$x$. Therefore, for any fixed~$t$ 
$\varphi^\varsigma(t,x)$ does not vanish for at most a finite number of values of~$x$, i.e., 
by continuity all $\varphi^\varsigma$ are equal to 0 identically. 
This implies a contradiction. 
 
In the case $\mu=0$ we additionally have the operators~$\p_x$ and~$G$. 
They also possesses no simultaneously invariant finite dimensional subspaces in the solution space 
of~$\lceil w^p\rfloor$ since otherwise by an analogous reasoning we have 
\begin{gather*}
\p_x[\varphi^\sigma]=-\varphi^\sigma_x=-\kappa^1_{\sigma\varsigma}\varphi^\varsigma, 
\\[1ex]
G[\varphi^\sigma]=-2t\varphi^\sigma_x-x\varphi^\sigma=-\kappa^2_{\sigma\varsigma}\varphi^\varsigma,
\end{gather*}
where $\kappa^1_{\sigma\varsigma}$ and $\kappa^2_{\sigma\varsigma}$ are constants. 
Then $(\kappa^2_{\sigma\varsigma}-\kappa^1_{\sigma\varsigma}t-x\delta_{\sigma\varsigma})\varphi^\varsigma=0$ 
and $\varphi^\varsigma=0$ identically, which implies the same contradiction. 

There is only one independent nontrivial Lie symmetry operator~$\p_t$ for the general value $P=P(x)$. 
If it does not induce a strictly $p$-th order potential symmetry operator of~\eqref{EqReducedLPE} then 
it has an invariant subspace~$\mathcal I$ in~$\langle\psi^1,\dots,\psi^p\rangle$. 
Let $q=\dim\mathcal I$ and $\varphi^1$,~\dots,~$\varphi^q$ form a basis of~$\mathcal I$. 
Then $\varphi^\sigma_t=\kappa_{\sigma\varsigma}\varphi^\varsigma$ for some constants~$\kappa_{\sigma\varsigma}$ 
and, therefore, for the modified $(p-q)$-level potential equation obtained from~$\lceil w^p\rfloor$ 
with ${\rm DT}[\varphi^1,\dots,\varphi^q]$ we have
\[
V^{p-q}_t=P_t-2\left(\frac{\bigl(W(\varphi^1,\dots,\varphi^q)\bigr)_x}{W(\varphi^1,\dots,\varphi^q)}\right)_{xt}=
-2(\kappa_{\sigma\sigma}-\kappa_{\sigma\sigma})
\left(\frac{\bigl(W(\varphi^1,\dots,\varphi^q)\bigr)_x}{W(\varphi^1,\dots,\varphi^q)}\right)_x=0,
\]
i.e., $V^{p-q}=V^{p-q}(x)$. As a result, the problem for the $p$-order potential frame is completely reduced 
to the same problem for the $(p-q)$-order potential frame.
\end{proof}

Based on Theorem~\ref{TheoremOnGenPotSymsOfLPEs}, 
we can formulate pure symmetry criteria on the existence of potential symmetries of arbitrary order 
without involving equivalence transformations. 

\begin{corollary}\label{CorollaryOnSymCriteriaOfExestanceOfPotSymsForLPEs}
A linear second-order parabolic equation admits the strictly $p$-th order potential symmetry algebra
associated with its $p$-dimensional characteristic subspace $\langle\alpha^1,\dots,\alpha^p\rangle$ only if 
the corresponding $p$-level potential equation possesses nontrivial Lie symmetry operators. 
If the $p$-level potential equation has more than one independent nontrivial Lie symmetry operators then 
the potential symmetry algebra is of strictly $p$-th order. 
More precisely, if the $p$-level potential equation has more than one (three) 
independent nontrivial Lie symmetry operators then 
for any choice of basis in $\langle\alpha^1,\dots,\alpha^p\rangle$
the potential symmetry algebra contains at least one (two) independent operators which 
essentially involve the $p$-th order potential.
\end{corollary}

\section{On number and order of potential symmetries}\label{SectionOnNumberAndOrderOfPotSyms}

The title of the section is slightly vague. 
Our aim here is to investigate the following basic questions about potential symmetries 
of linear parabolic equations: 
Can a fixed equation from class~\eqref{EqGenLPE} possess an infinite series of potential symmetry algebras? 
How can equations having potential symmetry algebras of all orders be constructed?
What orders of potential symmetries are possible for a fixed equation?
When are the orders of potential symmetries bounded?
We study only certain examples that, nevertheless, allow us to formulate quite general statements 
answering some of these questions. 

It is easy to construct equations having potential symmetry algebras of all orders. 
Indeed, we take the linear heat equation $w_t=w_{xx}$ as a potential equation (i.e., $P=0$) 
and choose, for any $p\in\mathbb N$, the $p$-tuple $\bar\psi=(P_0,\dots,P_{p-1})$ of its solutions. 
Here $P_k$ is the canonical heat polynomial of degree~$k$, 
\begin{gather*}
P_{2m}(t,x)=\frac{x^{2m}}{(2m)!}+\frac{t}{1!} \frac{x^{2m-2}}{(2m-2)!}+
\frac{t^2}{2!} \frac{x^{2m-4}}{(2m-4)!}+\cdots +
\frac{t^{m-1}}{(m-1)!}\frac{x^2}{2!}+\frac{t^m}{m!},
\\
P_{2m+1}(t,x)=\frac{x^{2m+1}}{(2m+1)!}+\frac{t}{1!}\frac{x^{2m-1}}{(2m-1)!}
+\frac{t^2}{2!}\frac{x^{2m-3}}{(2m-3)!}+ \cdots +
\frac{t^{m-1}}{(m-1)!}\frac{x^3}{3!}+\frac{t^m}{m!}\frac{x}{1!}, 
\end{gather*}
$k,m\in\mathbb N\cup\{0\}$. Note that $\p P_{k+1}/\p x=P_k$ hence $\p^kP_k/\p x^k=1$, $\p^kP_{k+1}/\p x^k=x$. 
A~direct calculation implies that $W(P_0,\dots,P_{p-1})=1$ 
and that the corresponding value of the arbitrary element~$V$ equals 0 for any $p\in\mathbb N$.
Therefore, ${\rm DT}[P_0,\dots,P_{p-1}]\,\lceil w\rfloor=\lceil u\rfloor$, 
where $\lceil u\rfloor$ also is the linear heat equation $u_t=u_{xx}$. 
In fact the multiple Darboux transformation ${\rm DT}[P_0,\dots,P_{p-1}]$ is nothing but $p$-order 
differentiation with respect to~$x$: $u={\rm DT}[P_0,\dots,P_{p-1}]\,w=\p^pw/\p x^p$. 
Let $(\alpha^{p1},\dots,\alpha^{pp})$ be the characteristic tuple of~$\lceil u\rfloor$, 
dual to the solution tuple $(P_0,\dots,P_{p-1})$ of~$\lceil w\rfloor$, i.e.,
\[\alpha^{ps}=(-1)^{s-1} W(P_1,\dots,P_{p-s}),\quad s=1,\dots,p-1,\quad \alpha^{pp}=(-1)^{p-1}.\] 
(See Corollary~\ref{CorollaryOnDualMultipleDarbouxTrans}.) 
The Wronskians $W^q=W(P_1,\dots,P_q)$, $q\in\mathbb N$, and $W^0:=1$ are solutions of the backward heat equation and 
additionally satisfy the conditions \[\p W^q/\p x=W^{q-1},\quad W^q(0,0)=0.\] Therefore, 
$W^q=P_q(-t,x)$ is the backward heat polynomial of order~$q$
and \[\alpha^{ps}=(-1)^{s-1}P_{p-s}(-t,x),\quad s=1,\dots,p.\]
In view of Theorems~\ref{TheoremOnDualMultipleDarbouxTrans} and~\ref{TheoremOnGenPotSymsOfLPEs},
for any $p\in\mathbb N$ the potential symmetry algebra~$\mathfrak g_p$ of~$\lceil u\rfloor$, associated with 
the $p$-dimensional characteristic subspace $\langle\alpha^{p1},\dots,\alpha^{pp}\rangle$, 
is of strictly $p$-th order. 
For any choice of basis in $\langle\alpha^{p1},\dots,\alpha^{pp}\rangle$
the potential symmetry algebra contains at least two independent operators which 
essentially involve the $p$-th order potential. 
Summarizing these results, we can formulate the following statement. 

\begin{proposition}\label{PropositionOnInfiniteSeriesOfPotSymAlgsForLHE}
The linear heat equation admits an infinite series $\{\mathfrak g_p,\, p\in\mathbb N\}$ 
of potential symmetry algebras. 
For any $p\in\mathbb N$ the algebra $\mathfrak g_p$ is of strictly $p$-th potential order and is 
associated with $p$-tuples of the linearly independent lowest order polynomial 
solutions of the backward heat equation. 
Moreover, it is the standard $p$-th prolongation, with respect to only~$x$, 
of the Lie invariance algebra~$\mathfrak g_0$, re-written in terms of $(t,x,w)$ and, 
hence, is isomorphic to~$\mathfrak g_0$. 
\end{proposition}

Other examples can be constructed in a similar way. 
Thus, for the same potential equation and the $p$-tuple $\bar\psi=(P_0,P_1,\dots,P_{p-2},P_p)$ 
we have $W(P_0,P_1,\dots,P_{p-2},P_p)=x$ 
and the corresponding value of the arbitrary element~$V$ is equal to $2x^{-2}$ for any $p\in\mathbb N$.
Therefore, ${\rm DT}[P_0,P_1,\dots,P_{p-2},P_p]\,\lceil w\rfloor=\lceil u\rfloor$, 
where $\lceil u\rfloor$ denotes the equation $u_t-u_{xx}+2x^{-2}u=0$.
Let $(\alpha^{p1},\dots,\alpha^{pp})$ be the characteristic tuple of~$\lceil u\rfloor$ 
dual to the solution tuple $(P_0,P_1,\dots,P_{p-2},P_p)$ of~$\lceil w\rfloor$, i.e.,
\[
\alpha^{ps}=(-1)^{s-1}x^{-1}W(P_1,\dots,P_{p-1-s},P_{p-s+1}), \quad s=1,\dots,p-1,\quad 
\alpha^{pp}=(-1)^{p-1}x^{-1}. 
\]
The Wronskians $\tilde W^q=W(P_1,\dots,P_{q-2},P_q)$, $q\geqslant2$, 
satisfy the conditions $\p\tilde W^q/\p x=xW^{q-2}$, $\tilde W^q(0,0)=0$. 
The quotients $\tilde W^q/x$ are solutions of the adjoint equation \mbox{$v_t+v_{xx}-2x^{-2}v=0$}.
Therefore, $\tilde W^q=xW^q_x-W^q$ and \[\alpha^{ps}=(-1)^{s-1}(W^{p-s}-x^{-1}W^{p-s+1}),\quad s=1,\dots,p-1.\]
Here $W^q$ denotes the backward heat polynomial $P_q(-t,x)$ of order~$q$.
Analogously to the previous example,
for any $p\in\mathbb N$ the potential symmetry algebra~$\tilde{\mathfrak g}_p$ of~$\lceil u\rfloor$ associated with 
the $p$-dimensional characteristic subspace $\langle\alpha^{p1},\dots,\alpha^{pp}\rangle$, 
is of strictly $p$-th order. 
For any choice of basis in $\langle\alpha^{p1},\dots,\alpha^{pp}\rangle$
the potential symmetry algebra contains at least two independent operators which 
essentially involve the $p$-th order potential. 
As a result, the equation~$\lceil u\rfloor$ has an infinite series $\{\tilde{\mathfrak g}_p,\, p\in\mathbb N\}$ 
of potential symmetry algebras isomorphic to the Lie invariance algebra~$\mathfrak g_0$ of the linear heat equation. 

The described construction can be generalized. 
Let $\lceil u\rfloor$ be the equation $u_t-u_{xx}+Vu=0$, 
where the function~$V$ has the form~\eqref{EqFormOfVWithPotSyms} with $P=0$, i.e., 
$\psi^s=\psi^s(t,x)$ are linearly independent solutions of the linear heat equation $\psi_t=\psi_{xx}$. 
This means that $\lceil u\rfloor$ is the image of the linear heat equation $\lceil w\rfloor$ 
under the Darboux transformation ${\rm DT}[\psi^1,\dots,\psi^p]$.
Then for any $q\in\mathbb N$
\[
\lceil u\rfloor={\rm DT}[P_0,\dots,P_{q-1},\tilde\psi^{q1},\dots,\tilde\psi^{qp}]\,\lceil w\rfloor,
\]
where $\tilde\psi^{qs}$ for any fixed $s$ is a solution of the linear heat equation, being a preimage of $\psi^s$ 
under the Darboux transformation ${\rm DT}[P_0,\dots,P_{q-1}]$, i.e., ${\rm DT}[P_0,\dots,P_{q-1}]\tilde\psi^{qs}=\psi^s$.
The function~$\tilde\psi^{qs}$ is found by $q$-fold integration of the function~$\psi^s$ 
with respect to~$x$ with a special choice of an `integration constant' depending on~$t$. 
The best way is to employ of the recursive formula 
$\tilde\psi^{qs}_x=\tilde\psi^{q-1,s}$, $\tilde\psi^{qs}_t=\tilde\psi^{q-1,s}_x$, $\tilde\psi^{0s}:=\psi^s$. 
In view of Proposition~\ref{PropositionOnInfiniteSeriesOfPotSymAlgsForLHE}, this implies the following statement.

\begin{corollary}\label{CorollaryOnInfiniteSeriesOfPotSymAlgsForEqsEquivToLHE}
Suppose that a linear $(1+1)$-dimensional second-order parabolic equation is equivalent 
with respect to point equivalence transformations to an equation from class~\eqref{EqReducedLPE} 
in which 
\[
V=-2\left(\frac{\bigl(W(\psi^1,\dots,\psi^p)\bigr)_x}{W(\psi^1,\dots,\psi^p)}\right)_x, 
\]
where $\psi^s=\psi^s(t,x)$ are linearly independent solutions of the linear heat equation $\psi_t=\psi_{xx}$. 
Then this equation possesses an infinite series $\{\tilde{\mathfrak g}_{p+k},\, k\in\mathbb N\cup\{0\}\}$ 
of potential symmetry algebras. 
Each algebra $\tilde{\mathfrak g}_{p+k}$ from the series is of strictly $(p+k)$-th potential order, 
is isomorphic to the Lie invariance algebra~$\mathfrak g_0$ of the linear heat equation 
and contains at least two operators essentially involving potentials of the $(p+k)$-th level.
\end{corollary}

It is easy to see that a basis of the above consideration is formed by the construction of 
the infinite series of the auto-Darboux transformations $\{{\rm DT}[P_0,\dots,P_{p-1}],p\in\mathbb N\}$
for the linear heat equation. 
Similar series of auto-Darboux transformations exist for all the equations of the form 
\begin{equation}\label{EqMuX-2LPE}
u_t-u_{xx}+\mu x^{-2}u=0.
\end{equation}
To show this, it is sufficient to prove that equation~\eqref{EqMuX-2LPE} has an infinite series of 
linearly independent solutions such that 
the Wronskian of $p$ first solutions from the series is constant for an infinite set of~$p$'s.
We restrict ourself to solutions which are polynomials in~$t$. 
Let $\varphi^{0i}$, $i=1,2$, be linearly independent stationary solutions of~\eqref{EqMuX-2LPE}, i.e., 
\[\varphi^{0i}_t=0,\quad\varphi^{0i}_{xx}=\mu x^{-2}\varphi^{0i}.\] 
The following values of~$\varphi^{0i}$ can be taken:
\begin{gather*}
\varphi^{01}=|x|^{\nu_-}, \quad \varphi^{02}=|x|^{\nu_+}, \quad \nu_\pm=\frac{1\pm\sqrt{1+4\mu}}2, 
\quad\mbox{if}\quad 1+4\mu>0,\\[.3ex] 
\varphi^{01}=|x|^{1/2}, \quad \varphi^{02}=|x|^{1/2}\ln|x| \quad\mbox{if}\quad 1+4\mu=0,\\[1ex] 
\varphi^{01}=|x|^{1/2}\cos\kappa\ln|x|, \quad \varphi^{02}=|x|^{1/2}\sin\kappa\ln|x|, \quad \kappa=\sqrt{-1/4-\mu},
\quad\mbox{if}\quad 1+4\mu<0.
\end{gather*}
Consider the functions $\varphi^{ki}=\widehat\Pi^k\varphi^{0i}$, where $\widehat\Pi=-4t^2\p_t-4tx\p_x-x^2-2t$. 
They are linearly independent and polynomial in~$t$. 
They are solutions of~\eqref{EqMuX-2LPE} as a result of the action of the symmetry operator~$\Pi$ on solutions. 
The Wronskian $W^k=W(\varphi^{01},\varphi^{02},\dots,\varphi^{k1},\varphi^{k2})$ does not depend on~$t$ 
since $\varphi^{ki}_t$ necessarily is a linear combination of the functions $\varphi^{k'\!i'\!}$, $k'<k$. 
Therefore, it is enough to evaluate $W^k$ for a single value of~$t$. 
Putting~$t=0$, we prove by induction that $W^k_x=0$, i.e., $W^k$ is a nonzero constant. 
As a result, ${\rm DT}[\varphi^{01},\varphi^{02},\dots,\varphi^{k1},\varphi^{k2}]$ is an auto-Darboux transformation 
of equation~\eqref{EqMuX-2LPE} for any $k\in\mathbb N\cup\{0\}$. 

\begin{proposition}\label{PropositionOnInfiniteSeriesOfPotSymAlgsForMuX-2LPE}
Equation~\eqref{EqMuX-2LPE} admits an infinite series $\{\hat{\mathfrak g}_{2q},\, q\in\mathbb N\}$ 
of potential symmetry algebras. 
For any $q\in\mathbb N$ the algebra $\hat{\mathfrak g}_{2q}$ is of strictly $2q$-th potential order and is 
associated with $2q$-tuples of the linearly independent solutions of~\eqref{EqMuX-2LPE} 
which are lowest order polynomials in~$t$. 
Moreover, it is isomorphic to the Lie invariance algebra~$\hat{\mathfrak g}_0$ of equation~\eqref{EqMuX-2LPE}
and contains at least one operator essentially involving potentials of the $2q$-th level.
\end{proposition}

Combining Darboux, auto-Darboux and equivalence transformations similarly to the proof 
of Corollary~\ref{CorollaryOnInfiniteSeriesOfPotSymAlgsForEqsEquivToLHE}, we derive the following corollary 
of Proposition~\ref{PropositionOnInfiniteSeriesOfPotSymAlgsForMuX-2LPE}.

\begin{corollary}\label{CorollaryOnInfiniteSeriesOfPotSymAlgsForEqsEquivToMuX-2LPE}
Suppose that a linear $(1+1)$-dimensional second-order parabolic equation is equivalent 
with respect to point equivalence transformations to an equation from class~\eqref{EqReducedLPE} 
in which 
\[
V=\frac\mu{x^2}-2\left(\frac{\bigl(W(\psi^1,\dots,\psi^p)\bigr)_x}{W(\psi^1,\dots,\psi^p)}\right)_x, 
\]
where $\psi^s=\psi^s(t,x)$ are linearly independent solutions of the equation $\psi_t-\psi_{xx}+\mu x^{-2}\psi=0$. 
Then this equation possesses an infinite series $\{\check{\mathfrak g}_{p+2k},\, k\in\mathbb N\cup\{0\}\}$ 
of potential symmetry algebras. 
Each algebra $\check{\mathfrak g}_{p+2k}$ from the series is of strictly $(p+2k)$-th potential order, 
is isomorphic to the Lie invariance algebra~$\hat{\mathfrak g}_0$ of equation~\eqref{EqMuX-2LPE} 
and contains at least one operator essentially involving potentials of the $(p+2k)$-th level.
\end{corollary}

\begin{note}
An equation from class~\eqref{EqReducedLPE} with a stationary value of the arbitrary element~$V$ 
has auto-Darboux transformations constructed with solutions polynomial in~$t$. 
In contrast to the special case $V=\mu x^2$, this does not imply conclusions in the general case $V=V(x)$ 
since, in particular, Theorem~\ref{TheoremOnGenPotSymsOfLPEs} does not give a sufficiently powerful criterion on the 
existence of potential symmetries for such a situation. 
\end{note}

\section{Discussion}\label{SectionDiscussion}

In the present paper we investigate symmetries and conservation laws of 
linear $(1+1)$-dimensional second-order parabolic equations.
In our opinion, the most important results of the paper are the following.
\begin{itemize}
\item
It is proved that any potential conserved vector of a linear parabolic equation is equivalent to a local one 
(Theorem~\ref{TheoremOnPotConsLawsOfLPEs}). 
The local conservation laws are described in Theorem~\ref{TheoremLocalCLsLPEs}. 
Namely, the space of characteristics associated with the local conservation laws 
of a linear parabolic equation can be identified with the space of the functions depending only on~$t$ and~$x$ 
which are solutions of the adjoint equation.
\item
Different criteria on the existence of potential symmetries of arbitrary order are formulated. 
Theorem~\ref{TheoremOnCriterionOfStrictlyPthOrderPotSymsForLPE} gives a criterion 
on operators of potential equations, generating potential symmetries with the same orders
as the level numbers of the potential equations. The shape of the equations having potential symmetries 
of a strongly fixed order is characterized in Theorem~\ref{TheoremOnGenPotSymsOfLPEs}. 
Corollary~\ref{CorollaryOnSymCriteriaOfExestanceOfPotSymsForLPEs} supplies a pure symmetry criterion 
in terms of the number of independent nontrivial Lie symmetries of potential equations. 
\end{itemize} 

\looseness=1
Extensive preparatory considerations were needed for obtaining these results. 
Thus, the group classification of class~\eqref{EqGenLPE} has a twofold application for the investigation 
of potential symmetries in this class. Firstly, an exhaustive knowledge on Lie symmetries is necessary 
for focusing our attention on pure potential symmetries. 
At the same time, the modified potential equations of any fixed level for the equations from class~\eqref{EqGenLPE} 
form a copy of this class. 
Therefore, results on the group classification in class~\eqref{EqGenLPE} also are directly used in the description 
of potential symmetries. This entails the need for a careful revision of these results. 
The study of normalization properties of class~\eqref{EqGenLPE}, their subclasses and the associated classes of 
inhomogeneous equations justifies the choice of the gauge $A=1$, $B=0$ under both the group classification 
of single equations and the analysis of the whole potential frame over the class~\eqref{EqGenLPE}. 
It also gives a well-founded explanation of the difficulties arising in the classification 
of the Kolmogorov and Fokker--Plank equations. 

The investigation of local conservation laws of equations from class~\eqref{EqGenLPE} naturally leads to 
the consideration of systems of two mutually adjoint equations from class~\eqref{EqGenLPE}. 
Such systems form the basis of the potential frame over class~\eqref{EqGenLPE} 
and are directly connected with the so-called adjoint variational principle. 
The framework of the adjoint variational principle is extended to admissible transformations. 
A number of auxiliary statements on a hierarchy of normalized classes of second-order evolution systems are proved. 
The group classification of Fokker--Plank equations is obtained from the group classification of the Kolmogorov equations 
in a simple way based on the application of the adjoint variational principle.

\looseness=1
A feeling for the general problem, understanding possible ways of solving it and a rough shape of formulating the final results
arise during the consideration of simplest potential symmetries. 
The techniques resulting from this approach involve different ideas. 
Thus, the potential systems should be simultaneously studied with the associated initial, adjoint, potential, 
modified potential and adjoint modified potential equations which together form the potential frame 
over the class~\eqref{EqGenLPE} (of the first order in the case of simplest potential symmetries).  
For applications of this idea to general potential symmetries,
an analogue of the potential equation, corresponding to a tuple of characteristics, should be proposed at first. 
The equivalence group of the initial class~\eqref{EqGenLPE} is prolonged to the whole potential frame including 
characteristics of the initial equations, potentials, modified potentials 
and characteristics of the modified potential equations. 
It is shown that the problem can be investigated up to the equivalence relation generated by the 
prolonged equivalence group.
Moreover, the classification of Lie symmetries of potential systems with respect to the above equivalence relation 
is reduced to the same classification for modified potential equations. 
This allow us to use the Lie--Ovsiannikov classification. 
A connection between different objects of the potential frame is provided via 
the dual Darboux transformation which, for this reason, is an important component of the proposed technique. 

The study of general potential symmetries requires a development of the above ideas in relation to simplest 
potential symmetries and the creation of specific tools. 
It appears that the $p$-order and $p$-level potential systems 
associated with the same tuple of characteristics should be investigated simultaneously.  
(This conclusion cannot be drawn in the consideration of simplest potential symmetries since the
corresponding 1-order and 1-level potential systems coincide.) 
The potential equations associated with characteristic tuples are introduced in a natural way due to 
the special iterative procedure for the construction of potential systems with levels higher than~1. 
Moreover, the modified potential equations are in fact associated with characteristic subspaces 
spanned by the corresponding characteristic tuple. This allows us to study potential symmetries up to 
forming linear combinations of elements of a characteristic tuple. 
Enhancing results on the multiple dual Darboux transformation between equations from the class~\eqref{EqGenLPE}, 
presented in Theorem~\ref{TheoremOnDualMultipleDarbouxTrans}, is of fundamental importance for both the consolidation 
of the whole potential frame and for deriving definitive statements on potential symmetries. 
The proof of Lemma~\ref{LemmaOn1to1CorrespondenceBetweenLieSymsOfGenPotSysAndPthLevelPotEqOfLPE}
on Lie symmetries of potential systems is quite intricate and involves a number of tricks. 
It is essentially based on the higher-level representation of potential systems. 
The next important step is the prolongation of the equivalence group of the class~\eqref{EqGenLPE} to 
the whole potential frame in Lemma~\ref{LemmaOnEquivTransOfGenPotFrameOfLPEs}. 
The confluence of these three components (the known shape of the Lie symmetries of potential systems, 
the multiple dual Darboux transformation and the prolonged equivalence transformations) 
results in criteria on the existence of potential symmetries, formulated in 
Theorems~\ref{TheoremOnCriterionOfStrictlyPthOrderPotSymsForLPE} and~\ref{TheoremOnGenPotSymsOfLPEs} 
and Corollary~\ref{CorollaryOnSymCriteriaOfExestanceOfPotSymsForLPEs}, 
as well as in subsequent estimations of the number of potential symmetries for subclasses of the class~\eqref{EqGenLPE}.

Another new problem on potential symmetries, which was first posed in general and solved for 
the $(1+1)$-dimensional linear parabolic equations in the present paper, is to justify
the choice of natural representatives among 
the equivalent conserved vectors for introducing potentials. 
Different choices of representatives are equivalent only with respect to generalized potential symmetries of arbitrarily high order. 
The construction of the natural potential frame for the class~\eqref{EqGenLPE} was justified by 
Corollary~\ref{CorollaryOnChiceOfConservedVertorsForConstructionOfPotSymsOfLPE}.  
To study (usual) potential symmetries of equations from the class~\eqref{EqGenLPE}, 
it is in fact sufficient to only consider conserved vectors of the canonical form~\eqref{eqCVofLPEs}, 
which have the lowest orders in both the flux and the density. 
Any tuple of conserved vectors of any equation from the class~\eqref{EqGenLPE}, 
containing a conserved vector of a higher than lowest order, gives only symmetries trivial in all senses.  
In contrast to the case of simplest potential symmetries, the proof of the above statements for an arbitrary number of potentials 
needed a description of the generalized potential symmetries the equation under consideration. 
As shown in Lemma~\ref{LemmaOn1to1CorrespondenceBetweenGeneralizedSymsOfGenPotSysAndPthLevelPotEqOfLPE}, 
each generalized symmetry of every potential system associated with conserved vectors in canonical form is 
linear in the dependent variables and their derivatives up to standard equivalence of generalized symmetries. 
Moreover, there is a one-to-one correspondence between generalized symmetries of a potential system and 
those of the associated potential equation.

\looseness=1
In spite of the wide range of the performed investigations, a number of interesting and difficult problems 
concerning potential symmetries of linear $(1+1)$-dimensional second-order parabolic equations remains unsolved. 
We list some of them, without making any claim on completeness. 

\begin{problem}\label{ProblemOnNumberInequivCharsForSimplestPotSymsOfLPEs}
What is the maximal number of inequivalent characteristics leading to simplest pure potential symmetries for 
a fixed equation from class~\eqref{EqGenLPE}? Classify such characteristics. 
A similar question can be asked about $p$-tuples of linearly independent characteristics giving 
strictly $p$-order potential symmetries.
\end{problem}

In Section~\ref{SectionOnSimpestPotSymsOfLHE}
Problem~\ref{ProblemOnNumberInequivCharsForSimplestPotSymsOfLPEs} 
is solved only for simplest potential symmetries of the linear heat equation. 
The answer to the question about the number of characteristics in this case is two. 
More precisely, any appropriate characteristic is equivalent to~1 or~$x$ with respect to the essential part of 
the point symmetry group of the linear heat equation. 

\begin{problem}\label{ProblemOnOrderOfPotSymsOfLPEs}
Given a fixed equation from class~\eqref{EqGenLPE}, are strict orders of its potential symmetries bounded 
and if so, what is their maximal value, or does it possess potential symmetries of arbitrarily large orders? 
In other words, does there exist an integer~$p$ such that the equation has strictly $p$-order potential symmetries and 
all potential symmetries of orders higher than~$p$ are reduced to potential symmetries of lesser orders?
\end{problem}

Wide classes of linear parabolic equations possessing infinite series of potential symmetry algebras, 
whose sequences of potential orders are unbounded, are constructed in Section~\ref{SectionOnNumberAndOrderOfPotSyms}
making use of auto-Darboux transformations.
The examples given are connected with the cases $V=0$ and $V=\mu x^{-2}$ of the Lie--Ovsiannikov classification. 
The problem is to study the reduced potential equations with $V=V(x)$. 
An obstacle which should be surmounted is the absence of a powerful criterion for the existence of potential symmetries 
for this case.

\begin{problem}\label{ProblemOnCriterionOfExistenceOfPotSymsOfLPEsWithGenVX}
Propose a necessary and sufficient criterion on the existence of potential symmetries for the cases 
when a potential equation possesses a single linearly independent nontrivial Lie symmetry operator. 
\end{problem}

Theorem~\ref{TheoremOnGenPotSymsOfLPEs} gives only a sufficient criterion on the existence of potential symmetries 
in such cases, in contrast to the cases with a higher number of nontrivial symmetry operators.

The criteria proposed in Section~\ref{SectionOnGenPotSysForLPEs} 
establish, in fact, a connection between the existence of potential symmetries 
of an equation from class~\eqref{EqGenLPE} and the reducibility of the equation to a special form. 
The investigation of point equivalence of linear parabolic equation was stimulated by the celebrated paper 
of Kolmogorov~\cite{Kolmogorov1938}. He posed the problem of describing Kolmogorov equations ($C=0$) which are
reduced to the heat equation by point transformations of a special form. 
This problem was completely solved in~\cite{Cherkasov1957}. 
A symmetry criterion on the reducibility naturally arises in the framework of 
the Lie--Ovsiannikov group classification~\cite{Lie1881,Ovsiannikov1982}. 
It was additionally discussed in a number of papers. See, e.g., 
\cite{Bluman1980,Shtelen&Stogny1989,Spichak&Stognii1999a,Spichak&Stognii1999b,Spichak&Stognii1999c}. 
Nevertheless, the symmetry criterion in its present forms is not as constructive as Cherkasov's. 

Constructive criteria on the reducibility can be obtained via the calculation of 
differential invariants and semi-invariants of the equivalence group used. 
Second-order differential semi-invariants of the linear transformations of the dependent variable 
in class~\eqref{EqGenLPE} were calculated in \cite{Ibragimov2002} by the infinitesimal method.
The same method was used in \cite{Johnpilla&Mahomed2001} for finding 
a necessary and sufficient invariant condition on the coefficients of the equations from class~\eqref{EqGenLPE}, 
which reduced to the linear heat equation by point transformations. 
The more effective approach to calculations concerning invariants and equivalence problems 
in classes of differential equations is given by Cartan's method of moving frames~\cite{Olver1995} 
in its Fels--Olver version~\cite{Fels&Olver1998,Fels&Olver1999}. 
Within the framework of the method of moving frames, the equivalence problem for equations from class~\eqref{EqGenLPE} 
is neatly investigated in \cite{Morozov2003}.

In \cite{Bluman&Shtelen2004} the generalized Kolmogorov problem on reducibility by combining 
equivalence and Darboux transformations was posed. 
Note that in fact the term `Darboux transformation' was not used in this paper, as well as 
the Crum representation~\cite{Crum1955,Matveev&Salle1991} of multiple Darboux transformations 
between linear parabolic equations. Only the iterative procedure in terms of a sequence of potential systems 
was presented. 
A simple application of the Crum theorem does not give an exhaustive solution of the generalized Kolmogorov problem. 
Additional tools should be used to create a constructive criterion for the generalized Kolmogorov problem, 
similar to the Cherkasov criterion for the classical Kolmogorov problem 
and the criterion in terms of differential invariants of the equivalence group.
 
\begin{problem}\label{ProblemOnCriterionOfReducibilityOfLPEsViaEquivAndDarbouxTrans}
Does there exist a constructive criterion on equations from class~\eqref{EqGenLPE} to be 
connected via compositions of point equivalence and Darboux transformations? 
In particular, is it possible to formulate explicit conditions on the arbitrary elements $A$, $B$ and~$C$ 
under which equation~\eqref{EqGenLPE} is reduced by a composition of point equivalence and Darboux transformations 
to the linear heat equation (or to the equation of the form~\eqref{EqReducedLPE} with $V=\mu x^{-2}$ or general $V=V(x)\,$). 
\end{problem}

Potential symmetries can be used to construct new exact solutions of equations from class~\eqref{EqGenLPE}, 
especially equations possessing no nontrivial Lie symmetries. 
An obvious way for this is Lie reduction of associated potential systems. 
An alternative but equivalent possibility is based on Theorem~\ref{TheoremOnCriterionOfStrictlyPthOrderPotSymsForLPE} 
and other statements of Section~\ref{SectionOnGeneralPotSymsOfLPEs}. 
Namely, we can at first find exact solutions of the corresponding potential equations and 
then map them to exact solutions of the initial equations by appropriate Darboux transformations. 
In fact, wide families of exact solutions are already known for the linear parabolic equations admitting nontrivial Lie symmetries.
That is why the second way is preferable. 

The results presented in this paper can be extended to other subjects. 
For example, they can be applied to the investigation of potential symmetries of nonlinear equations 
which are linearized to equations from class~\eqref{EqGenLPE} (the Burgers equations, $u^{-2}$-diffusion equation, etc.).
Since the field (real or complex numbers) in which the dependent and independent variables take 
values does not have an appreciable influence 
on our investigations, most of the obtained results are easily extended to $(1+1)$-dimensional linear Schr\"odinger equations. 
The nonclassical symmetries of the equations from class~\eqref{EqGenLPE} were described 
in~\cite{Popovych1995,Popovych2006b}. 
We hope that a simultaneous application of tools 
from~\cite{Fushchych&Shtelen&Serov&Popovych1992,Popovych1995,Popovych2006b,Saccomandi1997} and this paper will
allow us to investigate the potential nonclassical symmetries in the class~\eqref{EqGenLPE}.

\subsection*{Acknowledgements}
Research of NMI was partially supported by the Erwin Schr\"odinger Institute for Mathematical Physics
(Vienna, Austria) in form of a Junior Fellowship. MK was supported by START-project Y237 of the Austrian
Science Fund. The research of ROP was supported by the Austrian Science Fund (FWF), Lise Meitner project M923-N13.
The authors thank the referees for several helpful remarks.

\end{document}